\documentclass[journal,12pt,onecolumn]{IEEEtran}

\usepackage{booktabs}
\usepackage{comment}
\usepackage{xcolor}
\usepackage{url}

\usepackage[inline]{enumitem}
\usepackage{color}
\usepackage{mathtools}
\usepackage{mathrsfs}
\usepackage{caption}
\usepackage{dsfont}

\usepackage{subcaption}

\definecolor{linkblue}{rgb}{0,0,0.7}
\definecolor{citeblue}{rgb}{0,0.3,0.5}
\definecolor{lightgrey}{gray}{0.8}

\usepackage[colorlinks = true, linkcolor = linkblue, urlcolor=blue,
citecolor = citeblue]{hyperref}

\usepackage{amsmath,amssymb,amsthm}
\usepackage{graphicx}
\usepackage{aliascnt}

\def\ls{$(\mathrm{LS}_{\tau})$}
\def\qp{$(\mathrm{QP}_{\lambda})$}
\def\bp{$(\mathrm{BP}_{\sigma})$}
\def\bq{$(\mathrm{BQ}_{\sigma})$}

\def\gls{$(\mathrm{LS}_{\tau, \mathcal{K}})$}
\def\gbp{$(\mathrm{BP}_{\sigma, \mathcal{K}})$}
\def\gqp{$(\mathrm{QP}_{\lambda, \mathcal{K}})$}

\DeclareMathOperator{\1}{\mathds{1}}

\DeclareMathOperator*{\argmin}{\arg\min}

\DeclareMathOperator{\cone}{\mathrm{cone}}
\DeclareMathOperator{\cvx}{\mathrm{cvx}}

\DeclareMathOperator{\E}{\mathbb{E}}
\DeclareMathOperator{\Null}{\mathrm{null}}
\DeclareMathOperator{\Rng}{\mathrm{range}}
\DeclareMathOperator{\sgn}{\mathrm{sgn}}
\DeclareMathOperator{\Span}{\mathrm{span}}
\DeclareMathOperator{\w}{w}
\DeclareMathOperator{\rad}{rad}

\DeclareMathOperator{\supp}{supp}

\newcommand{\iid}{\ensuremath{\overset{\text{iid}}{\sim}}}

\newcommand{\reals}{\ensuremath{\mathbb{R}}}
\newcommand{\nats}{\ensuremath{\mathbb{N}}}
\newcommand{\ints}{\ensuremath{\mathbb{Z}}}

\newcommand{\sph}{\ensuremath{\mathbb{S}}}

\newcommand{\ip}[1]{\ensuremath{\langle #1\rangle}}

\newcommand{\ie}{i.e.,~}
\newcommand{\eg}{e.g.,~}

\usepackage{xparse}

\let\oldeqref\eqref
\makeatletter
\RenewDocumentCommand\eqref{s m}{%
  \IfBooleanTF#1%
  {\textup{\tagform@{\ref*{#2}}}}
  {\oldeqref{#2}}
}
\makeatother

\makeatletter
\renewcommand*\sectionautorefname{\S\@gobble}
\makeatother

\theoremstyle{plain}
\newtheorem{thm}{Theorem}[section]
\newaliascnt{prop}{thm}
\newtheorem{prop}[prop]{Proposition}
\aliascntresetthe{prop}
\newaliascnt{coro}{thm}
\newtheorem{coro}[coro]{Corollary}
\aliascntresetthe{coro}
\newaliascnt{lem}{thm}
\newtheorem{lem}[lem]{Lemma}
\aliascntresetthe{lem}
\theoremstyle{definition}
\newtheorem{defn}{Definition}[section]

\theoremstyle{remark}
\newtheorem{fact}{Fact}[section]
\newtheorem{rmk}[fact]{Remark}
\newtheorem*{rmk*}{Remark}
\newtheorem*{fact*}{Fact}

\newcommand{\vertiii}[1]{{\left\vert\kern-0.25ex\left\vert\kern-0.25ex\left\vert #1 
    \right\vert\kern-0.25ex\right\vert\kern-0.25ex\right\vert}}

\title{\Large On the best choice of \textsc{Lasso} program given data
  parameters}

\author{Aaron~Berk, Yaniv~Plan, and~Ozgur~Yilmaz\thanks{This work was supported
    by the Natural Sciences and Engineering Research Council of Canada (NSERC)
    [CGSD3-489677 to A.B., 22R23068 to Y.P., 22R82411 to O.Y., 22R68054 to
    O.Y.]; and the Pacific Institute for the Mathematical Sciences (PIMS) [CRG
    33: HDDA to Y.P., CRG 33: HDDA to O.Y.].}
  \thanks{A.\ Berk, Y.\ Plan \& O.\ Yilmaz are with the Dept.\ Mathematics,
    University of British Columbia, Vancouver, BC, Canada (email:~aberk@math.ubc.ca)
    }}

\begin{document}
\setlength{\parskip}{0pt plus .5pt minus .1pt}
\setlength{\abovedisplayskip}{5pt plus .2pt minus 2pt}
\setlength{\belowdisplayskip}{5pt plus .2pt minus 2pt}
\setlength{\abovedisplayshortskip}{3pt plus 1.2pt minus 1pt}
\setlength{\belowdisplayshortskip}{3pt plus 1.2pt minus 1pt}

\maketitle

\begin{abstract}
  Generalized compressed sensing (GCS) is a paradigm in which a structured
  high-dimensional signal may be recovered from random, under-determined, and
  corrupted linear measurements. Generalized \textsc{Lasso} (GL) programs are
  effective for solving GCS problems due to their proven ability to leverage
  underlying signal structure. Three popular GL programs are equivalent in a
  sense and sometimes used interchangeably. Tuned by a governing parameter, each
  admit an optimal parameter choice. For sparse or low-rank signal structures,
  this choice yields minimax order-optimal error. While GCS is well-studied,
  existing theory for GL programs typically concerns this optimally tuned
  setting. However, the optimal parameter value for a GL program depends on
  properties of the data, and is typically unknown in practical
  settings. Performance in empirical problems thus hinges on a program's
  parameter sensitivity: it is desirable that small variation about the optimal
  parameter choice begets small variation about the optimal risk. We examine the
  risk for these three programs and demonstrate that their parameter sensitivity
  can differ for the same data. We prove a \emph{gauge-constrained} GL program
  admits asymptotic cusp-like behaviour of its risk in the limiting low-noise
  regime. We prove that a \emph{residual-constrained} \textsc{Lasso} program has
  asymptotically suboptimal risk for very sparse vectors. These results contrast
  observations about an \emph{unconstrained} \textsc{Lasso} program, which is
  relatively less sensitive to its parameter choice. We support the asymptotic
  theory with numerical simulations, demonstrating that parameter sensitivity of
  GL programs is readily observed for even modest dimensional
  parameters. Importantly, these simulations demonstrate regimes in which a GL
  program exhibits sensitivity to its parameter choice, though the other two do
  not. We hope this work aids practitioners in selecting a GL program for their
  problem.
\end{abstract}

\begin{IEEEkeywords}
  Parameter sensitivity, \textsc{Lasso}, Compressed sensing, Convex
  optimization, Minimax risk
\end{IEEEkeywords}

\IEEEpeerreviewmaketitle

\section{Introduction\label{sec:introduction}}

\IEEEPARstart{U}{nder-determined} inverse problems are of fundamental importance
to modern mathematical and machine learning applications. In these problems, one
aims to recover or approximate a ground truth signal $x_{0} \in \reals^{N}$ from
noisy measurements $y \in \reals^{m}$ when $m \ll N$. The paradigm of
generalized compressed sensing (GCS) further specifies that the signal $x_{0}$
be well characterized by some structural proxy, known \emph{a priori}, and that
the measurement process be linear: $y = Ax_{0} + \eta z$, where
$A \in \reals^{m \times N}$ is typically \emph{random}, and $z$ is either a
random or deterministic corruption with noise scale $\eta > 0$.

For example, in MR imaging, one may wish to recover the wavelet coefficients of
an image by randomly subsampling its Fourier
coefficeints~\cite{lustig2007sparse, lustig2008compressed}. In geophysics, one
may wish to determine a region's bathymetry from a small number of radar
measurements taken at the surface~\cite{kumar2015source}, or to obtain a
subsurface image using a small number of
geophones~\cite{herrmann2012fighting}. Recent investigations suggest how
well-analyzed approaches for solving inverse problems may help to elucidate
``mysterious'' behaviours of high-dimensional non-linear function
approximators~\cite{hastie2019surprises, mei2019generalization}. Moreover,
compressed sensing theory may be used to prove recovery guarantees for certain
neural network architectures and particular data regimes~\cite{hand2019global}.

The now classical CS result~\cite{foucart2013mathematical} shows that when
$x_{0} \in \reals^{N}$ is an $s$-sparse signal, $m \geq Cs \log (e N /s)$
measurements suffice to efficiently recover $x_{0}$ from $(y, A)$ \emph{with
  high probability} on the realization of $A$. The \textsc{Lasso} is a common
and well-analyzed tool for effecting the recovery of
$x_{0}$~\cite{bickel2009simultaneous, chen2001atomic, shaobing1994basis,
  foucart2013mathematical, stojnic2013framework, tibshirani1996regression,
  van2008probing}. Currently, ``\textsc{Lasso}'' is an umbrella term referring
to three or more different \textsc{Lasso} programs, though it originally
referred to the $\ell_{1}$-constrained program {\ls}, which is defined
below~\cite{tibshirani1996regression}. Effective for its ability to perform
simultaneous best-basis and subset selection~\cite{tibshirani1996regression},
the \textsc{Lasso} is a convex optimization approach that has several variants
and cousins~\cite{bickel2009simultaneous, oymak2013squared,
  thrampoulidis2018precise, van2008probing}.

Of particular interest to this work, we introduce three common \textsc{Lasso}
programs and their solutions:
\begin{align}
  \label{eq:ls}
  \hat x(\tau) %
  &\in \argmin_{x\in \reals^{N}} \big\{ \|y - Ax\|_{2}^{2} : \|x\|_{1} \leq \tau \big\} %
    \tag{$\mathrm{LS}_{\tau}$}
  \\
  \label{eq:bp}
  \tilde x(\sigma) %
  &\in \argmin_{x\in \reals^{N}} \big\{ \|x\|_{1} : \|y - Ax\|_{2}^{2} \leq \sigma^{2} \big\} %
    \tag{$\mathrm{BP}_{\sigma}$}
  \\
  \label{eq:qp}
  x^{\sharp}(\lambda) %
  &\in \argmin_{x\in \reals^{N}} \big\{ \frac12 \|y - Ax\|_{2}^{2} + \lambda \|x\|_{1} \big\}. %
    \tag{$\mathrm{QP}_{\lambda}$}
\end{align}
Some naming ambiguity for these programs exists in the literature. Here, we
refer to {\ls} as constrained \textsc{Lasso}, {\qp} as unconstrained
\textsc{Lasso}, and {\bp} as [quadratically constrained] basis pursuit;
solutions for each program are denoted, respectively, by $\hat x(\tau)$,
$\tilde x(\sigma)$ and $x^{\sharp}(\lambda)$. Our notation and naming convention
for these three programs is similar to that used in~\cite{van2008probing}.

Generalizations of these programs, commonly referred to as generalized
\textsc{Lasso}, allow for the recovery of signals with other kinds of structure
that are well modelled by convex proxy sets. To introduce the generalized
\textsc{Lasso} programs, first let
$\emptyset \neq \mathcal{K} \subseteq \reals^{N}$ be a convex set and denote by
$\|\cdot \|_{\mathcal{K}}$ the Minkowski functional of $\mathcal{K}$ (\ie gauge). For
$\sigma, \tau, \lambda > 0$, the following \emph{generalized} \textsc{Lasso}
programs, which are convex, are defined by:
\begin{align}
  \label{eq:lsK}
  \hat x(\tau; A, y, \mathcal{K}) %
  &\in \argmin_{x\in \reals^{N}} \big\{ \|y - Ax\|_{2}^{2} : %
    x \in \tau \mathcal{K} \big\} %
    \tag{$\mathrm{LS}_{\tau, \mathcal{K}}$}
  \\
  \label{eq:bpK}
  \tilde x(\sigma; A, y, \mathcal{K}) %
  &\in \argmin_{x\in \reals^{N}} \big\{ \|x\|_{\mathcal{K}} : \|y - Ax\|_{2}^{2} \leq \sigma^{2} \big\} %
    \tag{$\mathrm{BP}_{\sigma, \mathcal{K}}$}
  \\
  \label{eq:qpK}
  x^{\sharp}(\lambda; A, y, \mathcal{K}) %
  &\in \argmin_{x\in \reals^{N}} \big\{ \frac12 \|y - Ax\|_{2}^{2} + \lambda \|x\|_{\mathcal{K}} \big\}. %
    \tag{$\mathrm{QP}_{\lambda, \mathcal{K}}$}
\end{align}
In the standard CS setting, the gauge is the $\ell_{1}$-norm, though $x_{0}$ is
assumed to belong to the set of $s$-sparse vectors
$\Sigma_{s}^{N} := \{ x \in \reals^{N} : |\supp(x)| \leq s\}$. So,
$x_{0}$ does not necessarily belong to the convex proxy set $\mathcal{K} = B_{1}^{N}$,
where $B_{1}^{N}$ denotes the $N$-dimensional unit $1$-norm ball. In particular,
$B_{1}^{N}$ itself serves as a convex proxy set for sparse vectors in the sense
that if $x\in \reals^{N}$ is $s$-sparse, then $\|x\|_{1} / \|x\|_{2}$ is small
relative to non-sparse vectors.

It is worth mentioning a brief note on uniqueness. When $A$ is under-determined
and has a suitable randomness, it is straightforward to show the programs {\bp}
and {\qp} admit unique solutions almost surely on the realization of $A$. A
detailed exposition for {\qp} is given in~\cite{tibshirani2013lasso}. For a
sufficient condition on $A$ giving uniqueness of {\bp},
see~\cite{zhang2015necessary}. However, {\ls} does not always admit a unique
solution. For instance, if $\tau$ is ``too large'', then there may be infinitely
many solutions $x \in \tau B_{1}^{N}$ satisfying $\|y - Ax \|_{2} = 0$. This
fact is fundamental to one of our results in~\autoref{sec:LS-instability}. By
mild abuse of notation, when the solution to a program is unique we will replace
``$\in$'' with ``$=$'' in the definitions of the solutions for each
program. Otherwise, we define each of $\hat x(\tau), \tilde x(\sigma)$, and
$x^{\sharp}(\lambda)$ as the solution yielding worst-case error, and which
appears first when ordered lexicographically. For example, $\hat x(\tau)$ refers
to the particular solution solving {\ls} such that
$\|\hat x(\tau) - x_{0}\|_{2} \geq \|\hat x - x_{0}\|_{2}$ for any other
$\hat x$ solving {\ls}. We make an analogous modification to the definitions of
the solutions to the generalized \textsc{Lasso} programs.

To relate the recovery performance of each program, we compare their recovery
errors. While there are several possibilities for measuring the recovery error,
the expected squared error and noise-normalized expected squared error of the
estimator are common when the noise, $z$, is random~\cite{oymak2013squared}. In
this work, we'll define the loss for an estimator as the noise normalized
squared error of that estimator (with respect to the ground truth signal
$x_{0}$); and define the estimator's risk as the expectation of the loss with
respect to $z$. Note that the risk and loss are functions of the random matrix
$A$. Specifically, the loss is defined for {\ls}, {\bp}, {\qp} respectively by:
\begin{align*}
  \hat{L}(\tau; x_{0}, A, \eta z) %
  &:= \eta^{-2} \|\hat x(\tau) - x_{0} \|_{2}^{2},
  \\
  \tilde L(\sigma; x_{0}, A, \eta z) %
  &:= \eta^{-2} \|\tilde x(\sigma) - x_{0} \|_{2}^{2},
  \\
  L^{\sharp}(\lambda; x_{0}, A, \eta z) %
  &:= \eta^{-2} \|\hat x(\eta\lambda) - x_{0} \|_{2}^{2},
\end{align*}
and the risk by:
\begin{align*}
  \hat R (\tau; x_{0}, A, \eta) %
  &:= \E_{z} \hat{L}(\tau; x_{0}, A, \eta z) %
  \\\tilde R (\sigma; x_{0}, A, \eta) %
  &:= \E_{z} \tilde L(\sigma; x_{0}, A, \eta z) %
  \\R^{\sharp} (\lambda; x_{0}, A, \eta) %
  &:= \E_{z} L^{\sharp}(\lambda; x_{0}, A, \eta z). %
\end{align*}

  Minimax order-optimal error rates are well-known for
  $\hat x(\tau; y, A, \mathcal{K})$ when $\tau$ is equal to the optimal
  parameter choice, $A$ is a matrix whose rows are independent, isotropic
  subgaussian random vectors, and $\mathcal{K}$ is a symmetric, closed convex
  set containing the origin~\cite{foucart2013mathematical, liaw2017simple,
    oymak2013squared}. A kind of equivalence between the three estimators
  (\emph{cf}. \autoref{prop:foucart-program-equivalence}) allows, in kind, for
  the characterization of the error rates for $\tilde x(\sigma)$ and
  $x^{\sharp}(\lambda)$ when $\sigma$ and $\lambda$ are optimally
  tuned. However, the error of $\hat x(\tau; y, A, \mathcal{K})$ is not fully
  characterized in the setting where $\tau$ is not the optimal
  choice. Similarly, the programs {\ls}, {\bp} and {\qp} are often referred to
  interchangeably, but a full comparison of the error of the three estimators
  $\hat x(\tau), \tilde x(\sigma)$, and $x^{\sharp}(\lambda)$, as a function of
  their governing parameters, is lacking. It is an open question if there are
  settings in which one estimator is always preferable to another.

  Understanding the sensitivity of a \textsc{Lasso} program to its parameter
  choice is crucial. While theoretical guarantees for recovery error are
  typically given for an oracular choice of the parameter, the optimal parameter
  setting is generally unknown in practice. Thus, the usefulness of theoretical
  recovery guarantees may hinge on the assumption that the recovery error is
  stable with respect to variation of the governing parameter. In particular,
  one may hope that small changes in the governing parameter beget no more than
  small changes in the risk or loss. 

  We take a step toward characterizing the performance and sensitivity of the
  three programs introduced by examining particular asymptotic parameter regimes
  for each program. We do this by extending results of~\cite{berk2020sl1mpc,
    berk2019pdparmsens} from the setting where $A$ is identity to the setting
  where $A$ is a matrix whose rows are independent, isotropic subgaussian random
  vectors. In this setting, we prove the existence of regimes in which CS
  programs exhibit sensitivity to their parameter choice: small changes in
  parameter values can lead to blow-up in risk. Despite the notion of
  equivalence hinted at above, we demonstrate regimes in which one program
  exhibits sensitivity, while the other two do not. For example, in the very
  sparse regime, our theory and simulations suggest not to use {\bp}. In the
  low-noise regime, they suggest not to use {\ls}. Assuredly, we identify
  situations where CS programs perform well in theory and \emph{in silico}
  alike. We hope that the asymptotic theory, coupled with fairly extensive
  numerical simulations, aid practitioners in deciding which CS program to
  select.

\section{Summary of results to follow\label{sec:summary-results}}

As a way of alluding to the main results to follow, we start by describing three
sibling results. We intend for them to contrast the behaviour of the three
$\ell_{1}$ programs that are the main focus of this work. Define the worst-case
risk for {\ls} in the low-noise regime by:
\begin{align*}
  R^{*}(s, A) %
  :=  \lim_{\eta \to 0} \sup_{x \in \Sigma_{s}^{N} \cap \partial B_{1}^{N}} %
  \hat R (1; x, A, \eta)
\end{align*}
Importantly, under mild assumptions, $R^{*}(s, A)$ is nearly equivalent to the
optimally tuned worst-case risk for {\ls}. Namely, for all $\eta > 0$,
\begin{align*}
  R^{*}(s, A) %
  \leq \sup_{x \in \Sigma_{s}^{N}} \hat R(\|x\|_{1}; x, A, \eta)
  \leq C R^{*}(s, A).
\end{align*}
This result is treated formally in \autoref{sec:rhat-nearly-monotone}. In each
case discussed below, the performance of the estimators will be compared to
$R^{*}(s, A)$ as a benchmark, noting that this quantity is minimax order optimal
in the sense of \autoref{prop:Rstar-minimax-optimal}.

In~\autoref{sec:LS-instability}, we show that {\ls} exhibits an asymptotic
instability in the low-noise regime. There is exactly one value $\tau^{*}$ of
the governing parameter yielding minimax order-optimal error, with any choice
$\tau \neq \tau^{*}$ yielding markedly worse behaviour. This result holds for
normalized $K$-subgaussian matrices $A$, which are defined
in~\autoref{def:isgrm}. The intuition provided by this result is that {\ls} is
extremely sensitive to the value of $\tau$ in the low-noise regime, making
empirical use of {\ls} woefully unstable in this regime.

\begin{thm}[{\ls} instability simplified]
  Let $1 \leq s \leq m < N < \infty$ be integers. If $A \in \reals^{m \times N}$
  is a normalized $K$-subgaussian matrix with
  $m > C_{\varepsilon} \delta^{-2} K^{2} \log (K) s \log
  \left(\tfrac{N}{s}\right)$, then with probability at least $1 - \varepsilon$ on
  the realization of $A$,
  \begin{align*}
     \lim_{\eta \to 0} \sup_{x \in \Sigma_{s}^{N} \cap B_{1}^{N}} %
    \frac{\hat R(\tau; x, A, \eta)}{R^{*}(s, A)} =
    \begin{cases}
      1 & \tau = 1\\
      \infty & \text{otherwise}
    \end{cases}
  \end{align*}

\end{thm}

Next, in~\autoref{sec:qp} we state a rephrasing
of~\cite[Theorem~3]{shen2015stable}. The result shows there is a parameter
$\lambda^{*}$ such that {\qp} is not sensitive to its parameter choice for
$\lambda \geq \lambda^{*}$. Right-sided parameter stability of {\qp} was first
established in~\cite[Theorem~7.2]{bickel2009simultaneous}. This well-known
result is contrasted in~\autoref{sec:numerical-results} with numerical results
demonstrating a left-sided parameter instability for {\qp} in the regime of high
sparsity, low noise, and large dimension.


\begin{thm}[{\qp} right-sided stability]
  Let $A \in \reals^{m \times N}$ be a normalized $K$-subgaussian matrix for
  $1 \leq m < N < \infty$. There is an absolute constant $C> 0$ such that if
  $\lambda \geq C \sqrt{\log N}$ and
  $m \geq C_{\varepsilon} \delta^{-2} K^{2} \log(K) s \log
  \left(\tfrac{N}{s}\right)$, then with probability at least $1 - \varepsilon$
  on the realization of $A$,
  \begin{align*}
    R^{\sharp}(\lambda; x_{0}, A, \eta) \leq C \lambda^{2} s.
  \end{align*}
\end{thm}

In the initial work, $\lambda^{*} = 2 \sqrt{2\log N}$
\cite[Theorem~7.2]{bickel2009simultaneous}. In the present phrasing, we write
only that $\lambda^{*} = C\sqrt{\log N}$ for an absolute constant $C > 0$. Note
that when the data are Gaussian, right-sided stability of {\qp} has been
examined in~\cite{thrampoulidis2015asymptotically}; of {\gqp},
in~\cite{thrampoulidis2018precise}.

Finally, in~\autoref{sec:analysis-bp} we show that {\bp} is poorly behaved for
all $\sigma > 0$ when $x_{0}$ is very sparse. In particular, under mild
restrictions on the aspect ratio of the measurement matrix, we show that
$\tilde R(\sigma; x_{0}, N, \eta)$ is asymptotically suboptimal for \emph{any}
$\sigma > 0$ when $s / N$ is sufficiently small. Below, this theorem shows that
the minimax risk for {\bp}, relative to the benchmark risk $R^{*}$, converges in
probability to $\infty$.

\begin{thm}[{\bp} instability simplified]
  \label{thm:bp-instability-simplified}
  Fix $\eta > 0$, an integer $s \geq 1$, and suppose for $m : \nats \to \nats$
  that $m (N)/ N \to \gamma \in (0, 1)$. For each $N$, suppose
  $A = A(N) \in \reals^{m(N) \times N}$ is a normalized $K$-subgaussian
  matrix. Then, for all $M > 0$,
  \begin{align*}
    \lim_{N\to \infty} \mathbb{P}\left(%
    \inf_{\sigma > 0} \sup_{x \in \Sigma_{s}^{N}}     %
    \frac{ \tilde R(\sigma; x, A, \eta)}{R^{*}(s,A)} > M \right) %
    = 1.
  \end{align*}
\end{thm}

Numerical results supporting~\autoref{sec:LS-instability}
and~\autoref{sec:analysis-bp} are discussed
in~\autoref{sec:numerical-results}. Proofs of most of the theoretical results
are deferred to~\autoref{sec:proofs}. Next, we add two clarifications. First,
the three programs are equivalent in a sense.

\begin{prop}[Program equivalence {\cite[Proposition 3.2]{foucart2013mathematical}}]
  \label{prop:foucart-program-equivalence}
  Let $0 \neq x_{0} \in \reals^{N}$ and $\lambda > 0$. Where
  $x^{\sharp}(\lambda)$ solves {\qp}, define
  $\tau := \|x^{\sharp}(\lambda)\|_{1}$ and
  $\sigma := \|y - Ax^{\sharp}(\lambda)\|_{2}$. Then $x^{\sharp}(\lambda)$
  solves {\ls} and {\bp}.
\end{prop}

However, $\tau$ and $\sigma$ are functions of $z$, a random variable, and this
mapping may not be smooth. Thus, parameter stability of one program is not
implied by that of another. Second, $R^{*}(s, A)$ has the desirable property
that it is computable up to multiplicative constants~\cite{liaw2017simple}.

\begin{prop}[Risk equivalences]
  \label{prop:Rstar-minimax-optimal}
  
  Fix $\delta, \varepsilon > 0$, let $1 \leq s \leq m < \infty, N \geq 2$ be
  integers, let $\eta > 0$. Suppose $A \in \reals^{m \times N}$ is a normalized
  $K$-subgaussian matrix satisfying
  $m > C_{\varepsilon} \delta^{-2} K^{2} \log K s\log(N/s)$, and suppose that
  $y = A x_{0} + \eta z$ for $z \in \reals^{m}$ with
  $z_{i} \iid \mathcal{N}(0, 1)$. Let
  $M^{*}(s, N) :=  \inf_{x^{*}} \sup_{x_{0} \in \Sigma_{s}^{N}}
  \eta^{-2} \|x_{*} - x_{0}\|_{2}^{2}$ be the minimax risk over arbitrary
  estimators $x^{*} = x^{*}(y)$. With probability at least $1 - \varepsilon$ on
  the realization of $A$, there is $c, C_{1}, C_{\delta} > 0$ such that
  \begin{align*}
    cs \log (N /s ) %
    & \leq M^{*}(s, N) %
    \leq  \inf_{\lambda > 0} \sup_{x_{0} \in \Sigma_{s}^{N}} %
      R^{\sharp}(\lambda; x_{0}, A, \eta) %
    \\
    & \leq C_{\delta} R^{*}(s, A) %
    \leq C_{\delta} s \log (N /s). 
  \end{align*}

\end{prop}

In this work, we focus primarily on the versions of \textsc{Lasso} for which
$\|\cdot\|_{1}$ is the structural proxy, and $\Sigma_{s}^{N}$ the structure set
for the data $x_{0}$. In addition, we discuss the pertinence of our results to
the Generalized \textsc{Lasso} setting. For instance,
in~\autoref{lem:nuc-norm-recovery}, we show how~\autoref{thm:ls-instability}
adapts to setting of parameter sensitivity for low-rank matrix recovery using
nuclear norm. Further, we connect our discussion on {\qp} in~\autoref{sec:qp} to
results for more general gauges~\cite{thrampoulidis2015asymptotically,
  thrampoulidis2018precise}, which works have developed tools suitable for
analyzing parameter sensitivity of {\gqp} when the data are Gaussian. While it
remains an open question to determine how our results
in~\autoref{sec:analysis-bp} may be extended to analyze parameter sensitivity of
{\gbp}, we conjecture that a suboptimality result like that exemplified
in~\autoref{thm:bp-instability-simplified} exists under analogous assumptions
for {\gbp}.

\section{Related Work}
\label{sec:related-work}

Several versions of the \textsc{Lasso} program are well-studied in the context
of solving CS problems~\cite{foucart2013mathematical}. The program {\ls} was
first posed in~\cite{tibshirani1996regression}. An analysis of its risk when
$\tau = \|x_{0}\|_{1}$ and the noise $z$ is deterministic may be found
in~\cite{foucart2013mathematical}. A sharp non-asymptotic analysis for the
generalized constrained \textsc{Lasso} may be found
in~\cite{oymak2013squared}. There, the risk was shown to depend on specific
geometric properties of the regularizer. When the measurement matrix has
independent isotropic subgassian rows, it has been demonstrated how a geometric
quantity may unify the quantification of generalized constrained \textsc{Lasso}
risk~\cite{liaw2017simple}. Risk bounds for generalized constrained
\textsc{Lasso} with nonlinear observations were characterized
in~\cite{plan2016generalized}. Recent work has shown how dimensional parameters
governing signal recovery problems in ridgeless least squares regression affect
the average out-of-sample risk in some settings~\cite{hastie2019surprises,
  mei2019generalization}.

Non-asymptotic bounds for the unconstrained \textsc{Lasso} were developed
in~\cite{bickel2009simultaneous}, which also determines an order-optimal choice
for the program's governing parameter. The asymptotic risk for the unconstrained
\textsc{Lasso} is determined analytically in~\cite{bayati2011dynamics,
  bayati2012lasso}. Sharp, non-asymptotic risk bounds for the generalized
unconstrained \textsc{Lasso} are developed
in~\cite{thrampoulidis2015regularized,
  thrampoulidis2015asymptotically}. In~\cite{thrampoulidis2015asymptotically},
$R^{\sharp}(\lambda)$ is examined for $\lambda$ about $\lambda_{\mathrm{opt}}$,
while~\cite{thrampoulidis2018precise} examines the risk as a function of its
governing parameter for other kinds of $M$-estimators. Both assume Gaussianity
of the data, and neither considers sensitivity with respect to parameter choice.

Basis pursuit is a third popular phrasing of the \textsc{Lasso} program, first
proposed in~\cite{shaobing1994basis, chen2001atomic}. For a theoretical
treatment of basis pursuit, we refer to~\cite{foucart2013mathematical}. Analytic
connections between basis pursuit and other \textsc{Lasso} programs are
exploited for fast computation of solutions in~\cite{van2008probing}.

Other modifications of the standard \textsc{Lasso} have also been examined. For
example, sharp non-asymptotic risk bounds for the so-called square-root
\textsc{Lasso} were obtained in~\cite{oymak2013squared}. Related to basis
pursuit, instance optimality of an exact $\ell_{1}$ decoder is analyzed
in~\cite{wojtaszczyk2010stability}.

Sensitivity to parameter choice was analyzed for three proximal denoising (PD)
programs that are analogues of the ones considered in this
work~\cite{berk2020sl1mpc, berk2019pdparmsens}. PD is a simplification of CS, in
which $A$ is the identity matrix. There, the authors prove an asymptotic
cusp-like behaviour for constrained PD risk in the low-noise regime, an
asymptotic phase transition for unconstrained PD risk in the low-noise regime,
and asymptotic suboptimality of the basis pursuit PD risk in the very sparse
regime. The current work develops non-trivial generalizations of the results
in~\cite{berk2020sl1mpc}, proving asymptotic results about the sensitivity of
$\ell_{1}$ minimization for the generalized constrained \textsc{Lasso}, and
generalized basis pursuit.

\section{Main theoretical tools}
\label{sec:main-theor-tools}

\subsection{Notation}
\label{sec:notation}

Let $\cvx (\mathcal{C})$ denote the convex hull of the set
$\mathcal{C}\subseteq \reals^{N}$:
\begin{align*}
  \cvx (\mathcal{C}) %
  & := \{ \sum_{j = 1}^{J} \alpha_{j} x_{j} : x_{j} \in \mathcal{C}, \alpha_{j} \geq 0, \sum \alpha_{j} = 1, J < \infty \} \\
  & = \bigcap_{\mathclap{\substack{\mathcal{C}' \supseteq \mathcal{C}\\\mathcal{C}' \text{is convex}}}} \mathcal{C'}.
\end{align*}
Let $\cone (\mathcal{C})$ denote the cone of $\mathcal{C}$:
\begin{align*}
  \cone(\mathcal{C}) := \{ \lambda x: \lambda \geq 0, x \in \mathcal{C}\}.
\end{align*}
Define the descent cone of a convex function $f : \reals^{N} \to \reals$ at a
point $x \in \reals^{N}$ by
\begin{align*}
  T_{f}(x) := \cone \{ z - x : z\in \reals^{N}, f(z) \leq f(x) \}.
\end{align*}
When $f = \|\cdot\|_{1}$ we write $T(x) := T_{\|\cdot\|_{1}}(x)$. By abuse of
notation, we write $T_{\mathcal{C}}(x) := T_{\|\cdot\|_{\mathcal{C}}} (x)$ to
refer to the descent cone of the gauge $\|\cdot\|_{\mathcal{C}}$ at a point $x$,
where $\mathcal{C} \subseteq \reals^{N}$ is a convex set.

Given $x \in \reals^{N}$, denote the $0$-norm of $x$ by
$\|x\|_{0} = \# \{ j \in [N] : x_{j} \neq 0\}$, where
$[N] := \{ 1, 2, \ldots, N\}$. Note that $\|\cdot \|_{0}$ is not a norm. For
$0 \leq p \leq \infty$, denote the $\ell_{p}$ ball by
$B_{p}^{N} := \{ x \in \reals^{N} : \|x\|_{p} \leq 1\}$. For $s, N \in \nats$
with $0 \leq s \leq N$, denote the set of at most $s$-sparse vectors by
\begin{align*}
  \Sigma_{s}^{N} := \{ x\in \reals^{N} : \|x\|_{0} \leq s\},
\end{align*}
and define $\Sigma_{-1}^{N} := \emptyset$.  Define the following sets:
\begin{align*}
  \mathcal{L}_{s}(r) %
  & := r\cdot \cvx (\Sigma_{s}^{N} \cap \sph^{N-1}),
  \\
    \mathcal{L}_{s} %
  & := \mathcal{L}_{s}(2),
  \\
    \mathcal{L}^{*}_{s} %
  & := \mathcal{L}_{2s}(4). %
\end{align*}
Additionally, define the sets:
\begin{align*}
  \mathcal{J}_{s}^{N} %
  &:= \left\{ x \in \reals^{N} : \|x\|_{1} \leq \sqrt s \|x\|_{2}\right\},
  \\
  \mathcal{K}_{s}^{N} %
  &:= \left\{ x \in \reals^{N} : \|x\|_{2} \leq 1 \, %
  \And \, \|x\|_{1} \leq \sqrt s\right\}. %
\end{align*}
Observe that $\mathcal{J}_{s}^{N}$ is a cone, and that
$\mathcal{K}_{s}^{N} = B_{2}^{N} \cap \sqrt s B_{1}^{N} =
\cvx(\mathcal{J}_{s}^{N} \cap \sph^{N-1})$.

\subsection{Tools from probability theory}
\label{sec:geometric-tools}

We start by introducing subgaussian random variables, which generalize Gaussian
random variables, but retain certain desirable properties, such as Gaussian-like
tail decay and moment bounds.

\begin{defn}[Subgaussian random variable]
  A random variable $X$ is called subgaussian if there exists a constant $K > 0$
  such that the moment generating function of $X^{2}$ satisfies, for all
  $\lambda$ such that $|\lambda| \leq K^{-1}$,
  \begin{align*}
    \E \exp(\lambda^{2} X^{2}) \leq \exp(K^{2} \lambda^{2}).
  \end{align*}
  The subgaussian norm of $X$ is defined by
  \begin{align*}
    \|X\|_{\psi_{2}} := \inf \{ t > 0 : \E \exp(X^{2} / t^{2}) \leq 2 \}. 
  \end{align*}
\end{defn}

We similarly define subexponential random variables.

\begin{defn}[Subexponential random variable]
  A random variable $X$ is called subexponential if there exists a constant
  $K > 0$ such that the moment generating function of $|X|$ satisfies, for all
  $\lambda$ such that $0 \leq \lambda \leq K^{-1}$,
  \begin{align*}
    \E \exp ( \lambda |X|) \leq \exp(K \lambda ).
  \end{align*}
  The subexponential norm of $X$ is defined by
  \begin{align*}
    \|X\|_{\psi_{1}} := \inf \{ t > 0 : \E \exp(|X| / t) \leq 2 \}. 
  \end{align*}
\end{defn}

Additionally, we call $X \in \reals^{N}$ a $K$-subgaussian random vector if
$\|X\|_{\psi_{2}} := \sup_{a \in \reals^{N}} \|\ip{a, X}\|_{\psi_{2}} \leq K$;
analogously so for $K$-subexponential random vectors. Where it is either clear
or irrelevant, we may omit observing the norm parameter and refer to a
$K$-subgaussian random vector simply as a subgaussian random vector; likewise
with a subexponential random vector. For properties and equivalent definitions
of subgaussian and subexponential random variables and vectors,
see~\cite[Chapter~2]{vershynin2018high}. Next, we introduce a piece of jargon
for the sake of concision.

\begin{defn}[$K$-subgaussian matrix]
  \label{def:isgrm}
  Given $m, N \in \nats$, call $A \in \reals^{m \times N}$ a $K$-subgaussian
  matrix if $A$ has rows $A_{i}^{T}$ that are independent, isotropic
  $K$-subgaussian random vectors:
  \begin{align*}
    \E A_{i} A_{i}^{T} = I, \qquad %
    \| A_{i}\|_{\Psi_{2}} \leq K, \qquad %
    i \in [m]. %
  \end{align*}
  Further, call $\tfrac{1}{\sqrt m} A$ a normalized $K$-subgaussian matrix.
\end{defn}

In this work, we crucially leverage the fact that $A$ satisfies a restricted
isometry property (RIP). An exposition on RIP and restricted isometry constants
may be found in~\cite{foucart2013mathematical}. As the results of this work
concern $K$-subgaussian matrices, we state a classical version of RIP for such
matrices restricted to the set of $s$-sparse vectors.




\begin{thm}[RIP for subgaussian matrices {\cite[Theorem 9.2]{foucart2013mathematical}}]
  \label{thm:rip-subgaus}
  Let $A \in \reals^{m \times N}$ be a normalized $K$-subgaussian matrix. There
  exists a constant $C = C_{K} > 0$ such that the restricted isometry constant
  of $A$ satisfies $\delta_{s} \leq \delta$ with probability at least
  $1 - \varepsilon$ provided
  \begin{align*}
    m \geq C \delta^{-2} ( s \ln (eN/s) + \ln (2 \varepsilon^{-1})).
  \end{align*}
\end{thm}

\begin{rmk}
  Setting $\varepsilon = 2 \exp( -\delta^{2} m / (2C))$ yields the condition
  \begin{align*}
    m \geq 2 C \delta^{-2} s \ln ( eN/s)
  \end{align*}
  which guarantees that $\delta_{s} \leq \delta$ with probability at least $1 - 2 \exp ( - \delta^{2} m / (2C))$.

\end{rmk}

A tool necessary to the development of the results in
\autoref{sec:LS-instability} and \autoref{sec:analysis-bp} (specifically
\autoref{prop:unif-noise-scale}, \ref{prop:gw-lb-2}, and \ref{prop:gw-ub-2})
characterizes the variance and tail decay of the supremum of a Gaussian
process. For an introduction to random processes, we refer the reader
to~\cite{vershynin2018high, adler2009random}.

\begin{thm}[Borell-TIS inequality {\cite[Theorem 2.1.1]{adler2009random}}]
  \label{thm:borell-tis}
  Let $T$ be a topological space and let $\{f_{t}\}_{t \in T}$ be a centred (\ie
  mean-zero) Gaussian process almost surely bounded on $T$ with
  \begin{align}
    \label{eq:borell-norm}
    \vertiii{f}_{T} %
    &:= \sup_{t \in T} f_{t}, %
    &
      \sigma_{T}^{2} %
    &:= \sup_{t\in T} \E \big[ f_{t}^{2} \big]
  \end{align}
  such that $\vertiii{f}_{T}$ is almost surely finite. Then $\E \vertiii{f}_{T}$ and
  $\sigma_{T}$ are both finite and for each $u > 0$,
  \begin{align*}
    \mathbb{P}\big( \vertiii{f}_{T} > \E \vertiii{f}_{T} + u\big) %
    \leq \exp\big( - \frac{u^{2}}{2 \sigma_{T}^{2}} \big).
  \end{align*}
  Observe that $\vertiii{f}_{T}$ is notation; $\vertiii{\cdot}_{T}$ is not a
  norm. By symmetry, one may derive an analogous lower-tail
  inequality. Consequently, one also has for each $u > 0$, 
  \begin{align*}
    \mathbb{P}\big( \left|\vertiii{f}_{T} - \E \vertiii{f}_{T}\right| > u\big) %
    \leq 2 \exp\big( - \frac{u^{2}}{2 \sigma_{T}^{2}} \big).
  \end{align*}

\end{thm}

\subsection{Geometric tools from probability}
\label{sec:geometric-tools-from}

Analysis of structured signals benefits from the ability to characterize their
\emph{effective dimension}. In this work, we capture this notion of effective
dimension with the Gaussian complexity and Gaussian width (gw), which
are closely related.

\begin{defn}[Gaussian complexity]
  Let $T \subseteq \reals^{N}$. Define the Gaussian complexity of $T$ by
  \begin{align*}
    \gamma (T) %
    := \E \sup_{x \in T} \left|\ip{x, g}\right|, %
    \qquad g \sim \mathcal{N}(0, I_{N}).
  \end{align*}
\end{defn}

\begin{defn}[Gaussian width]
  Let $T \subseteq \reals^{N}$. Define the Gaussian width of $T$, by:
  \begin{align*}
    \w(T) %
    := \E \sup_{x \in T} \ip{x, g}, %
    \qquad g \sim \mathcal{N}(0, I_{N}).
  \end{align*}
\end{defn}

\begin{rmk*}
  If $T$ is symmetric, then the gw of $T$ satisfies
  $\w(T) = \frac12 \E \sup_{x \in T - T}\ip{x, g}$.
\end{rmk*}

Next we state two results controlling the deviation of a $K$-subgaussian matrix
on a bounded set, which generalize the idea of RIP introduced
in~\autoref{thm:rip-subgaus}. These results were first proved
in~\cite{liaw2017simple}, and an improved dependence on the constant $K$ was
then obtained in~\cite{xiaowei2019taildep}. The results are stated using the
improved constant $\tilde K := K \sqrt{\log K}$; we refer the reader
to~\cite[Theorem 2.1]{xiaowei2019taildep} for further details.

\begin{thm}[{\cite[Theorem 1.1]{liaw2017simple}}]
  \label{thm:liaw-11}
  Let $A \in \reals^{m\times N}$ be a $K$-subgaussian matrix and
  $T \subseteq \reals^{N}$ bounded. Then
  \begin{align*}
    \E \sup_{x \in T} \left| \|Ax\|_{2} - \sqrt m \|x\|_{2} \right| %
    \leq C \tilde K \gamma (T).
  \end{align*}
\end{thm}

Another version of this result holds, where the deviation is instead controlled
by the gw and radius, rather than the Gaussian complexity.

\begin{thm}[{\cite[Theorem 1.4]{liaw2017simple}}]
  \label{thm:liaw-14}
  Let $A \in \reals^{m \times N}$ be a $K$-subgaussian matrix and
  $T \subseteq \reals^{N}$ bounded. For any $u \geq 0$ the event
  \begin{align}
    \label{eq:thm-1-4}
    \sup_{x\in T} \big|\|Ax\|_{2} &- \sqrt m \|x\|_{2}\big| %
    \nonumber\\
    &\leq C\tilde K \left[ \w(T) + u \cdot \rad(T)\right]
  \end{align}
  holds with probability at least $1 - 3\exp(-u^{2})$. Here,
  $\rad(T) := \sup_{x \in T}\|x\|_{2}$ denotes the radius of $T$.
\end{thm}

\begin{rmk}
  If $u \geq 1$ the bound in~\eqref{eq:thm-1-4} can be loosened to the
  following simpler one:
  \begin{align*}
    \sup_{x \in T} \left|\|Ax\|_{2} - \sqrt m \|x\|_{2}\right| %
    \leq C\tilde K u \gamma(T).
  \end{align*}
\end{rmk}

Setting $T := \sph^{N-1}$, and using the improved constant obtained
in~\cite[Theorem 2.1]{xiaowei2019taildep} gives the following corollary.

\begin{coro}[Largest singular value of $K$-subgaussian matrices]
  \label{coro:largest-subgaus-singval}
  Let $A \in \reals^{m \times N}$ be a $K$-subgaussian matrix. For all
  $t \geq 0$, with probability at least $1 - 3 \exp(-t^{2})$,
  \begin{align*}
    \left|\|A\| - \sqrt m\right| %
    \leq C \tilde K \left[\sqrt N + t\right].
  \end{align*}
\end{coro}






Finally, we state the following comparison inequality for two centred Gaussian
processes.

\begin{thm}[Sudakov-Fernique inequality {\cite[Theorem
    7.2.11]{vershynin2018high}}]
  \label{thm:sudakov-fernique}
  Let $(X_{t})_{t \in T}, (Y_{t})_{t \in T}$ be two mean-zero Gaussian
  processes. Assume that, for all $s,t \in T$, we have
  \begin{align*}
    \E (X_{t} - X_{s})^{2} \leq \E (Y_{t} - Y_{s})^{2}.
  \end{align*}
  Then
  \begin{align*}
    \E \sup_{t \in T} X_{t} \leq \E \sup_{t \in T} Y_{t}.
  \end{align*}
\end{thm}

\subsection{Geometric tools}
\label{sec:geometric-tools-1}

In this section, we introduce tools primarily relevant to obtaining recovery
bounds for compressed sensing in the classical setting where
$\mathcal{K} = B_{1}^{N}$. We start by recalling that sparse vectors have low
effective dimension, as does their difference set.

\begin{lem}[gw of the sparse signal set {\cite[Lemma 2.3]{plan2013robust}}]
  \label{lem:mean-width-sparse-signals}
  There exist absolute constants $c, C > 0$ such that
  \begin{align*}
    cs \log (2N /s) %
    & \leq \w^{2} \left((\Sigma_{s}^{N}\cap B_{2}^{N}) %
      - (\Sigma_{s}^{N} \cap B_{2}^{N})\right) %
    \\
    & \leq C s \log (2N /s)
  \end{align*}
  For possibly different absolute constants $c, C >0$, one also has
  \begin{align*}
    cs \log (2N /s) %
    \leq \w^{2} \left(\Sigma_{s}^{N}\cap B_{2}^{N}\right) %
    \leq C s \log (2N /s).
  \end{align*}
\end{lem}

In addition, we also recall that the descent cone of the $\ell_{1}$ ball has
comparable effective dimension to the set of sparse vectors. For example,
see~\cite[Proposition 9.24]{foucart2013mathematical} for a related result
showing
$\w^{2}\left(T_{B_{1}^{N}}(x) \cap \sph^{N-1}\right) \leq 2s \log(eN/s)$
when $x$ is $s$-sparse. To clarify the connection between the gw of $B_{1}^{N}$
and that of $\Sigma_{s}^{N}$, we recall the following result.

\begin{lem}[Convexification {\cite[Lemma 3.1]{plan2013one}}]
  \label{lem:convexification}
  One has
  \begin{align*}
    \cvx(\Sigma_{s}^{N} \cap B_{2}^{N}) %
    \subseteq \mathcal{K}_{s}^{N} %
    \subseteq 2 \cvx(\Sigma_{s}^{N} \cap B_{2}^{N}).
  \end{align*}

\end{lem}

Next, it will be useful to leverage the following equivalent characterization
for the $\ell_{1}$ descent cone. Recall that $\sgn(\alpha) := 1$ if
$\alpha > 0$, $\sgn(\alpha) = -1$ if $\alpha < 0$ and $\sgn(0) = 0$. 

\begin{lem}[Equivalent descent cone characterization]
  \label{lem:equivalent-descent-cone}
  Let $x \in \Sigma_{s}^{N}$ with non-empty support set
  $\mathcal{T} \subseteq [N]$ and define $\mathcal{C} := \|x\|_{1}
  B_{1}^{N}$. If
  $K(x) := \{ h \in \reals^{N} : \|h_{\mathcal{T}^{C}}\|_{1} \leq - \ip{\mathrm{sgn}(x),
    h}\}$, then $T_{\mathcal{C}}(x) = K(x)$.
\end{lem}

Finally, recall that \textsc{Lasso} solutions admit the following descent cone
condition.

\begin{lem}[Descent cone condition]
  \label{lem:descent-cone-condition}
  Let $x \in \Sigma_{s}^{N}$ have non-empty support set $T \subseteq N$. Suppose
  $y = A x + \eta z$ for $\eta > 0$, $z \in \reals^{m}$, and
  $A \in \reals^{m\times N}$. Let $\hat x$ solve {\ls} with $\tau =
  \|x\|_{1}$. Then $\|h\|_{1} \leq 2 \sqrt s \|h\|_{2}$, where $h = \hat x - x$.
\end{lem}

\begin{proof}[Proof of {\autoref{lem:descent-cone-condition}}]
  We use \autoref{lem:equivalent-descent-cone} above before applying
  Cauchy-Schwarz:
  \begin{align*}
    \|\hat x - x\|_{1} %
    & = \|h_{T}\|_{1} + \|h_{T^{C}}\|_{1} %
    \\
    & \leq \ip{\mathrm{sgn}(\hat x_{T} - x), h_{T}} - \ip{\mathrm{sgn}(x), h} %
    \\
    & \leq \|\mathrm{sgn}(\hat x_{T} - x) - \mathrm{sgn}(x)\|_{2} \|h\|_{2} %
    \\
    & \leq 2 \sqrt s \|h\|_{2}.
  \end{align*}
\end{proof}

\subsection{Refinements on bounds for gw}
\label{sec:refin-bounds-gauss}

Crucial to the results of \autoref{sec:analysis-bp} are two recent results
appearing in~\cite{bellec2019localized}. These results are a fine-tuning of
standard gw results for bounded convex polytopes.

\begin{prop}[{\cite[Proposition 1]{bellec2019localized}}]
  \label{prop:bellec1}
  Let $m \geq 1$ and $N \geq 2$. Let $T$ be the convex hull of $2N$ points in
  $\reals^{m}$ and assume $T\subseteq B_{2}^{m}$. Then for $\gamma \in (0, 1)$,
  \begin{align*}
    \w & (T \cap \gamma B_{2}^{m}) %
    \\
    & \leq \min \big\{ %
    4 \sqrt{ \max \big\{1, \log(8e N \gamma^{2}) \big\} }, %
    \gamma \sqrt{ \min\{m, 2N\}} \big\}
  \end{align*}

\end{prop}

The work also proves a lower bound on the gw of bounded convex polytopes.

\begin{prop}[{\cite{bellec2019localized}}]
  \label{prop:bellec2}
  Let $m \geq 1$ and $N \geq 2$. Let $\gamma \in (0, 1]$ and assume for
  simplicity that $s = 1/ \gamma^{2}$ is a positive integer such that
  $s \leq N /5$. Let $T$ be the convex hull of the $2N$ points
  $\{ \pm M_{1}, \ldots, \pm M_{N}\} \subseteq \sph^{m-1}$. Assume that for some
  real number $\kappa \in (0, 1)$ we have
  \begin{align*}
    \kappa \| \theta \|_{2} \leq \| M \theta \|_{2} \qquad %
    \text{for all $\theta \in \reals^{N}$ such that $\|\theta\|_{0} \leq 2s$},
  \end{align*}
  Then
  \begin{align*}
    \w (T \cap \gamma B_{2}^{m}) %
    \geq ( \sqrt 2 / 4) \kappa \sqrt{ \log ( N \gamma^{2} / 5)}.
  \end{align*}
\end{prop}

In particular, the above two results may be combined to obtain a bound on a
random polytope, obtained by considering the image of a (non-random) polytope
under a normalized $K$-subgaussian matrix; proved in~\autoref{sec:proofs-refin-bounds}.

\begin{coro}[Controlling random hulls]
  \label{coro:bellec-random-hulls}
  Fix $\delta, \varepsilon > 0$, $\gamma \in (0, 1]$ and let
  $A \in \reals^{m \times N}$ be a normalized $K$-subgaussian matrix. Assume for
  simplicity that $s = 1 / \gamma^{2} \in \nats$ with $s < N/5$ and let $T$
  denote the convex hull of the $2N$ points $\{\pm A^{j} : j \in [N]\}$. Assume
  $m > C_{\varepsilon} \delta^{-2} \tilde K^{2} s \log (2N / s)$. With
  probability at least $1 - \varepsilon$, for any $\alpha \in (0, (1-\delta))$,
  \begin{align*}
    (\sqrt 2 / 4)& (1 - \delta)^{2} \sqrt{
    \log \frac{N\alpha^{2}}{5(1-\delta)^{2}}} %
    \\
    & \leq \w(T \cap \alpha B_{2}^{m}) %
    \\
    & \leq \min \big\{%
    4 (1 + \delta) \sqrt{ \max \big\{1,
    \log \frac{8e N \alpha^{2}}{(1 + \delta)^{2}} \big\} }, %
    \\
    &\qquad \qquad \alpha \sqrt{ \min\{m, 2N\}} \big\}.
  \end{align*}

\end{coro}

\subsection{Projection lemma}
\label{sec:projection-lemma}

For $x \in \reals^{N}$ and $\mathcal{C} \subseteq \reals^{N}$ nonempty, denote
the distance of $x$ to $\mathcal{C}$ by
$\mathrm{dist}(x, \mathcal{C}) := \inf_{w \in \mathcal{C}} \|x - w\|_{2}$. If
$\mathcal{C}$ is a closed and convex set, there exists a unique point in
$\mathcal{C}$ attaining the infimum. We denote this point
\begin{align*}
  \mathrm{P}_{\mathcal{C}}(x) := \argmin_{w \in \mathcal{C}} \|x - w\|_{2}.
\end{align*}

The projection lemma was an important tool in~\cite{berk2020sl1mpc} for showing
parameter instability of two \textsc{Lasso} programs in the proximal denoising
setting. The result was first proved in~\cite[Lemma 15.3]{oymak2016sharp}, and a
simpler proof of a restated version given in~\cite{berk2020sl1mpc}. The result
shall play a critical role in this work, too. We include it here for
completeness.

\begin{lem}[Projection lemma {\cite[Lemma 3.2]{berk2020sl1mpc}}]
  \label{lem:projection-lemma}
  Let $\mathcal{K} \subseteq \reals^{m}$ be a non-empty closed and convex set with
  $0 \in \mathcal{K}$, and fix $\lambda \geq 1$. Then, for any $z \in \reals^{m}$,
  \begin{align*}
    \|\mathrm{P}_{\mathcal{K}}(z) \|_{2} \leq \|\mathrm{P}_{\lambda \mathcal{K}}(z) \|_{2}.
  \end{align*}
\end{lem}

\autoref{lem:projection-lemma} admits the following corollary, useful in proving
\autoref{lem:bp-oc-zero}. Its proof is deferred to
\autoref{sec:proofs-proj-lemma}.

\begin{coro}[{\cite[Corollary 3.1]{berk2020sl1mpc}}]
  \label{cor:projection-lemma}
  Let $\mathcal{K} \subseteq \reals^{m}$ be a non-empty closed and convex set
  with $0 \in \mathcal{K}$ and let $\|\cdot \|_{\mathcal{K}}$ be the gauge of
  $\mathcal{K}$. Given $y \in \reals^{m}$, define for $\alpha > 0$,
  \begin{align*}
    q_{\alpha} := \argmin \{ \|q\|_{\mathcal{K}} : \|q - y\|_{2} \leq \alpha \}.
  \end{align*}
  Then $\|q_{\alpha}\|_{2}$ is decreasing in $\alpha$.
\end{coro}

\section{{\ls} parameter instability}
\label{sec:LS-instability}

The main result of this section is proved in the case of standard CS, where
$x_{0} \in \Sigma_{s}^{N}$ and where tight bounds on the effective dimension of
the structure set are known (\eg bounds on $\gamma(\mathcal{L}_{s})$ or
$\gamma (\mathcal{K}_{s}^{N})$). Define $\tau^{*} := \|x_{0}\|_{1}$. The
following result states that $\hat L$ is almost surely suboptimal in the
limiting low-noise regime when $\tau \neq \tau^{*}$, while $\hat R(\tau^{*})$ is
order-optimal. A proof of the result may be found in
\autoref{sec:optimal-choice-tau}, with supporting lemmas in
\autoref{sec:proofs-constr-lasso}.

\begin{thm}[Asymptotic singularity]
  \label{thm:ls-instability}
  Fix $\delta, \varepsilon > 0$ and let $1 \leq s \leq m < N < \infty$ be
  integers.  Let $x_{0} \in \Sigma_{s}^{N}\setminus \Sigma_{s-1}^{N}$ with
  $\tau^{*} := \|x_{0}\|_{1}$ and $\tau > 0$ such that $\tau\neq \tau^{*}$. Let
  $\eta > 0$ and let $z \in \reals^{m}$ with $z_{i} \iid
  \mathcal{N}(0,1)$. Suppose $A \in \reals^{m\times N}$ is a normalized
  $K$-subgaussian matrix, and assume $m$ satisfies
  \begin{align*}
    m > C_{\varepsilon} \tilde K^{2} \delta^{-2} s \log \frac{eN}{s}.
  \end{align*}
  Almost surely on the realization of $(A, z)$,
  \begin{align*}
    \lim_{\eta \to 0} \hat L(\tau; x_{0}, A, \eta z) %
    = \infty. %
  \end{align*}
  With probability at least $1 - \varepsilon$ on the realization of $A$, there
  exist constants $0 < c_{\delta} < C_{\delta} < \infty$ such that
  \begin{align*}
    c_{\delta} s \log \frac{N}{s}
    \leq  \lim_{\eta \to 0} \sup_{x \in \Sigma_{s}^{N}}
    \hat R(\|x\|_{1}; x, A, \eta) %
    \leq C_{\delta} s \log \frac{N}{2s}.
  \end{align*}
\end{thm}

For clarity, observe the similarity to the definition of $\tau^{*}$ of the
precise parameter, $\|x\|_{1}$, appearing in the lower bound. Importantly, the
spirit of this result extends to the situation where $x_{0}$ belongs, more
generally, to some convex proxy set $\mathcal{K} \subseteq \reals^{N}$. In
particular, the blow-up of $\hat L$ in the limiting low-noise regime holds
independent of the assumptions on $\mathcal{K}$, except that $\mathcal{K}$ be
bounded. For instance,~\autoref{lem:nuc-norm-recovery} addresses the an
analogous result in the case where the signal is a $d\times d$ matrix and
$\mathcal{K}$ is the nuclear norm ball. Further, worst-case bounds on $\hat R$
are well-known in the case where $\mathcal{K}$ is a convex polytope, and are
useful when $\mathcal{K}$ has small gw~\cite{bellec2019localized,
  liaw2017simple}.

\section{Analysis of {\qp}}
\label{sec:qp}

\subsection{Right-sided parameter stability}
\label{sec:right-sided-param}

In this section we present a contrast to the type of sensitivity observed in
\autoref{sec:LS-instability}. Specifically, the result serves to demonstrate
that {\qp} is not sensitive to its parameter choice if the chosen parameter is
too large. This so-called right-sided parameter stability is important in
practical settings, as it suggests that recovery will not be penalized ``too
heavily'' if the parameter is chosen incorrectly to be too large. Having such a
leniency is reassuring, since knowing the exact choice of $\lambda$ in an
experimental setting is unlikely at best.

The right-sided parameter stability for {\qp} was first proved
in~\cite[Theorem~7.2]{bickel2009simultaneous}. When the data are Gaussian,
right-sided stability of {\qp} has been examined
in~\cite{thrampoulidis2015asymptotically}; of {\gqp},
in~\cite{thrampoulidis2018precise}. Here, we state a specialized rephrasing of a
version more suitably adapted to the present
work~\cite[Theorem~3]{shen2015stable}. This version is the same as that stated
in~\autoref{sec:summary-results}.

\begin{thm}[Specialized {\cite[Theorem~3]{shen2015stable}}]
  \label{thm:qp-r-stability}
  For integers $1 \leq s \leq m < N < \infty$, suppose
  $x_{0} \in \Sigma_{s}^{N}$ and let $A \in \reals^{m \times N}$ be a normalized
  $K$-subgaussian matrix. There is an absolute constant $C> 0$ such that if
  $\lambda \geq C \sqrt{\log N}$ and
  $m \geq C_{\varepsilon} \tilde K^{2} \delta^{-2} s \log \tfrac{N}{s}$, then
  with probability at least $1 - \varepsilon$ on the realization of $A$,
  \begin{align*}
    R^{\sharp}(\lambda; x_{0}, A, \eta) \leq C s \lambda^{2}.
  \end{align*}
\end{thm}
In particular, over-guessing $\lambda$ results in no more than a quadratic
penalty on the bound for the recovery error. Consequently, {\qp} is right-sided
parameter stable --- it is not sensitive to variation of its governing parameter
when the parameter is sufficiently large. We note, thereby, that there exist
regimes in which \textsc{Lasso} programs are not sensitive to their parameter
choice. Perhaps more importantly: there are regimes in which one program may be
sensitive to its parameter choice, and another program is not. Namely, we see
that {\ls} can be sensitive to its parameter choice in the low-noise regime (so
the correct choice of $\tau$ is an imperative), while recovery for the very same
data using {\qp} is not sensitive to $\lambda$ (if $\lambda$ is sufficiently
large enough).

\section{Analysis of {\bp}}
\label{sec:analysis-bp}

The final program that we subject to scrutiny is {\bp}. It is well-known under
standard assumptions that an optimal choice of $\sigma$ yields order-optimal
risk $\tilde R(\sigma^{*}; x_{0}, A, \eta)$ with high probability on the
realization of $A$. In this section, we demonstrate the existence of a regime in
which any choice of $\sigma$ fails to yield order-optimal recovery for {\bp}. A
key message of this section is that {\bp} performs poorly if the signal is too
sparse and the number of measurements is too large. We demonstrate this
behaviour for two regimes: the underconstrained setting, where $\sigma$ is ``too
large'' and the overconstrained setting where $\sigma$ is ``too small''. Each of
these settings covers the case where $\sigma$ is chosen ``just right''; we will
see how {\bp} risk fails to achieve order optimality in this case, as well.

For the duration of this section, we will consider $x_{0} \in \Sigma_{s}^{N}$
where $s$ may or may not be allowed to be $0$. We will clarify this explicitly
in each instance. The main result of the section will be stated in the case
where $A \in \reals^{m \times N}$ is a normalized $K$-subgaussian matrix whose
rows $m = m(N)$ satisfying a particular growth condition. Thus, the measurement
vector will be given by $y = Ax_{0} + \eta z$ where $\eta > 0$ is the noise
scale and $z_{i} \iid \mathcal{N}(0, 1)$ as before. For the sake of analytical
and notational simplicity, we assume that $\eta$ is independent of $N$. However,
we eventually make clear how $\sigma$ may be allowed to depend on the ambient
dimension, and that our result holds irrespective of this dependence.

\subsection{Underconstrained parameter instability}
\label{sec:underc-param-inst}

As a ``warm-up'' for the main result, we start by demonstrating that there is a
regime in which $\tilde R(\sigma; x_{0}, A, \eta)$ fails to achieve minimax
order-optimality when restricted to $\sigma \geq \eta \sqrt m$. Specifically, if
$m$ is too large, then there is a (sufficiently sparse) vector
$x_{0} \in \Sigma_{s}^{N}$ such that, with high probability on the realization
of $A$, the risk $\tilde R(\sigma; x_{0}, A, \eta)$ is large regardless of the
choice of $\sigma \in [\eta \sqrt m, \infty)$. We defer the proof of this result
to \autoref{sec:subopt-regime-underc}.

\begin{lem}[Underconstrained maximin {\bp}]
  \label{lem:uc-bp-subgaus}
  Fix $\delta, \varepsilon, \eta > 0$, let $1 \leq s < m \leq N$ be integers,
  and suppose $A \in \reals^{m \times N}$ is a normalized $K$-subgaussian
  matrix. If
  \begin{align*}
    m %
    > C_{\varepsilon} \delta^{-2} \tilde K^{2} s^{2} \log^{2} \left(\frac{N}{s}\right),
  \end{align*}
  then with probability at least $1 - \varepsilon$ on the realization of $A$,
  \begin{align*}
    \inf_{\sigma \geq \eta \sqrt m} \sup_{x \in \Sigma_{s}^{N}} %
    \tilde R(\sigma; x, A, \eta) %
    \geq C \sqrt m.
  \end{align*}
\end{lem}

\begin{rmk}
  In some settings, it may be not be appropriate for $m$ to depend
  logarithmically on $N$. If $m$ and $N$ satisfy the power law relation
  $m = N^{k}$ for some $k > 0$, then under the assumptions of
  \autoref{lem:uc-bp-subgaus},
  \begin{align*}
    \E_{z} \|\tilde x (\sigma) - x_{0} \|_{2}^{2} = \Omega(N^{k/2}).
  \end{align*}

\end{rmk}

\subsection{Minimax suboptimality}
\label{sec:bp-minimax-suboptimality}

Just as observed in \autoref{sec:underc-param-inst}, the results of this section
hold in the regime where the aspect ratio approaches a constant:
$m / N \to \gamma \in (0, 1)$.  Our simulations in \autoref{sec:bp-numerics}
support suboptimality of $\tilde R$, and sensitivity of {\bp} to its parameter
choice for aspect ratios ranging from $\gamma = 0.1$ to $\gamma = 0.45$.

Our result is of an asymptotic nature in one additional sense. We have stated
that $\tilde R$ may be suboptimal for ``very sparse'' signals $x_{0}$. This is
specified in the sense that, while $m$ and $N$ may be allowed to grow, $s$
remains fixed. Our numeric simulations demonstrate how this assumption may be
interpreted as the inability of {\bp} to effectively recover the off-support of
the signal $x_{0}$ (\ie the all $0$ sub-vector $x_{T^{C}}$ where
$T \subseteq [N]$ denotes the support of $x_{0}$).
Thus, it is in this setting, where the number of measurements is sufficiently
large, and the sparsity sufficiently small, that we show
$\tilde R(\sigma; x_{0}, A, \eta)$ is asymptotically suboptimal, regardless of
the choice of $\sigma$.

\begin{thm}[{\bp} minimax suboptimality]
  \label{thm:bp-minimax-suboptimality}
  Fix $\varepsilon, \eta > 0$ and $m: \nats \to \nats$ with
  $m(N)/ N \to \gamma \in (0, 1)$. There is $N_{0} \geq 2$ and $p > 0$ so that
  for any $N \geq N_{0}$ and any $1 \leq s < m(N_{0})$, if
  $A \in \reals^{m \times N}$ is a normalized $K$-subgaussian matrix, then with
  probability at least $1 - \varepsilon$ on the realization of $A$,
  \begin{align}
    \label{eq:bp-minimax-suboptimality}
    \inf_{\sigma > 0} \sup_{x \in \Sigma_{s}^{N}}
    \tilde R(\sigma; x, A, \eta) %
    \geq C_{\gamma, K} N^{p}.
  \end{align}

\end{thm}

A minor modification of the above result allows one to show that the minimax
risk for {\bp}, relative to the benchmark risk $R^{*}(s, A)$, converges in
probability to $\infty$.

\begin{coro}[{\bp} suboptimal in probability]
  Fix $\eta > 0$, an integer $s \geq 1$, and suppose for $m : \nats \to \nats$
  that $m (N)/ N \to \gamma \in (0, 1)$. For each $N$, suppose
  $A = A(N) \in \reals^{m(N) \times N}$ is a normalized $K$-subgaussian
  matrix. Then, for all $M > 0$,
  \begin{align*}
    \lim_{N\to \infty} \mathbb{P}\left(%
    \inf_{\sigma > 0} \sup_{x \in \Sigma_{s}^{N}}     %
    \frac{ \tilde R(\sigma; x, A, \eta)}{R^{*}(s,A)} > M \right) %
    = 1.
  \end{align*}
\end{coro}

\section{Numerical results}
\label{sec:numerical-results}

Let $\mathfrak{P} \in \{ \text{\ls}, \text{\qp}, \text{\bp}\}$ be a CS program
with solution $x^{*}(\upsilon)$, where $\upsilon \in \{ \tau, \lambda, \sigma\}$
is the associated parameter. Given a signal $x_{0} \in \reals^{N}$, matrix
$A \in \reals^{m\times N}$, and noise $\eta z \in \reals^{m}$, denote by
$\mathscr{L}(\upsilon; x_{0}, A, \eta z)$ the loss associated to
$\mathfrak{P}$. For instance, if $\mathfrak{P} = \text{{\ls}}$, then
$\mathscr{L} = \hat L$. In most cases, the signal $x_{0}$ for our numerical
simulations will be $s$-sparse, and $s$ will be ``small''. For simplicity, and
to ensure adequate separation of the ``signal'' from the ``noise'', each
non-zero entry of $x_{0}$ will be equal to $N$, except where otherwise
noted. Unless otherwise noted the measurement matrix $A$ will have entries
$A_{ij} \iid \mathcal{N}(0, m^{-1})$.

Define $\upsilon^{*} := \upsilon^{*}(x_{0}, A, \eta) > 0$ to be the value of
$\upsilon$ yielding best risk (\ie where
$\E_{z} \mathscr{L} (\cdot; x_{0}, A, \eta z)$ is minimal) and let the
normalized parameter $\rho$ for the problem $\mathfrak{P}$ be given by
$\rho := \upsilon / \upsilon^{*}$. Note that $\rho = 1$ is a population estimate
of the argmin of $\mathscr{L} (\rho \upsilon^{*}; x_{0}, A, \eta z)$; by the law
of large numbers, this risk estimates well an average of such losses over many
realizations $\hat z$. Let $L(\rho) := \mathscr{L}(\rho \upsilon^{*})$ denote
the loss for $\mathfrak{P}$ as a function of the normalized parameter, let
$\{ \rho_{i}\}_{i=1}^{n}$ denote a sequence of points in the normalized
parameter space, and define the average loss for $\mathfrak P$ at any point
$\rho_{i}$ by
\begin{align*}
  \bar L ( \rho_{i}; x_{0}, A, \eta, k) %
  := k^{-1} \sum_{j = 1}^{k} L(\rho_{i} \upsilon^{*};
  x_{0}, A, \eta \hat z_{ij}),
\end{align*}
where $\hat z_{ij}$ is the $(i,j)$-th realization of noise;
$\hat z_{ij} \sim \mathcal{N}(0, I_{m})$ for all $(i, j) \in [n] \times [k]$. We
may also refer to $\bar L$ as the empirical risk or noise-normalized squared
error (nnse). Note that $\bar L$ depends on
$(\hat z_{ij} : i \in [n], j \in [k])$ and that notating this dependence is
omitted for simplicity. Below, $\hat z_{ij}$ are not necessarily sampled
independently. In fact, to obtain tractable computational simulations, we will
frequently have $\hat z_{ij} = \hat z_{i' j}$ for $i, i' \in [n]$. Where
necessary, we disambiguate the average losses with a subscript:
$\bar{L}_{\text{\ls}}$, $\bar{L}_{\text{\qp}}$ and $\bar{L}_{\text{\bp}}$ for
the programs {\ls}, {\qp} and {\bp}, respectively. 

In this section, we include plots representing average loss of a program with
respect to that program's normalized governing parameter. Both the optimal
parameter $\upsilon^{*}$ and the average loss $\bar L$ are approximated by
$\upsilon^{\dagger}$ and $L^{\dagger}$, respectively, using RBF interpolation as
described in~\autoref{sec:rbf-approximation}. Parameter settings for the
interpolations are provided in~\autoref{sec:interp-param-sett}. To this end, the
average loss $\E_{z} L(\rho; x_{0}, A, \eta z)$ is approximated from $k$
realizations of the true loss on a logarithmically spaced grid of $n$ points
centered about $\rho = 1$. The approximation is computed using multiquadric RBF
interpolation and the values of the parameters of the interpolation,
$(k, n, \varepsilon_{\text{rbf}}, \mu_{\text{rbf}}, n_{\text{rbf}})$, are stated
in each instance where the computation was performed. In every case, due to
concentration effects, for a given program and given parameter value the
realizations cluster very closely about the average loss. Therefore, RBF
interpolation is very close to the true approximated average loss curve computed
from the loss realizations; and has the added advantage of facility to account
for nonuniformly spaced data points. In the main graphics, we omit the original
data point cloud in favour of presenting clean, interpretable plots. However, we
include auxiliary plots of the average loss approximant and the point cloud to
visualize goodness of fit. In addition, these latter visualizations serve to
support how a program's order-optimality for a single realization may be
impacted by averaging over noise.

There is one final caveat to note in how the plots were generated. Computational
methods available to the authors for computing solutions to {\bp} and {\ls} were
much slower than those available for {\qp}. Consequently, ensuring computational
tractability of our numerical simulations required solving {\qp} and obtaining
corresponding parameter values from those problem instances. Namely, given
$\{\rho_{i}\}_{i=1}^{n}$ and solutions $x^{\sharp}(\lambda_{i})$ where
$\lambda_{i} = \rho_{i} \lambda^{*}, i \in [n]$, we use
\autoref{prop:foucart-program-equivalence} to obtain
$\tau_{i} := \|x^{\sharp}(\lambda_{i})\|_{1}$ and
$\sigma_{i} := \|y - Ax^{\sharp}(\lambda_{i})\|_{2}$. Thus, we obtain loss
curves $\hat L(\tau_{i}; x_{0}, A, \eta z)$ and
$\tilde L(\sigma_{i}; x_{0}, A, \eta z)$ by solving {\qp} on a sufficiently fine
grid, yielding $(\lambda_{i}, x^{\sharp}(\lambda_{i}))$. This allows us to
approximate the optimal parameter choices for {\ls} and {\bp} and therefore
determine within some numerical tolerance all of $\bar L_{\text{\ls}}$,
$\bar L_{\text{\qp}}$ and $\bar L_{\text{\bp}}$. Further details for
approximating the average loss of a program are given
in~\autoref{sec:rbf-approximation} where we describe radial basis function (RBF)
interpolation~\cite{buhmann2003radial, 2020SciPy-NMeth}.

\subsection{{\ls} numerics}
\label{sec:ls-numerics}

The data generating process for the numerics in this section is as follows. Fix
$A \in \reals^{m\times N}$, $\eta > 0$ and $x_{0} \in \Sigma_{s}^{N}$. Fix a
logarithmically spaced grid of $n$ points for the normalized parameter,
$\{\rho_{i}\}_{i=1}^{n}$, centered about $1$. Generate $k$ realizations
$\{z_{j}\}_{j=1}^{k}$ of the noise. Obtain
$\lambda_{i} := \rho_{i} \lambda^{*}(A, x_{0}, \eta)$ and obtain
$\ell_{ij} := L^{\sharp}(\lambda_{i}; x_{0}, A, \eta z_{j})$ after computing
$x^{\sharp}(\lambda_{i}; z_{j})$. Observe that
$\hat L(\tau_{ij}; x_{0}, A, \eta z_{j}) = \ell_{ij}$ where
$\tau_{ij} := \|x^{\sharp}(\lambda_{i}; z_{j})\|_{1}$. Similarly, for
$\sigma_{ij} := \|y_{j} - Ax^{\sharp}(\lambda_{i}; z_{j})\|_{2}$, one has
$\tilde L(\sigma_{ij}; x_{0}, A, \eta z_{j}) = \ell_{ij}$. Finally, for a
sufficiently fine and wide numerical grid, one may approximate the normalized
parameter grids $\{ \tau_{ij}\}$ and $\{ \sigma_{ij}\}$ using the values
$\{ \ell_{ij}\}$. Consequently, for each program we are able to approximate the
average loss $\bar L$ by obtaining a clever approximate interpolant of
$\{(\tau_{ij}, \ell_{ij}) : (i,j) \in [n] \times [k]\}$,
$\{(\sigma_{ij}, \ell_{ij}) : (i,j) \in [n] \times [k]\}$ or
$\cup_{j \in [k]}\{(\rho_{i}, \ell_{ij}): i \in [n]\}$, respectively, while only
having to solve {\qp}. This particular bit of good fortune is guaranteed to us
by \autoref{prop:foucart-program-equivalence}. As stated, the clever approximant
is obtained using multiquadric RBF
interpolation~\cite{buhmann2003radial,2020SciPy-NMeth}. For more background on
kernel methods for function approximation and radial basis functions in
particular, we refer the reader to~\cite{buhmann2003radial, hastie2009elements,
  murphy2012machine}. Some additional detail to this end is provided in \autoref{sec:rbf-approximation}.

The numerics for {\ls}, appearing in \autoref{fig:low-noise-numerics}, concern
the case in which $\eta$ is small. When the ambient dimension is modest
($N = 10^{4}$) and the noise scale only moderately small
($\eta = 2\cdot 10^{-3}$), parameter instability of {\ls} is readily
observed. Minute changes in $\tau$ lead to blow-up in the nnse and the peak
signal-to-noise ratio (psnr) (left and right plot, respectively). Indeed, for
the range plotted, it is difficult to visually segment the left half of the
{\ls} average loss curve from the right half. These observations support the
asymptotic theory of \autoref{sec:LS-instability}. Moreover, the simulations
suggest that the other two programs, {\bp} and {\qp} are relatively much less
sensitive to the choice of their governing parameter.

\begin{figure}
  \centering
  \includegraphics[width=.45\textwidth]{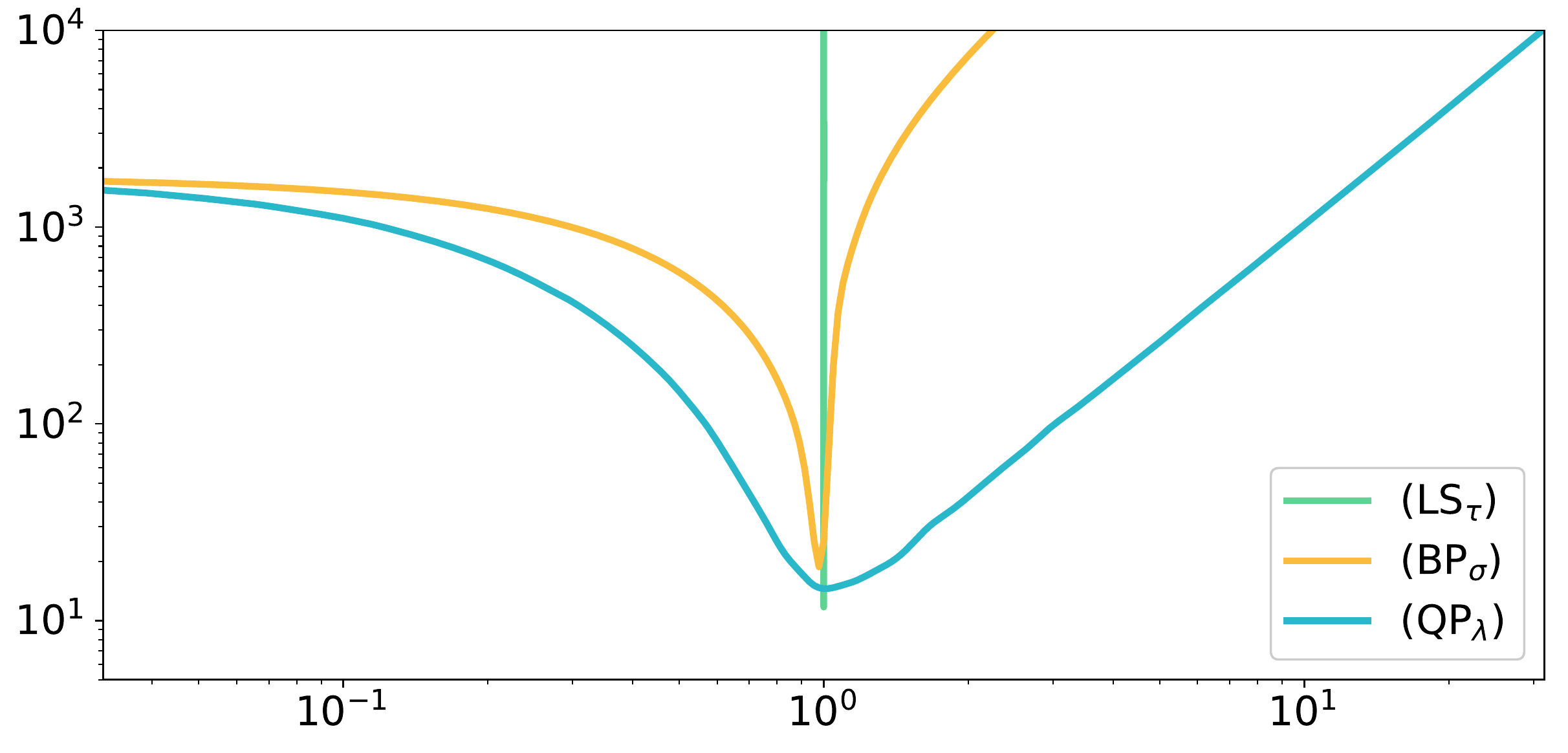}\quad
\includegraphics[width=.45\textwidth]{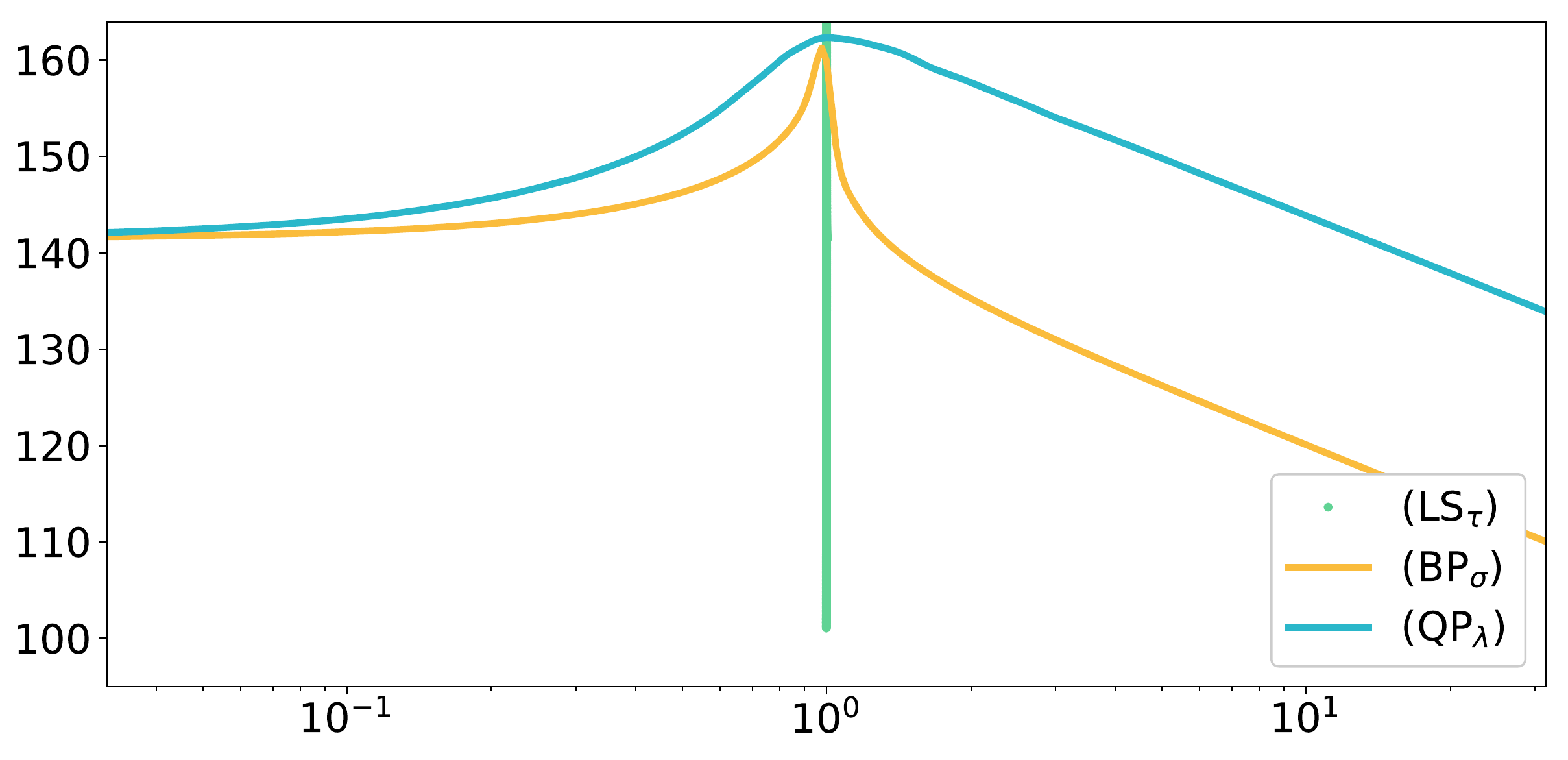}

\caption{{\ls} parameter instability in the low-noise regime. Average loss
  (left) plotted on a log-log scale with respect to the normalized parameter;
  average psnr (right) plotted on a log-linear scale with respect to the
  normalized parameter. The data parameters are
  $(s, m, N, \eta, k, n) = (1, 2500, 10^{4}, 2\cdot 10^{-3}, 15, 301)$.}
\label{fig:low-noise-numerics}
\end{figure}

\subsection{{\qp} numerics}
\label{sec:qp-numerics}

In this section we visualize the average loss of {\qp} as a function of its
normalized parameter $\rho = \lambda / \lambda^{*}$. In
\autoref{fig:qp-instability}, the average loss for {\qp} is plotted with respect
to the $\rho$ for an aspect ratio $\delta$ ranging between $0.25$ and $4$. As
suggested by \autoref{thm:qp-r-stability}, the average loss appears to scale
quadratically with respect to the normalized parameter for values $\rho >
1$. For $\rho \in (0.5, 0.9)$, the average loss appears to scale
super-quadratically with respect to the normalized parameter, with the rate of
growth increasing as a function of $\delta$. This behaviour suggests that {\qp}
can be sensitive to its parameter choice if $\rho$ is too small. The intuition
for the observed behaviour is that {\qp} increasingly behaves like ordinary
least squares when $\lambda \to 0$. Each average loss was approximated from 15
realizations of the loss using multiquadric RBF interpolation. Due to
concentration effects, the realizations for each parameter value clustered very
closely to the approximated average loss. The left-hand plot
\autoref{fig:qp-instability} is plotted on a log-log scale, while the right-hand
plot is plotted on a linear-linear scale. The linear-linear plot readily
demonstrates how over-guessing $\lambda$ by a factor of $2$ is more robust to
error than under-guessing $\lambda$ by a factor of $2$.

\begin{figure}[h]
  \centering
  \includegraphics[width=.45\textwidth]{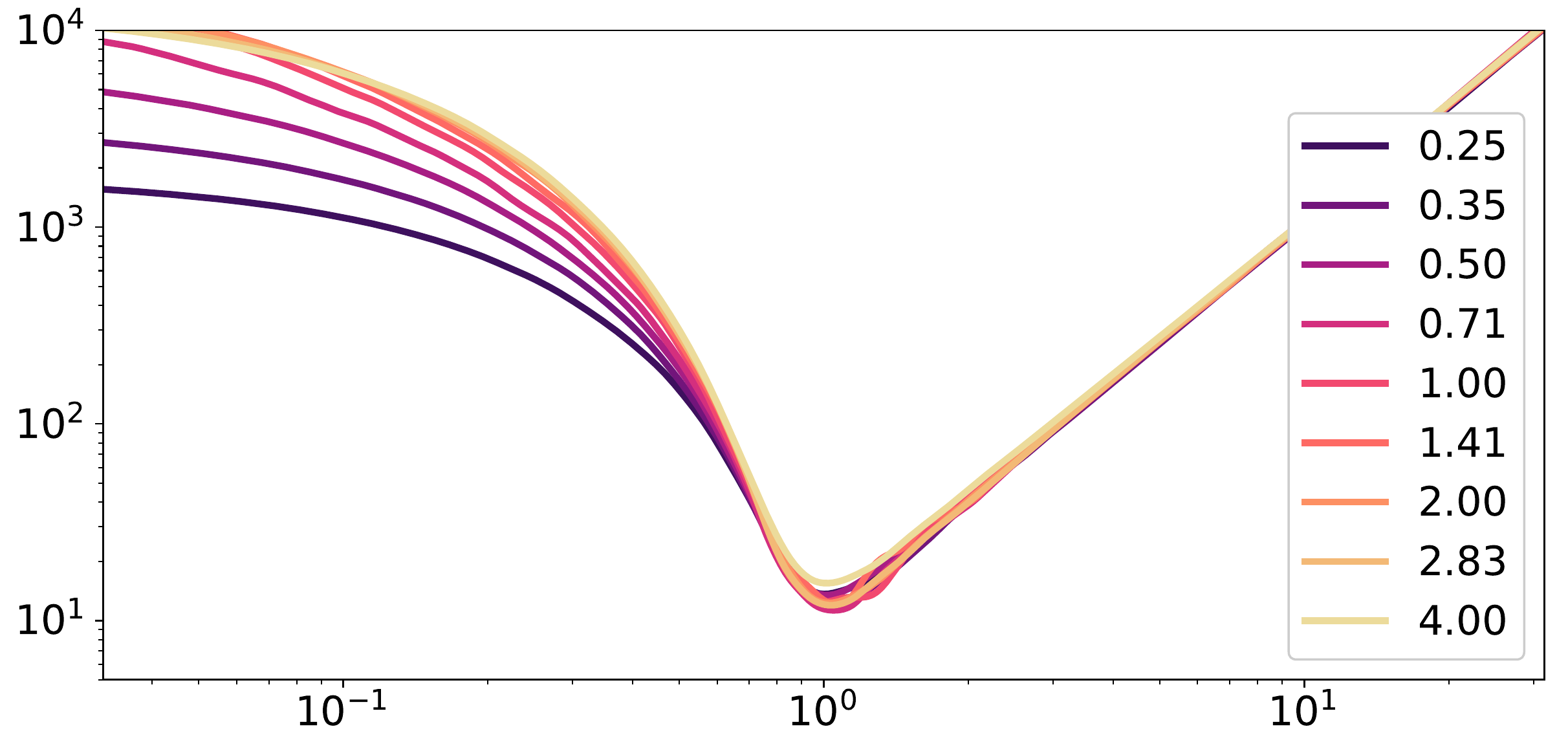}\quad %
\includegraphics[width=.45\textwidth]{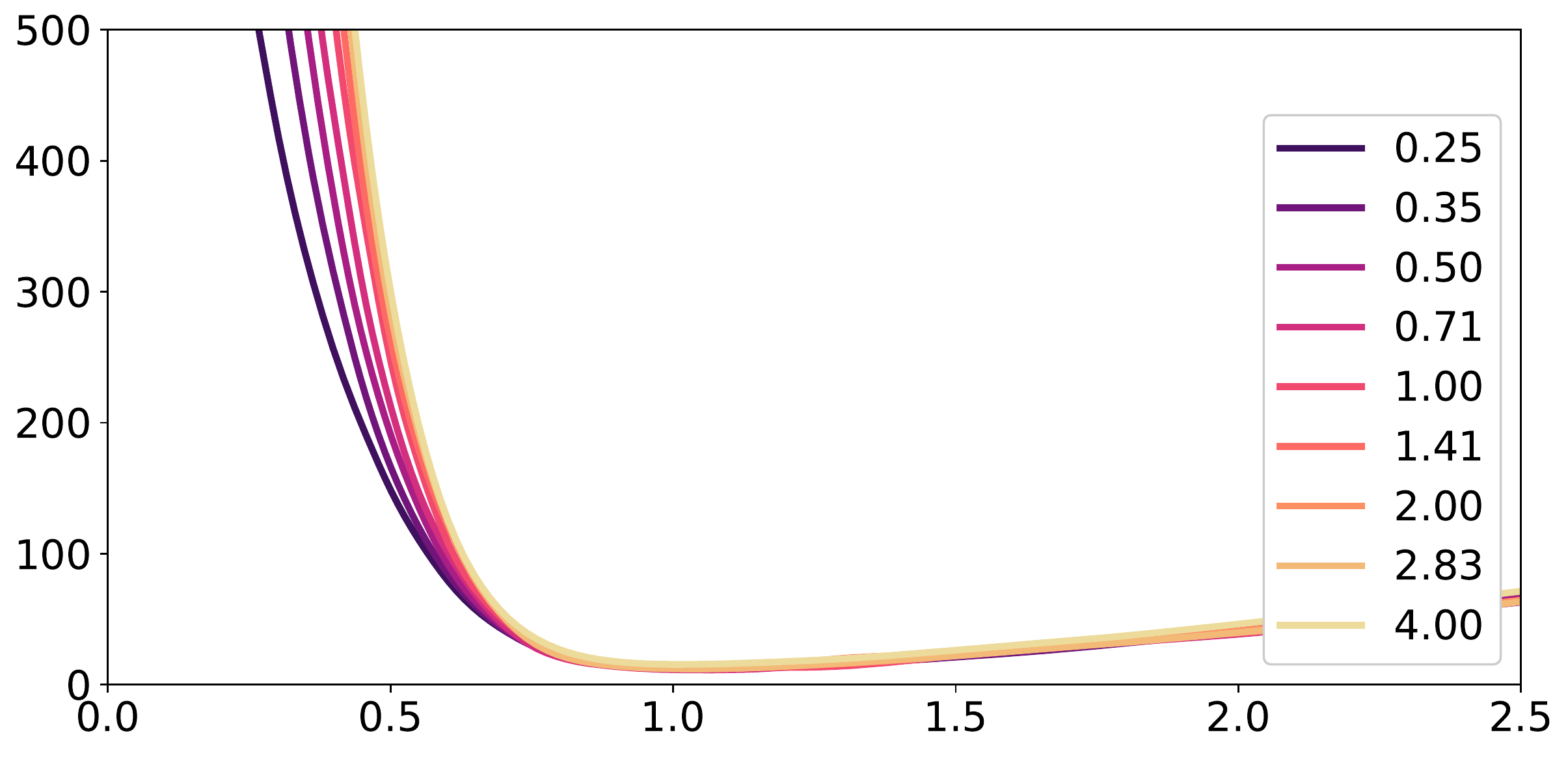}
\caption{Average loss of {\qp} plotted with respect its normalized parameter in
  the low-noise, high sparsity regime. Parameters for the simulation are
  $(s, N, \eta, k, n) = (1, 10^{4}, 10^{-5}, 15, 301)$. The aspect ratio of the
  matrix $A \in \reals^{m \times N}$ with $A_{ij} \iid \mathcal{N}(0, m^{-1})$
  takes values $\delta \in \{4^{-1}, \ldots, 4\}$, as shown in the legend. The
  data are visualized on a log-log scale (left) and linear scale (right).}
  \label{fig:qp-instability}
\end{figure}

In \autoref{fig:qp-instability-2}, we visualize {\qp} average loss with respect
to its normalized parameter. In the top row, we include two plots similar to
\autoref{fig:qp-instability}, but for $\delta = 0.25, 0.45$ only. Again, the
left-hand plot is on a log-log scale while the right-hand plot is on a
linear-linear scale. The bottom row depicts the goodness of fit of the RBF
approximation for the average loss.

\begin{figure*}[h]
  \centering
\includegraphics[width=.45\textwidth]{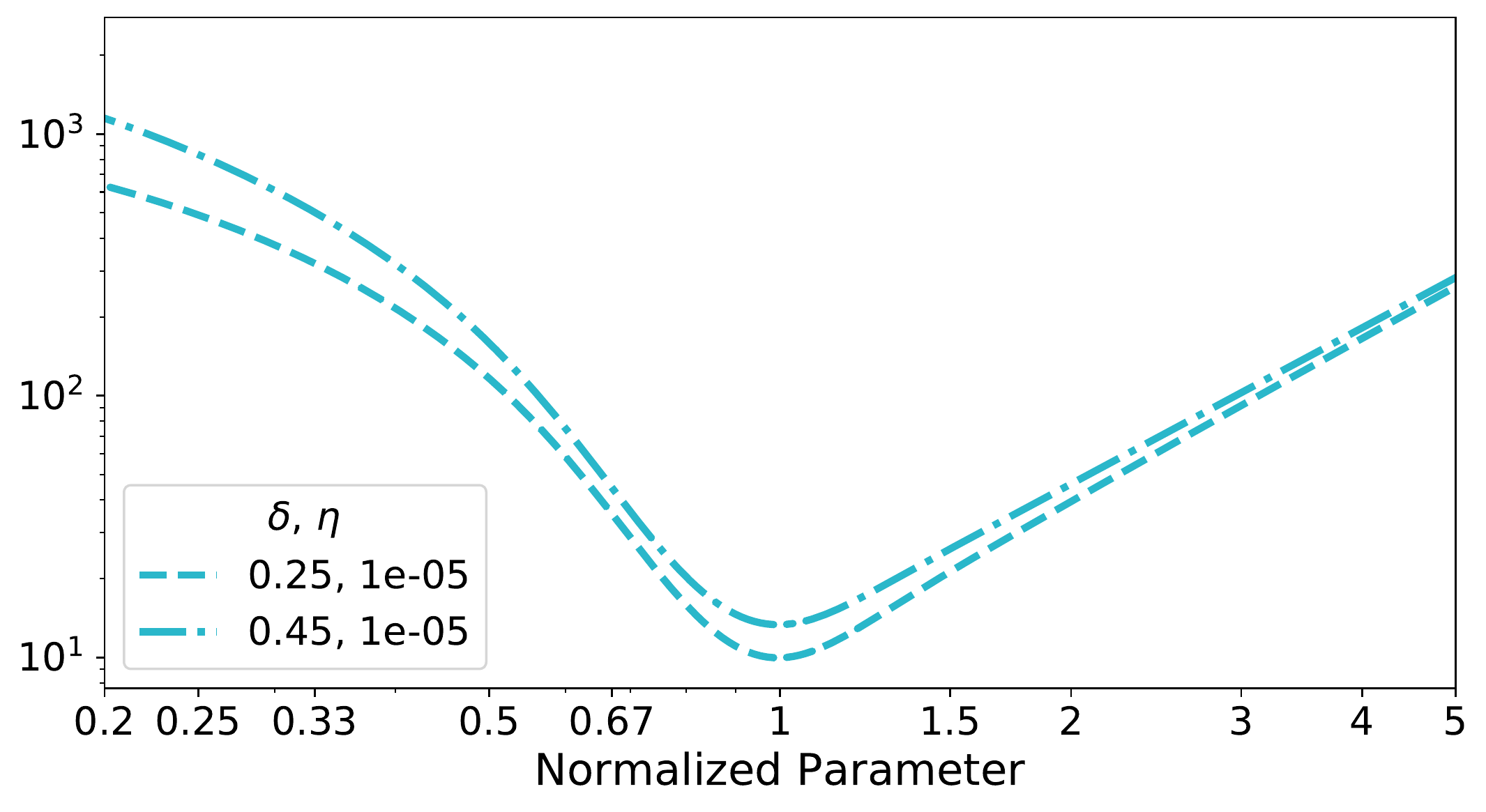}\quad
\includegraphics[width=.45\textwidth]{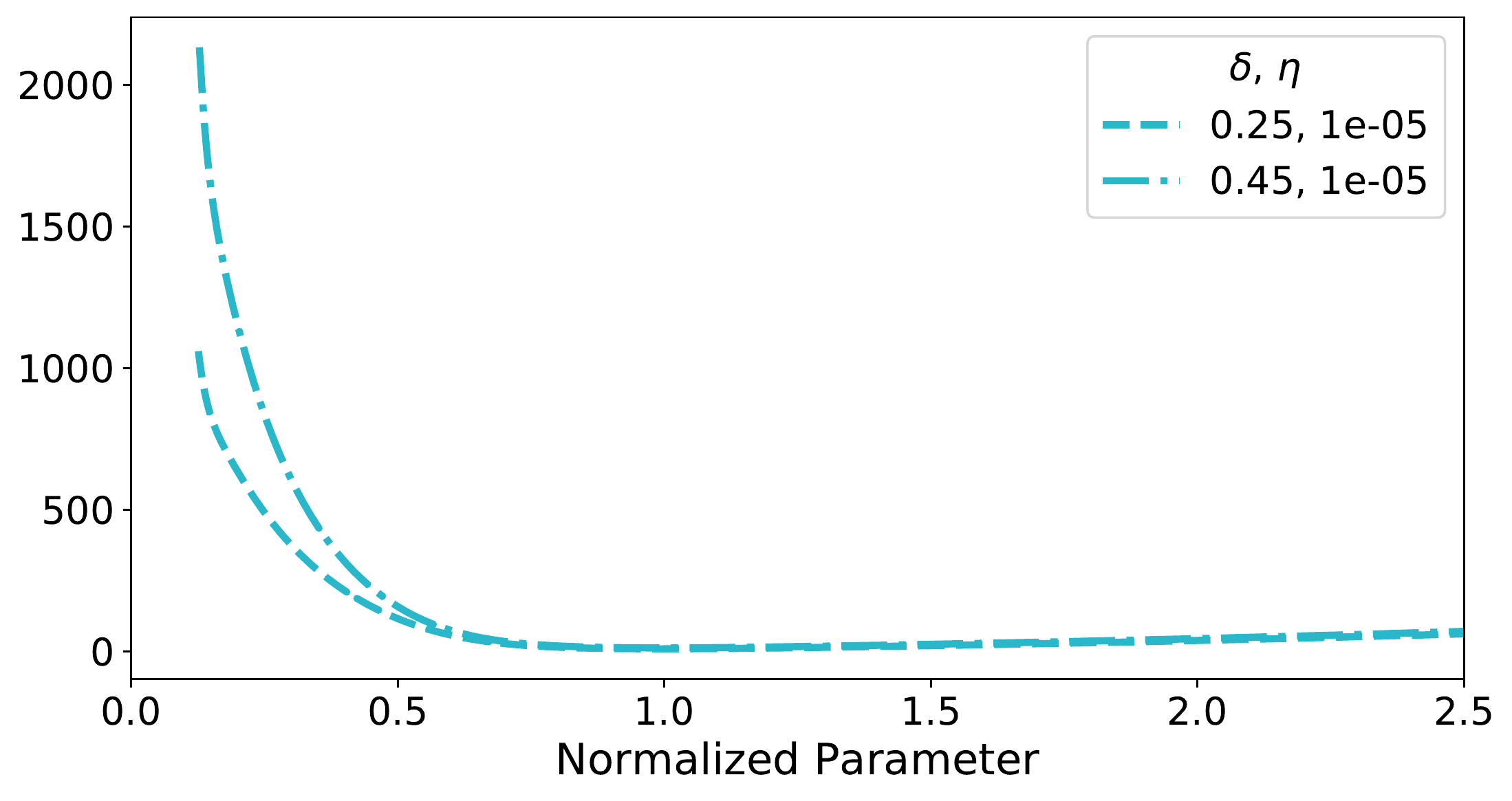}
  \includegraphics[width=.9\textwidth]{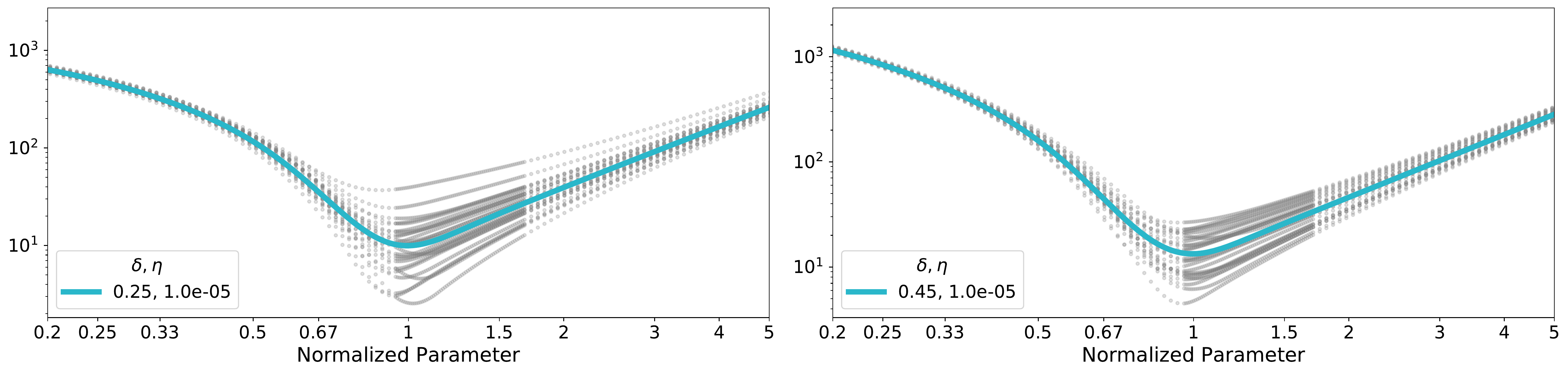}
  \caption{Average loss of {\qp} plotted with respect to its normalized
    parameter in the low-noise, high-sparsity regime. Parameters for the
    simulation are $(s, N, \eta, k, n) = (1, 10^{4}, 10^{-5}, 25, 201)$. The
    aspect ratio of the matrix $A \in \reals^{m \times N}$ with
    $A_{ij} \iid \mathcal{N}(0, m^{-1})$ takes values $\delta = 0.25, 0.45$, as
    shown in the legend. The data in the top row are visualized on a
    $\log$-$\log$ scale (left) and linear scale (right). The bottom row depicts
    the quality of fit for the RBF approximation of the average loss for
    $\delta = 0.25$ (left) and $\delta = 0.45$ (right). }
  \label{fig:qp-instability-2}
\end{figure*}

\subsection{{\bp} numerics}
\label{sec:bp-numerics}

This section includes numerical simulations depicting the sensitivity of {\bp}
to its parameter choice. These numerics serve to support the asymptotic theory
developed in \autoref{sec:analysis-bp}.

The graphics in \autoref{fig:bp-numerics-1} serve as an initial depiction of
{\bp} parameter instability, depicting the average loss for each program and for
$N \in \{4000, 7000\}$, $\delta \in \{ 0.1, 0.25, 0.45\}$. Each plot depicts the
average loss as a function of the normalized parameter for {\ls} (green), {\bp}
(orange) and {\qp} (blue). The domain of the normalized parameter in each plot
is $(0.2, 5)$. A single realization of $A$ was fixed and the average loss was
computed from $k = 50$ realizations of the noise by constructing a function
approximator using radial basis function approximation with a multiquadric
kernel. The RBF approximator was evaluated on a logarithmically spaced grid of
$n_{\text{rbf}} = 301$ points centered about $1$. The loss values for {\ls} and
{\bp} were computed by using the program equivalence described
by~\autoref{prop:foucart-program-equivalence}. In particular, for computational
expediency, once $A$ and $z$ were fixed, the \textsc{Lasso} program was solved
only using {\qp}, for all $\lambda$ in a specified range. For each
$x^{\sharp}(\lambda)$, we obtained $\hat x (\tau)$ and $\tilde x (\sigma)$ using
that $\hat x(\tau) = \tilde x (\sigma) = x^{\sharp}(\lambda)$ for
$\tau := \|x^{\sharp}(\lambda) \|_{1}$ and
$\sigma := \|y - A x^{\sharp}(\lambda)\|_{2}$.

For convenience, we refer to each plot in \autoref{fig:bp-numerics-1} \emph{via}
its (row, column) position in the figure. The collection of plots serves to
depict how the average loss changes as a function of $N$ and $\delta = m/N$ when
$\eta = 1$. Namely, as $N$ increases, the average loss for {\bp} becomes sharper
about the optimal parameter choice. In addition, as $\delta$ increase, we
observe the same phenomenon. In \autoref{fig:bp-numerics-2}, similar content is
depicted, but for $\eta = 100$. In this case, $n_{\text{rbf}} = 501$ was used.

Specific paramter settings for the RBF approximation for each set of problem
parameters and program are detailed in
\autoref{tab:rbf-parameter-settings}. Because {\bp} parameter instability is not
easily visualized in small dimensions (\eg for $N < 10^{6}$), we supply several
plots visualizing the quality of the RBF approximation. Namely, approximation
quality plots corresponding with \autoref{fig:bp-numerics-1} may be found in
\autoref{fig:bp-numerics-1-approx-quality}; and approximation quality plots are
included in \autoref{fig:bp-numerics-2}. Each row of these plots is a triptych;
each column corresponds to a program: {\ls} for the left-most, {\bp} in the
centre; and {\qp} on the right. These plots depict a single line and a
collection of points. The points correspond to individual loss values for each
realization of the noise and each normalize parameter value computed. The line
corresponds to the RBF approximation of the average loss for that program. The
domain for the {\ls} plots is $(0.95, 0.95^{-1})$ in the normalized parameter
space. For {\bp} it is $(0.9, 0.9^{-1})$, and for {\qp} $(0.75, 0.75^{-1})$.

In \autoref{fig:bp-numerics-1-approx-quality}, one may observe by inspection
that the loss realizations for {\ls} and {\qp} typically are achieved very close
to $1$. In contrast, there is a relatively wider range in the domain for where
the {\bp} loss achieves its optimum. This is integral to how {\bp} risk is
sensitive to the choice of $\sigma$.

There appear to be two competing factors that impact sensitivity to parameter
choice and program optimality. The first is stability of the parameter with
respect to variation due to the noise realization. This has already been
described: because $\sigma (\lambda^{*}, z)$ varies greatly as a function of
$z$, program optimality is destroyed by suboptimal loss values near
$\sigma = \sigma^{*}$. On the other hand, this also tends to have somewhat of a
smooth effect about the optimal normalized parameter. Thus, as
$\tau(\lambda^{*}, z)$ does not vary as greatly in this manner, its sensitivity
to parameter choice is not smoothed due to local averaging effects.

\begin{figure*}[h]
  \centering
  \includegraphics[width=.45\textwidth]{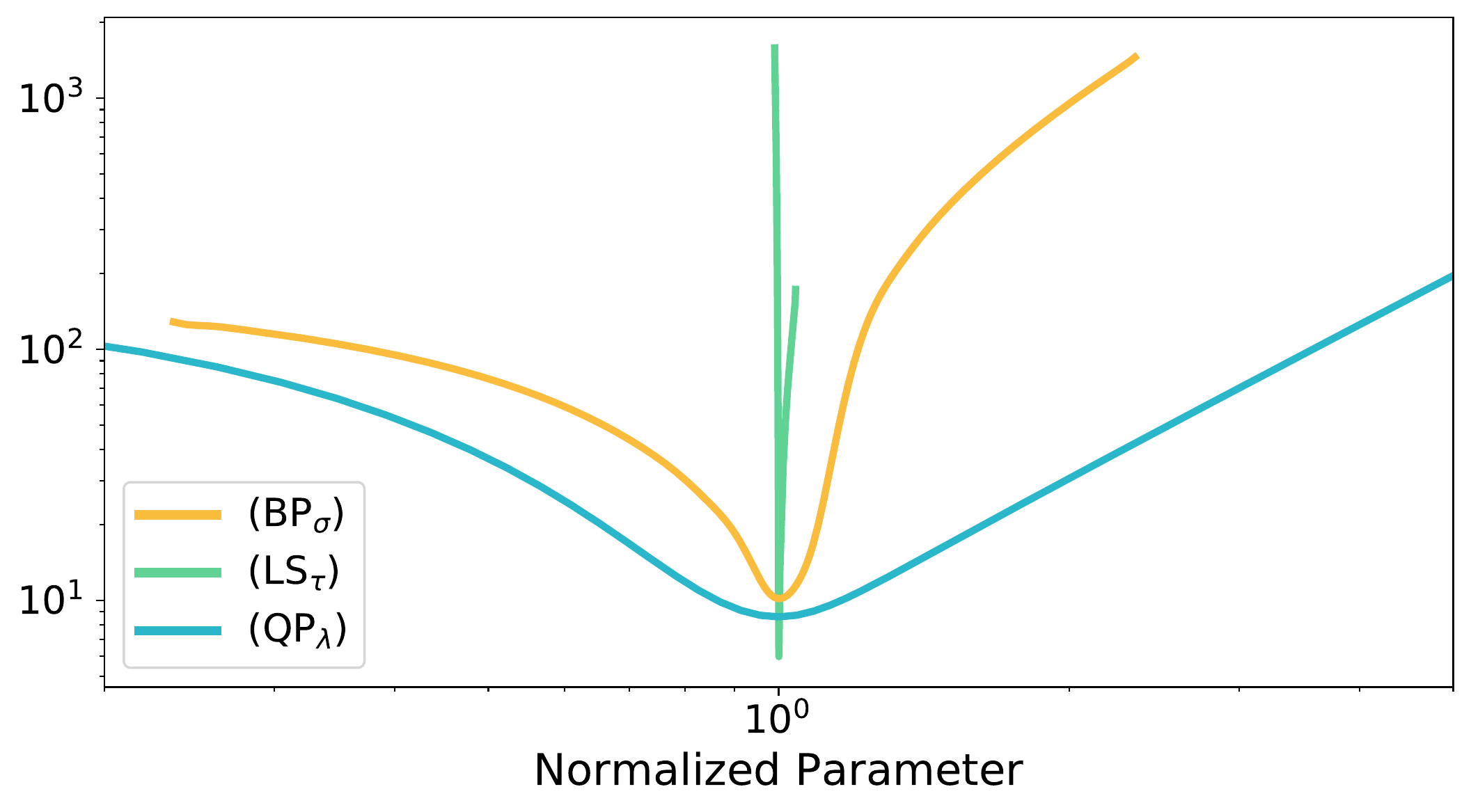}\hfill
\includegraphics[width=.45\textwidth]{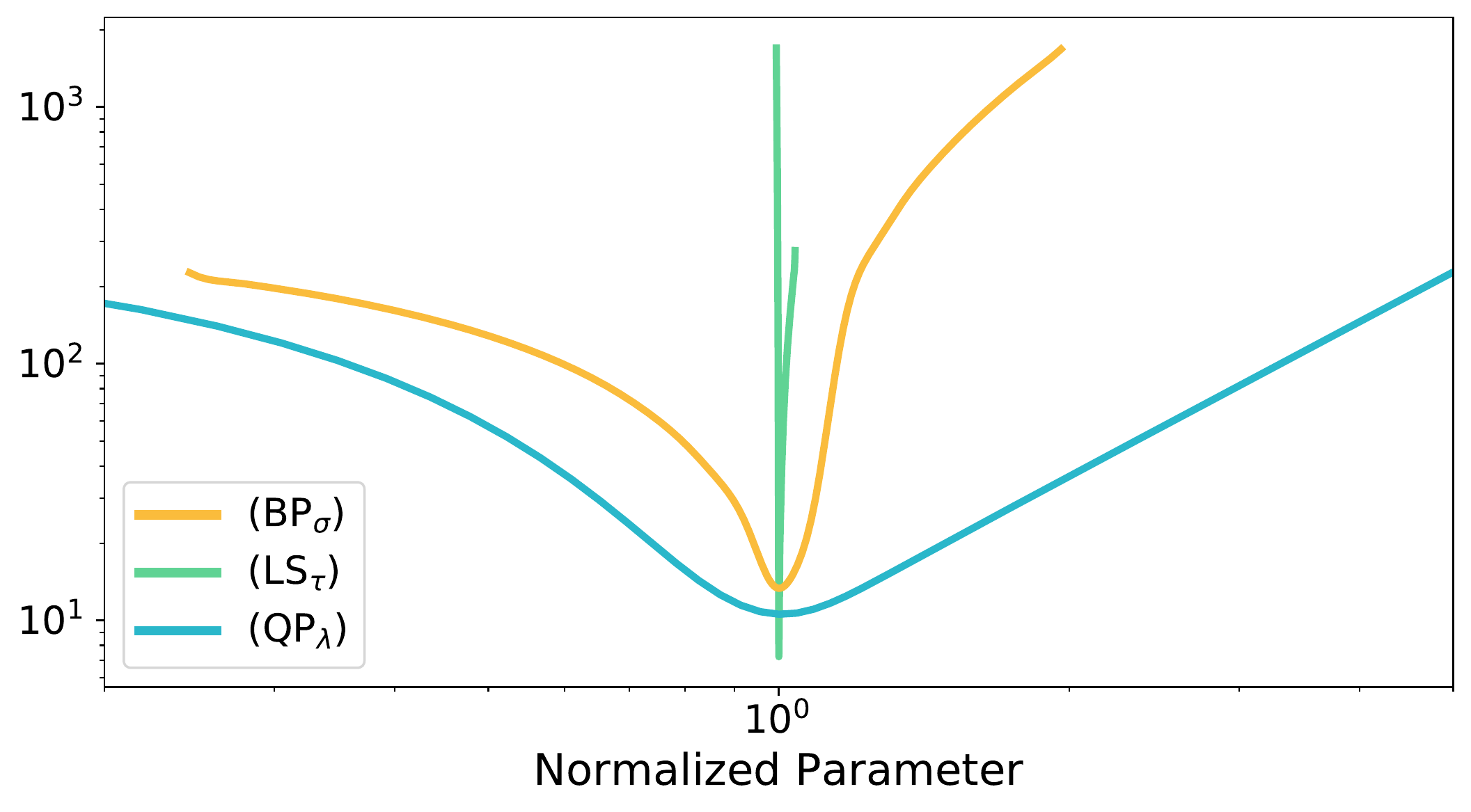}

  \includegraphics[width=.45\textwidth]{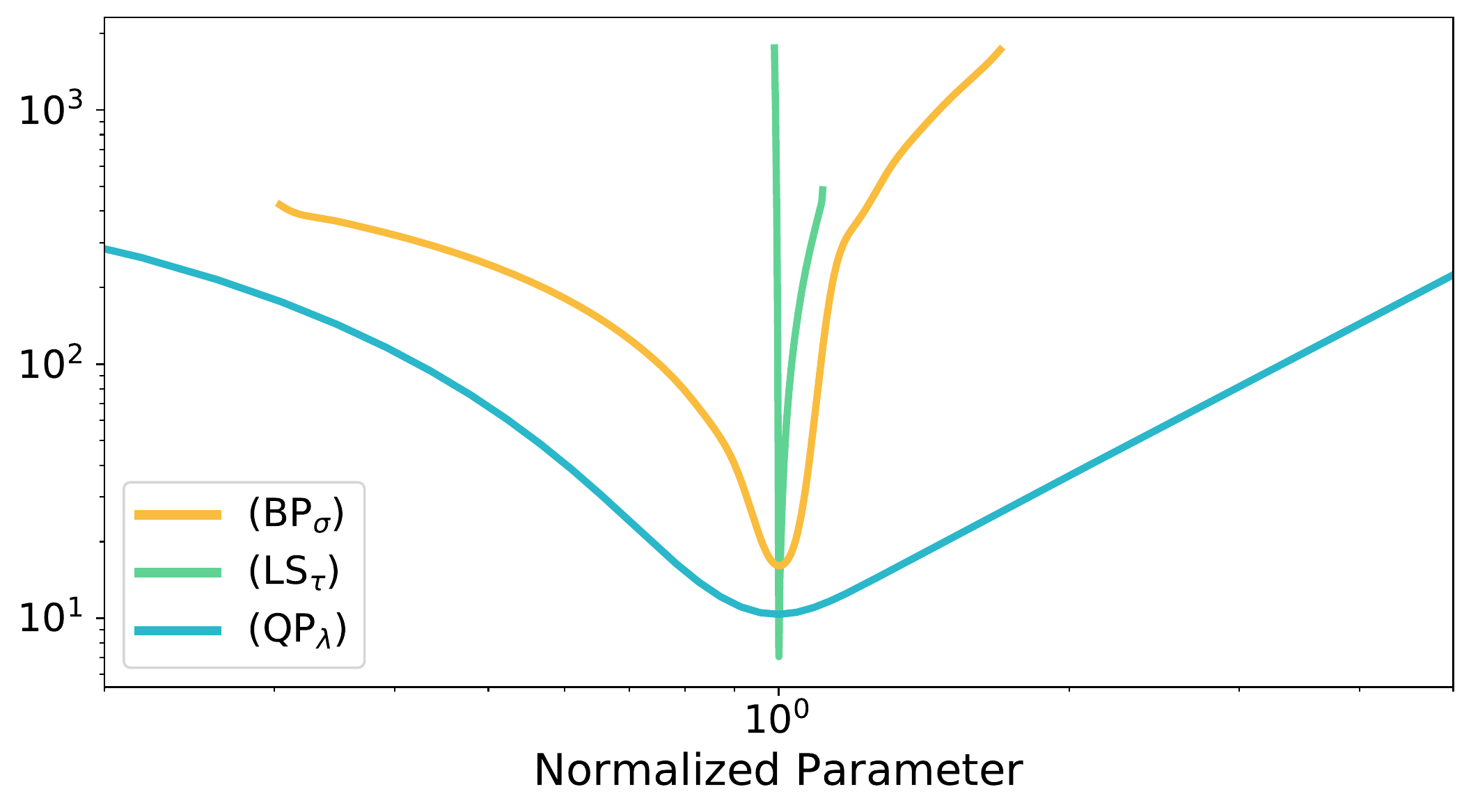}\hfill
\includegraphics[width=.45\textwidth]{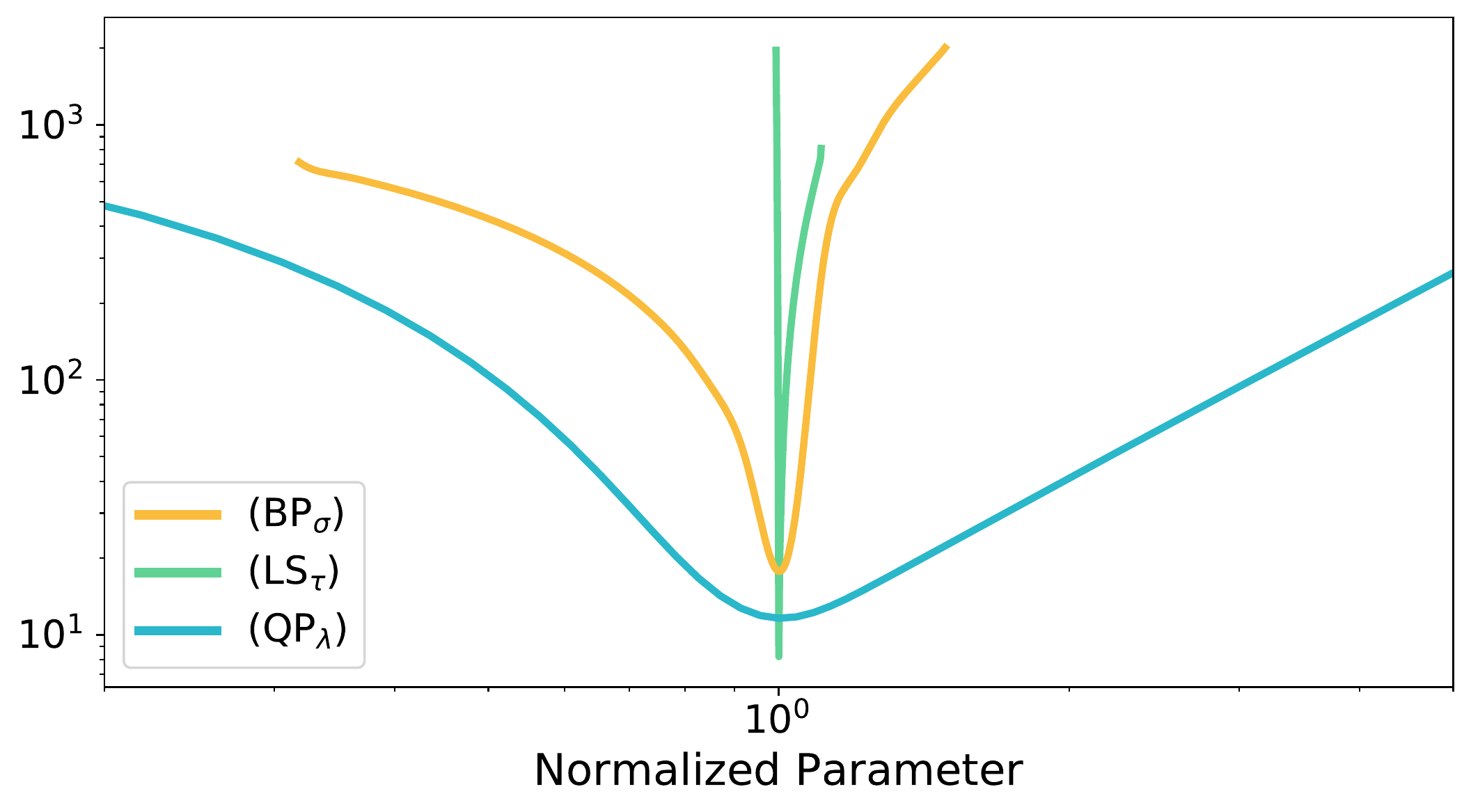}

  \includegraphics[width=.45\textwidth]{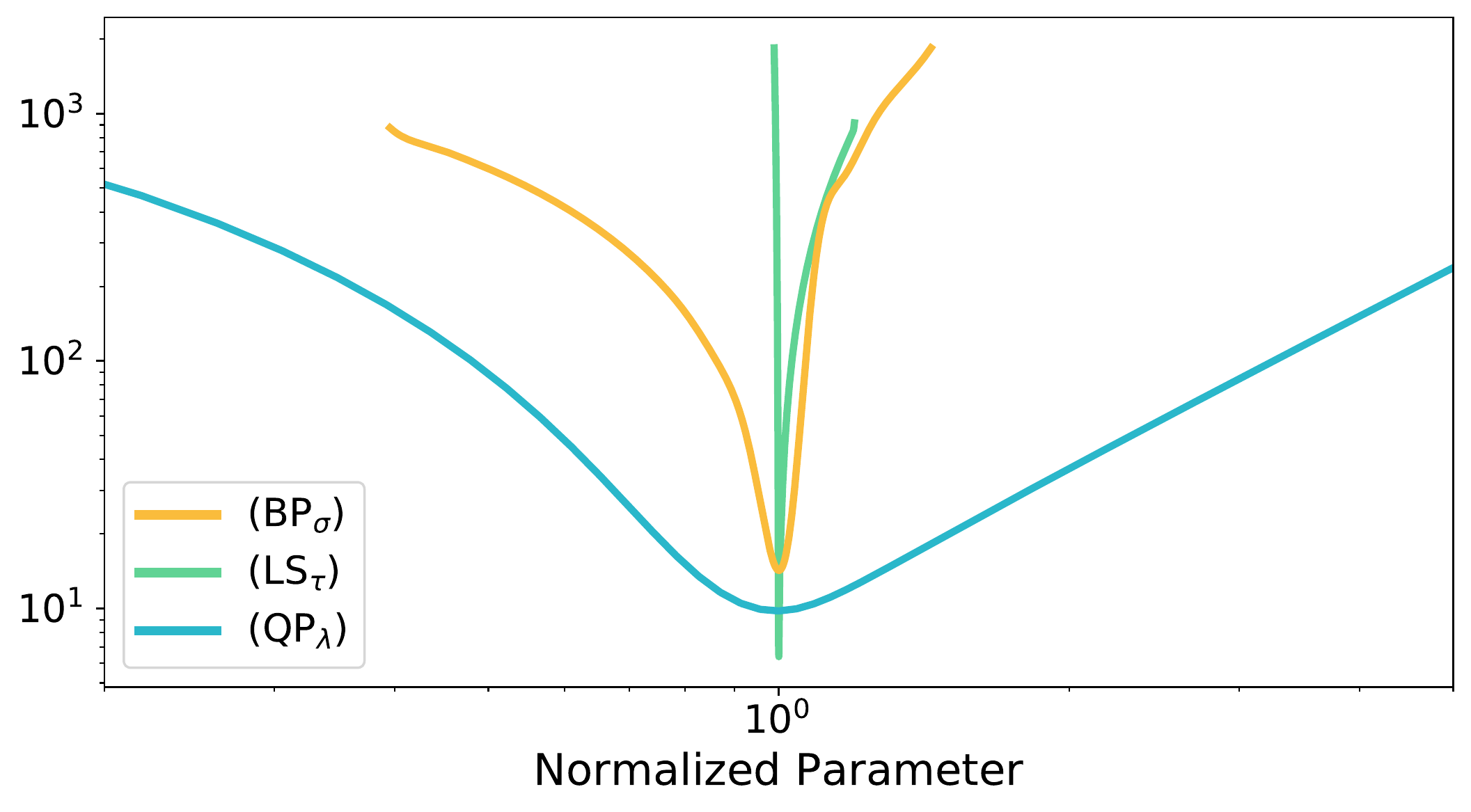}\hfill\hphantom{a}

  \caption{Each plot depicts the average loss as a function of the normalized
    parameter for each of the three programs under consideration. The collection
    of plots depicts how the average loss changes as a function of $N$ and
    $\delta = m/N$. Details for each plot will be given by referencing the (row,
    column) position of the plot in this figure. The domain of the normalized
    parameter in each plot is $(0.2, 5)$. A single realization of $A$ was fixed
    and the average loss was computed from $k = 50$ realizations of the noise by
    constructing a function approximator using radial basis function
    approximation with a multiquadric kernel. The RBF approximator was evaluated
    on a logarithmically spaced grid of $n = 301$ points centered about $1$. %
    \textbf{(1,1):} $(s, N, \delta, \eta) = (1, 4000, 0.1, 1)$;
    \textbf{(1,2):} $(s, N, \delta, \eta) = (1, 7000, 0.1, 1)$;
    \textbf{(2,1):} $(s, N, \delta, \eta) = (1, 4000, 0.25, 1)$;
    \textbf{(2,2):} $(s, N, \delta, \eta) = (1, 7000, 0.25, 1)$;
    \textbf{(3,1):} $(s, N, \delta, \eta) = (1, 4000, 0.45, 1)$.}
  \label{fig:bp-numerics-1}
\end{figure*}

\begin{figure*}[h]
  \centering

  \begin{minipage}[h]{.7\linewidth}
      \large $N = 4000, \eta = 1$
      \hfill
      \large $N = 7000, \eta = 1$
  \end{minipage}

  \includegraphics[width=.45\textwidth]{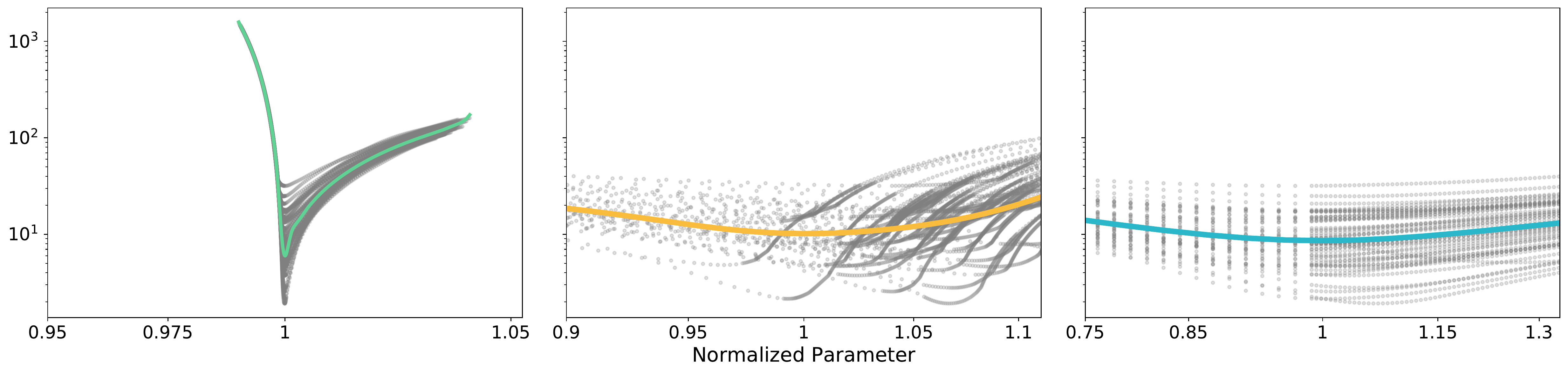}
  \hfill
\includegraphics[width=.45\textwidth]{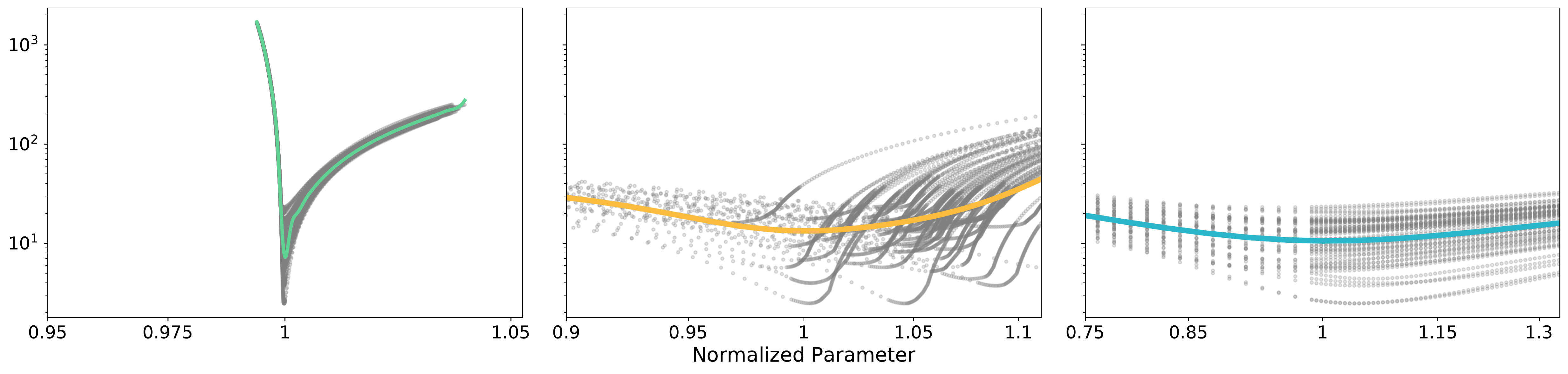}

\includegraphics[width=.45\textwidth]{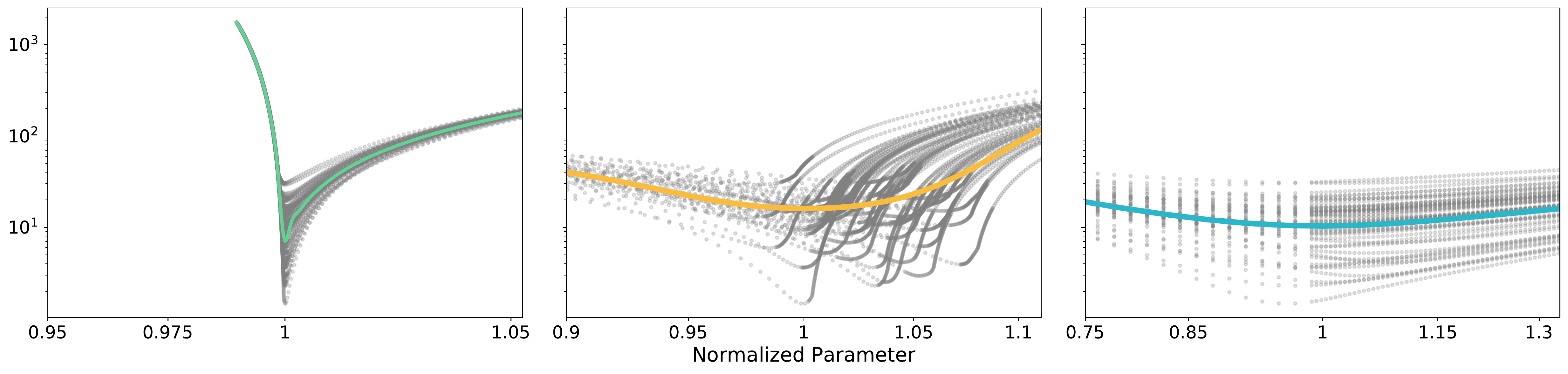}
\hfill
\includegraphics[width=.45\textwidth]{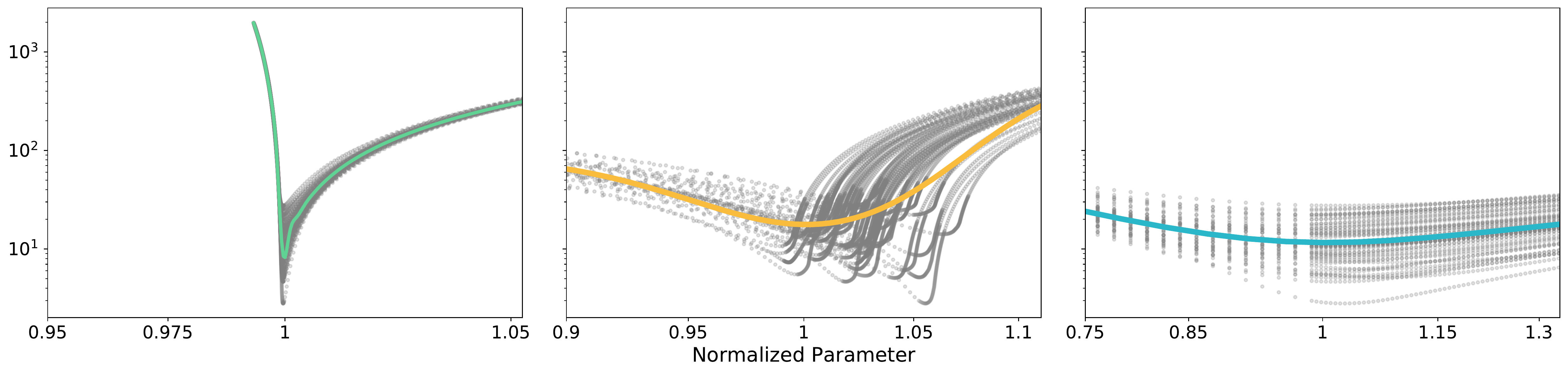}

\includegraphics[width=.45\textwidth]{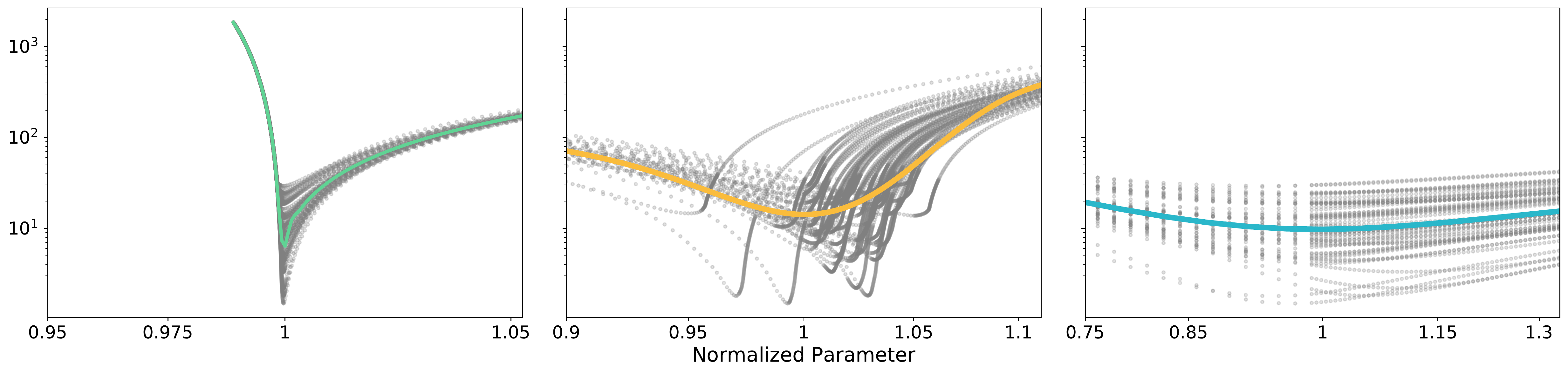}
\hfill\hphantom{a}

  

  \caption{Each plot depicts the quality of the RBF approximation about the
    optimal normalized parmaeter. The left-most plot is in every case depicting
    the loss and (approximate) average loss of {\ls}; the middle that for {\bp};
    and the right that for {\qp}. Top-to-bottom:%
    $(s, N, \delta, \eta) = (1, 4000, 0.1, 1)$; 
    $(s, N, \delta, \eta) = (1, 4000, 0.25, 1)$;
    $(s, N, \delta, \eta) = (1, 4000, 0.45, 1)$;
    $(s, N, \delta, \eta) = (1, 7000, 0.1, 1)$; 
    $(s, N, \delta, \eta) = (1, 7000, 0.25, 1)$.
    }
  \label{fig:bp-numerics-1-approx-quality}
\end{figure*}

\begin{figure*}[h]
  \centering
  
\includegraphics[width=.45\textwidth]{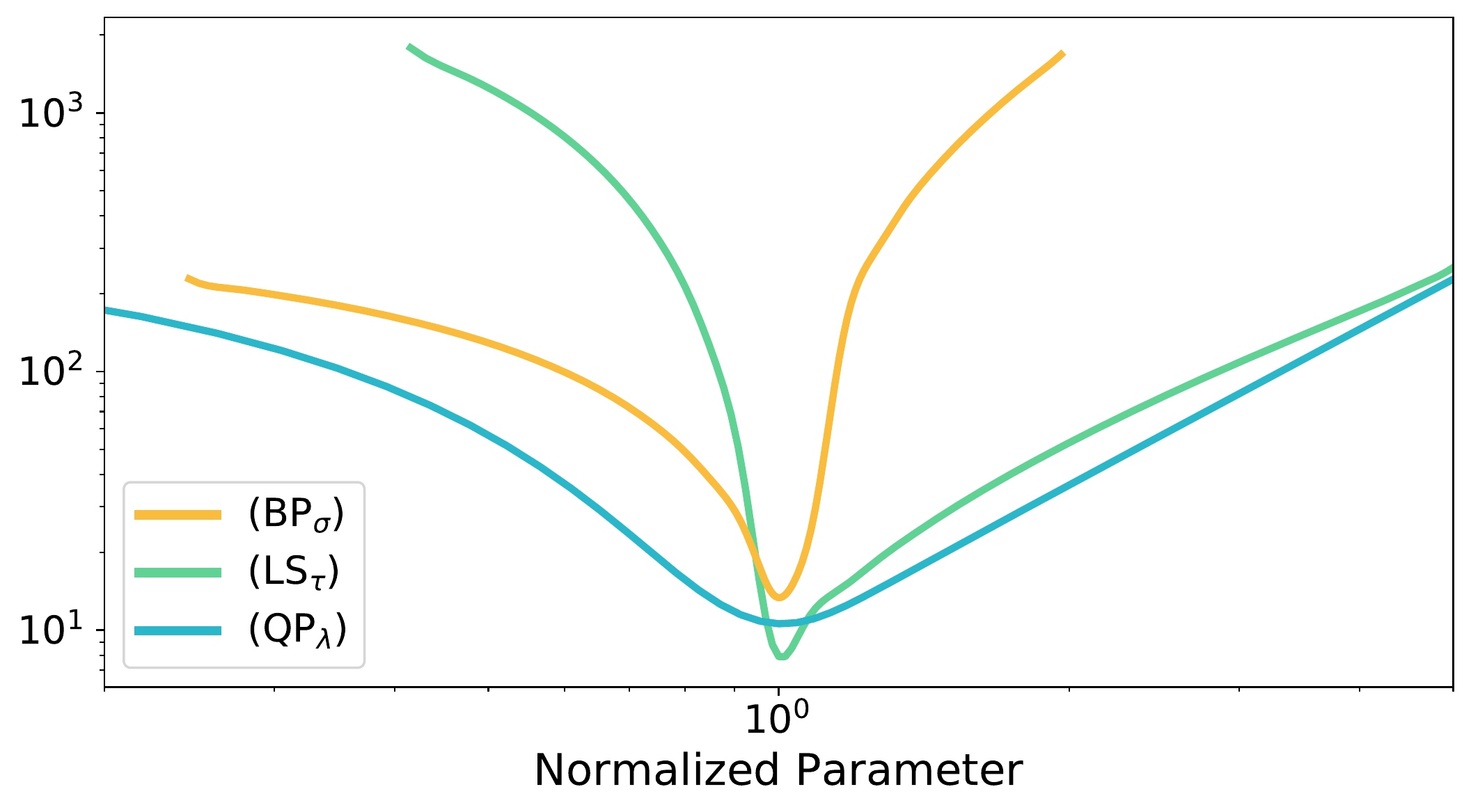}
\hfill \includegraphics[width=.45\textwidth]{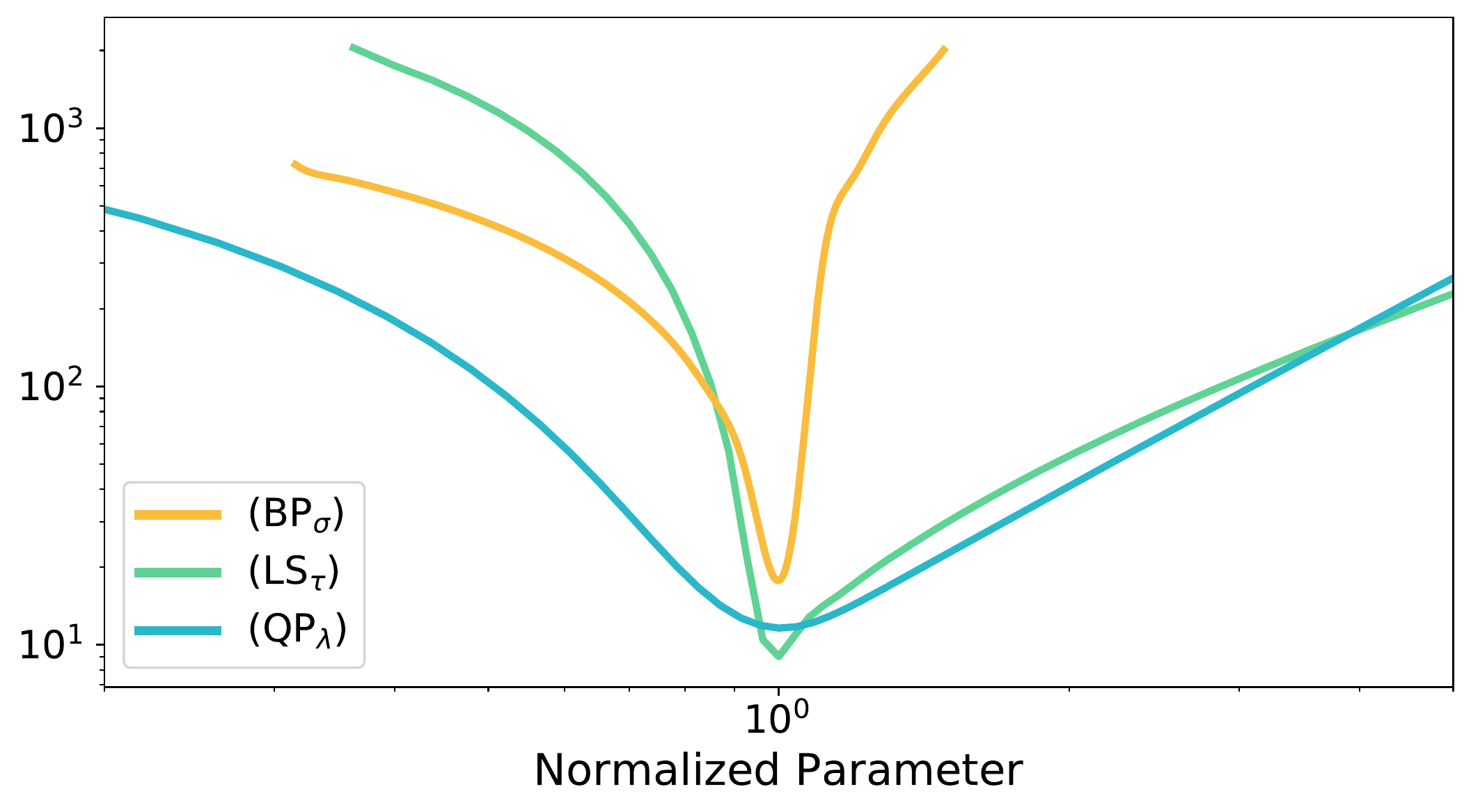}

  \vspace{-6pt}
  \textcolor{lightgrey}{\rule{.45\textwidth}{.75pt}}\hfill
  \textcolor{lightgrey}{\rule{.45\textwidth}{.75pt}}
  \vspace{6pt}

\includegraphics[width=.45\textwidth]{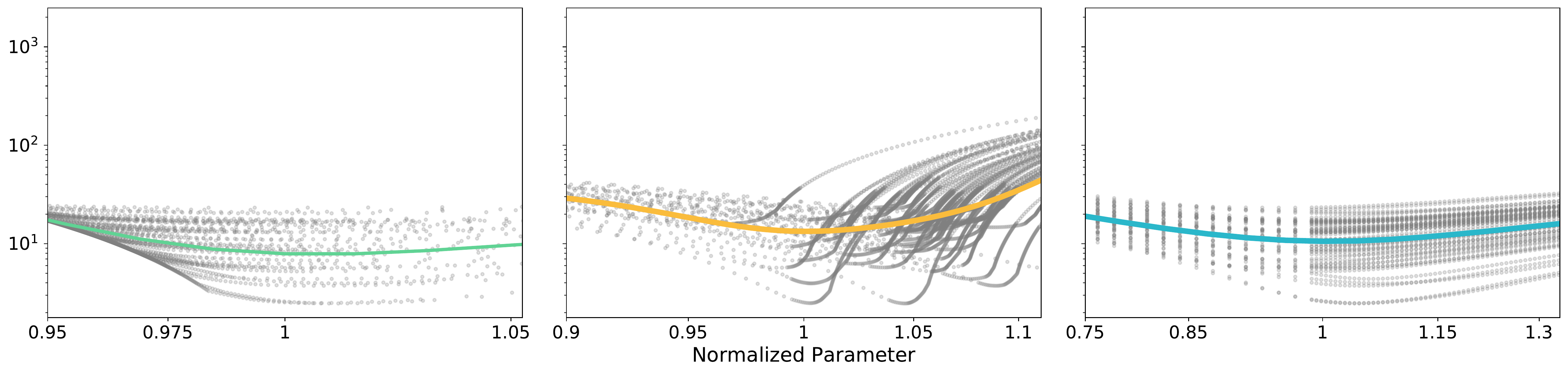}
\hfill
\includegraphics[width=.45\textwidth]{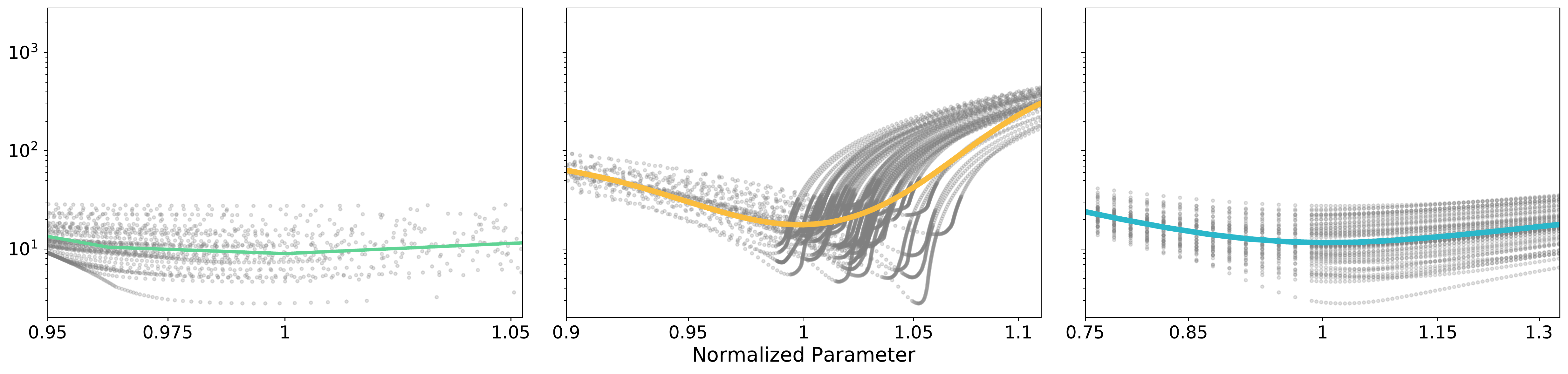}

\caption{\textbf{Top row:} Each plot depicts the average loss as a function of
  the normalized parameter for each of the three programs under
  consideration. The collection of plots depicts how the average loss changes as
  a function of $\delta = m/N$. The domain of the normalized parameter in each
  plot is $(0.2, 5)$. A single realization of $A$ was fixed and the average loss
  was computed from $k = 50$ realizations of the noise by constructing a
  function approximator using radial basis function approximation with a
  multiquadric kernel. The RBF approximator was evaluated on a logarithmically
  spaced grid of $n_{\text{rbf}} = 501$ points centered about
  $1$. \textbf{Bottom row:} Each plot depicts the quality of the RBF
  approximation about the optimal normalized parmaeter. In each triptych, the
  left plot depicts the loss and (approximate) average loss of {\ls}; the middle
  that for {\bp}; and the right that for {\qp}. \textbf{Left column:}
  $(s, N, \delta, \eta) = (1, 7000, 0.1, 100)$; %
  \textbf{Right column:} $(s, N, \delta, \eta) = (1, 7000, 0.25, 100)$.}
  \label{fig:bp-numerics-2}
\end{figure*}

\subsection{More synthetic examples}
\label{sec:more-synth-exampl}

In this section, we display three synthetic examples where only $s$ and $\eta$
were changed. Thus, the effect of sparsity and noise scale is readily
observed. In each of these figures, the aspect ratio of the measurement matrix
was $\delta = 0.25, 0.45$ for the left- and right-hand plots respectively
(except \autoref{fig:synthetic-example-1} where $\delta = 0.25$ was too small to
achieve recovery). The average loss curves for each program were computed from
$k = 25$ realizations of loss curves that were, themselves, generated on a
logarthmically spaced grid of $n = 201$ points centered about the optimal choice
of the normalized parameter, $\rho = 1$. The loss realizations were again
computed by solving {\qp} and using the correspondence between \textsc{Lasso}
programs to compute the loss curves for {\ls} and {\bp}.

The bottom row of \autoref{fig:synthetic-example-0} and
\autoref{fig:synthetic-example-2}, and the right half of
\autoref{fig:synthetic-example-1} depict the quality of the approximation of the
average loss curve for each program. Specifically, each program appears with its
own facet, in which are displayed the individual loss realizations
$L(\rho_{i}; x_{0}, A, \eta z_{j}), i \in [n], j \in [k]$ as grey points, and
the average loss $\bar L(\rho; x_{0}, A, \eta)$ as a coloured line. The top row
of \autoref{fig:synthetic-example-0} and \autoref{fig:synthetic-example-2}, and
the left half of \autoref{fig:synthetic-example-1} compare the the average loss
curves for each program, where the average losses are plotted on a log-log scale
with respect to the normalized parameter.

The first figure, \autoref{fig:synthetic-example-0}, displays a setting similar
to \autoref{fig:low-noise-numerics}. The noise scale was $\eta = 10^{-5}$ and
$s = 1$. Thus, the setting depicts the low-noise high-sparsity regime. The
second figure, \autoref{fig:synthetic-example-1}, depicts a moderately low-noise
regime, with a large value of $s$ (so large that $\delta = 0.25$ did not yield
adequate recovery). Thus, this figure depicts the regime in which $x_{0}$ is
near the limit of acceptable sparsity for the CS regime. Finally, the parameter
settings for \autoref{fig:synthetic-example-2} were $s = 100$ and $\eta =
100$. In particular, sparsity is modest, and the noise scale is large (the
variance equals the ambient dimension, $\eta^{2} = N = 10^{4}$).

It is readily observed that {\ls} is highly sensitive to its parameter choice in
both low-noise regimes. We observe that {\bp} becomes more sensitive to its
parameter choice as sparsity decreases from $750$ to $100$ to $1$. Finally, we
observe that {\qp} is most sensitive to its parameter choice in the low-noise
high-sparsity regime. This left-sided sensitivity is consistent with the theory
and numerical simulations for the corresponding proximal denoising setting
in~\cite{berk2019pdparmsens, berk2020sl1mpc}.

\begin{figure*}[h]
  \centering
  \includegraphics[width=.45\textwidth]{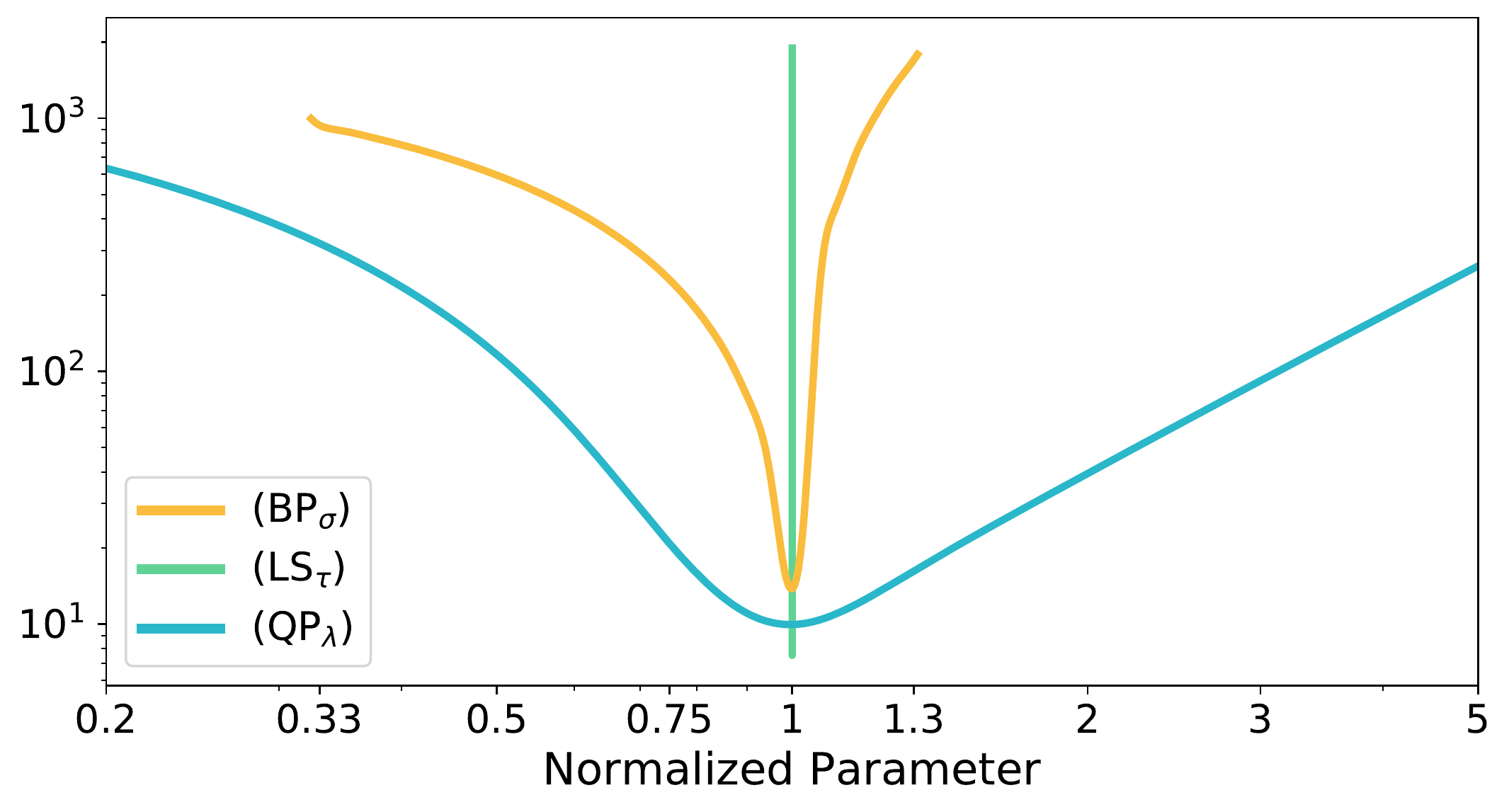}
  \hfill
\includegraphics[width=.45\textwidth]{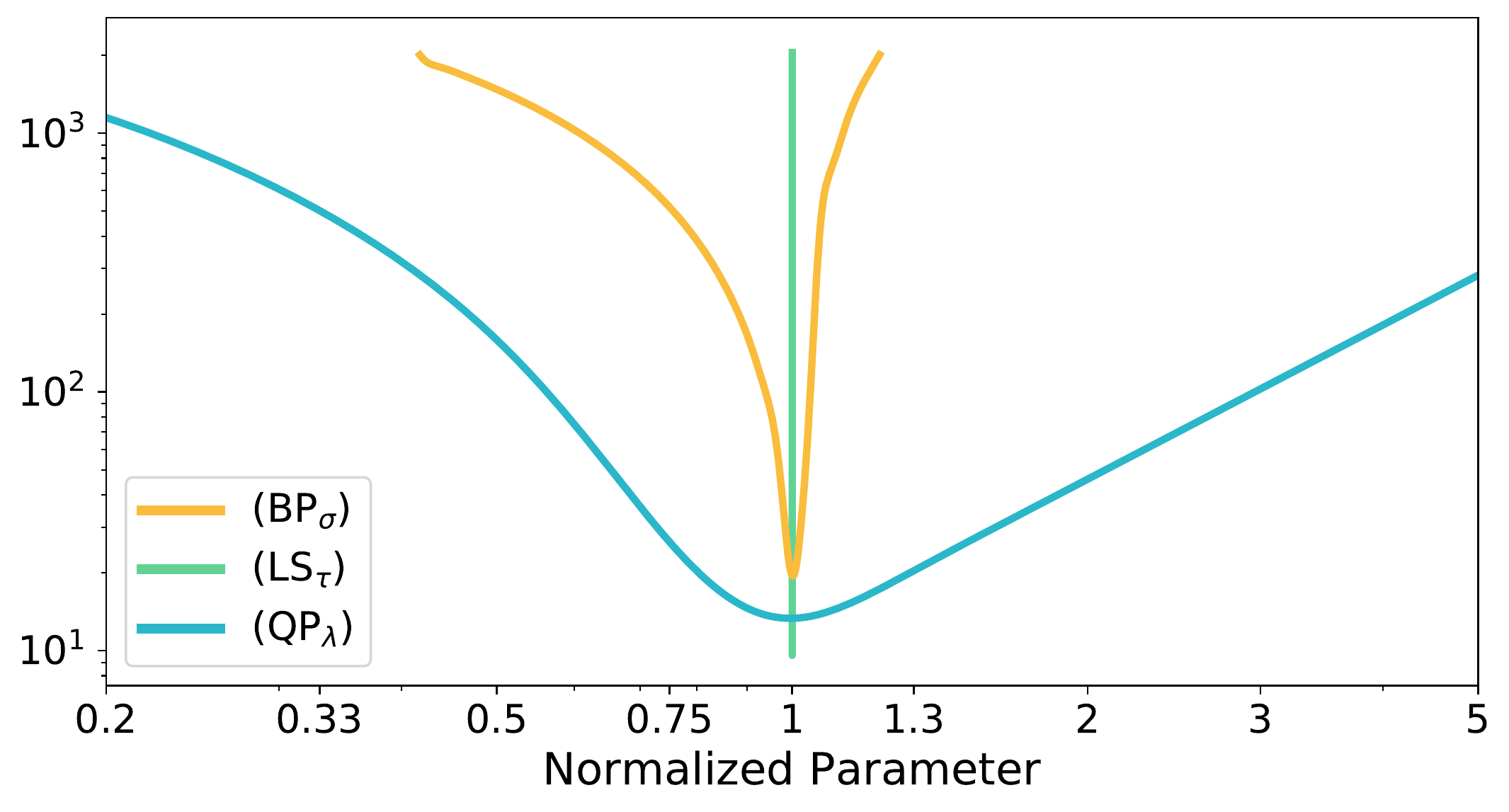}

  \includegraphics[width=.45\textwidth]{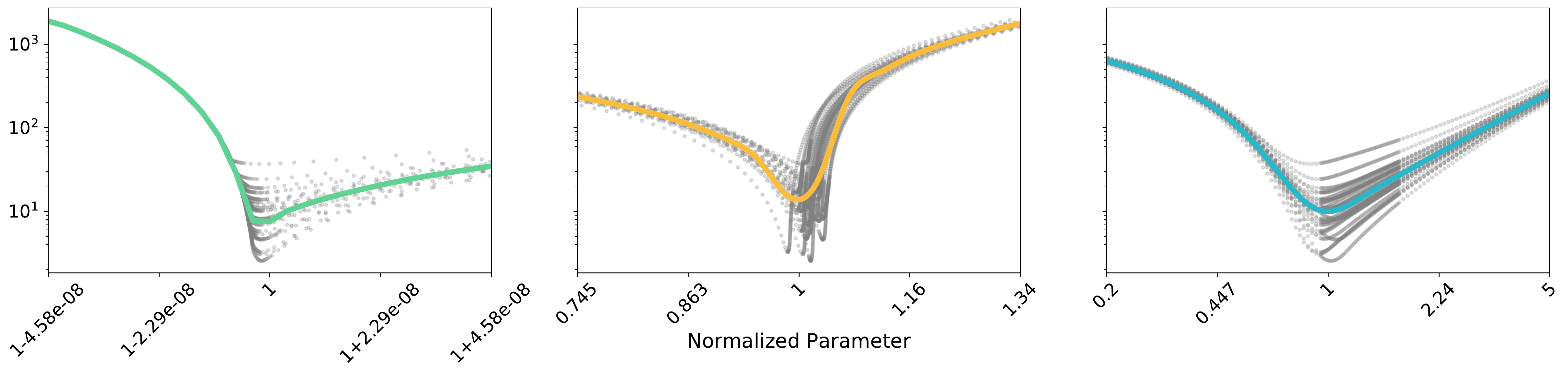}
  \hfill
\includegraphics[width=.45\textwidth]{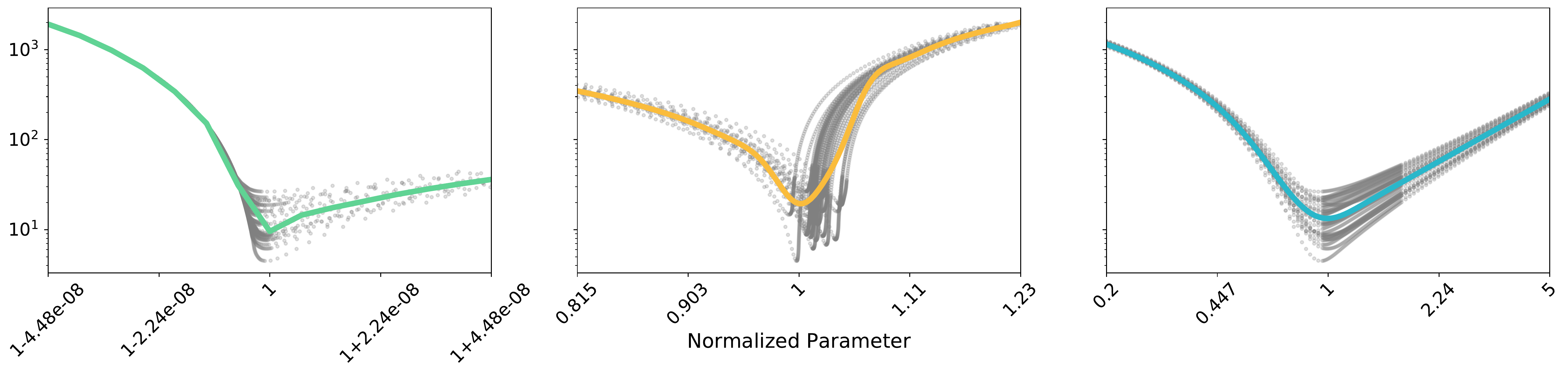}
\caption{Parameter instability numerics in the low-noise, high-sparsity
  regime. \textbf{Top row:} Average loss is plotted with respect to the
  normalized parameter for each program. \textbf{Bottom row:} Visualizations of
  RBF approximation quality for average loss (best seen on a
  computer). \textbf{Left:} %
  $(s, N, m, \eta, k, n) = (1, 10^{4}, 2500, 10^{-5}, 25, 201)$; \textbf{Right:}
  $(s, N, m, \eta, k, n) = (1, 10^{4}, 4500, 10^{-5}, 25, 201)$.}
  \label{fig:synthetic-example-0}
\end{figure*}

\begin{figure}[h]
  \centering
  \includegraphics[width=.44\textwidth]{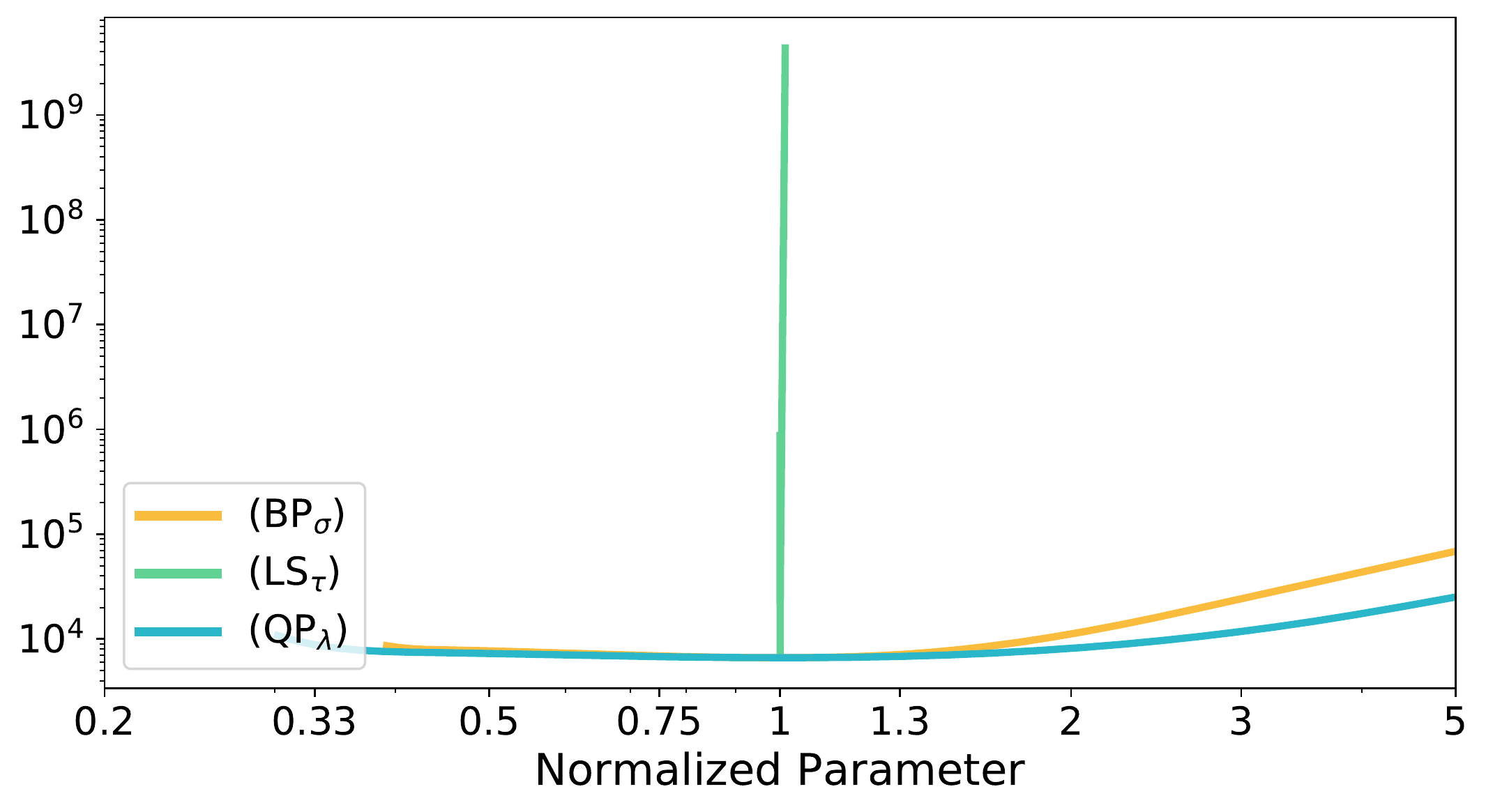}
\hfill  
\includegraphics[width=.46\textwidth]{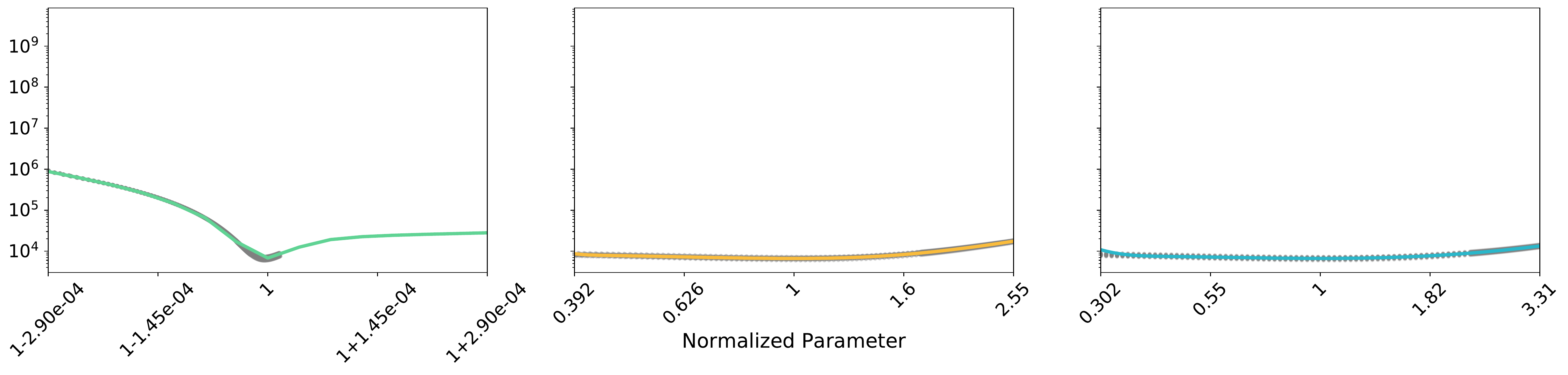}
\caption{Parameter instability numerics in the low-sparsity regime with
  parameters $(s, N, m, \eta, k, n) = (750, 10^{4}, 4500, 10^{-1}, 25,
  201)$. \textbf{Left:} Average loss is plotted with respect to the normalized
  parameter for each program. \textbf{Right:} Visualizations of the RBF
  approximation quality for average loss.}
  \label{fig:synthetic-example-1}
\end{figure}

\begin{figure*}[h]
  \centering
  \includegraphics[width=.45\textwidth]{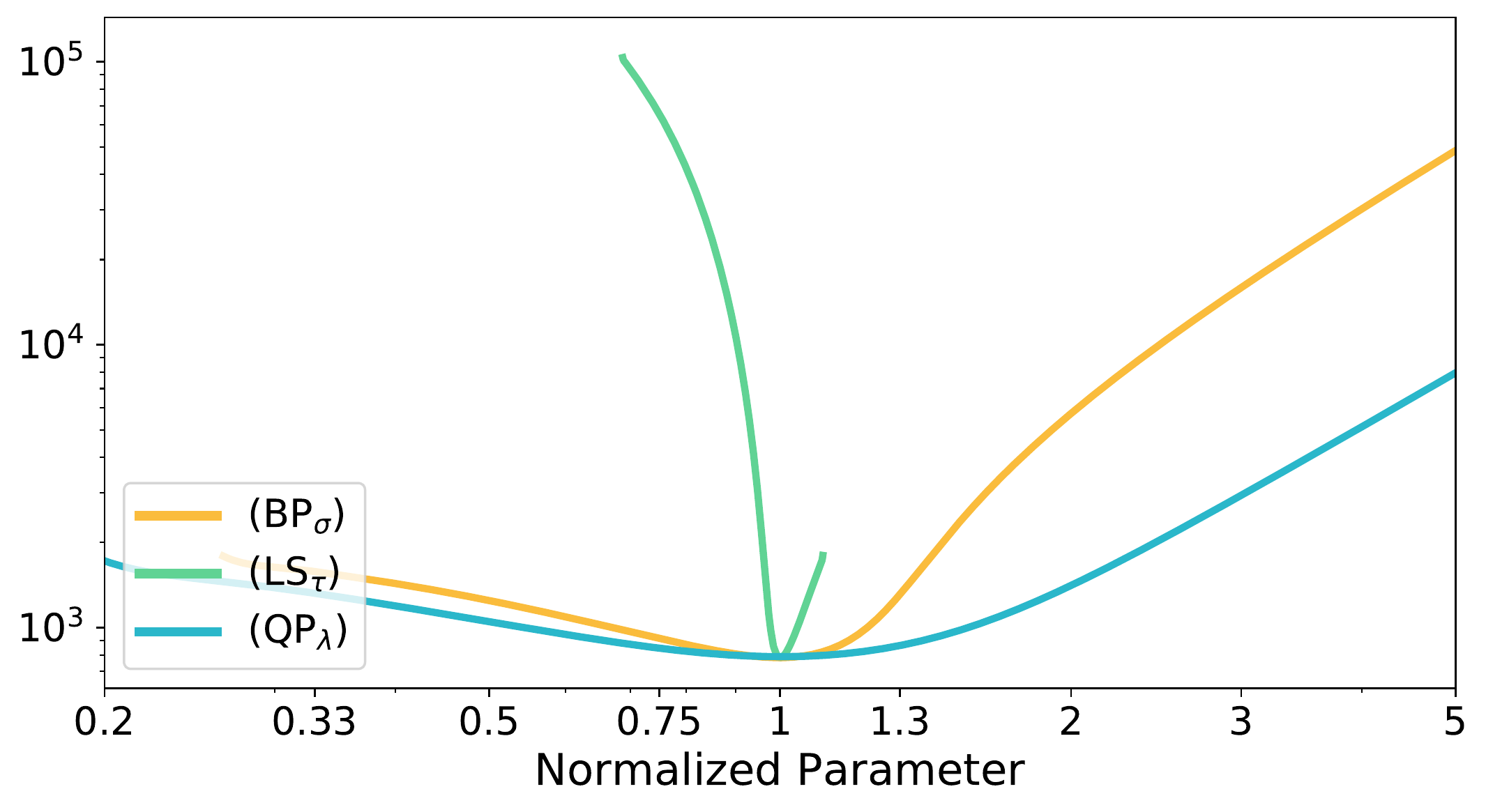}\hfill
\includegraphics[width=.45\textwidth]{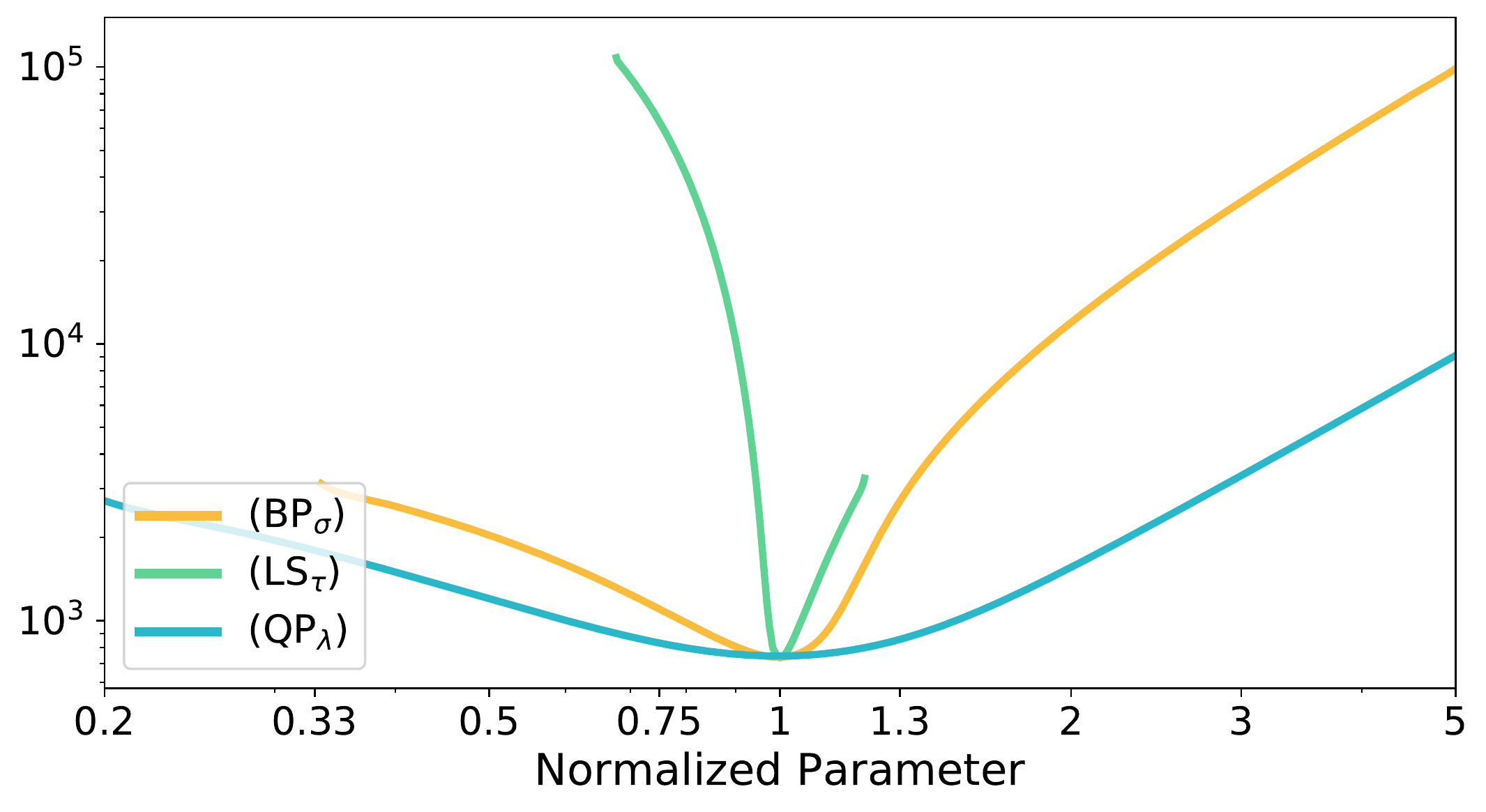}

\includegraphics[width=.45\textwidth]{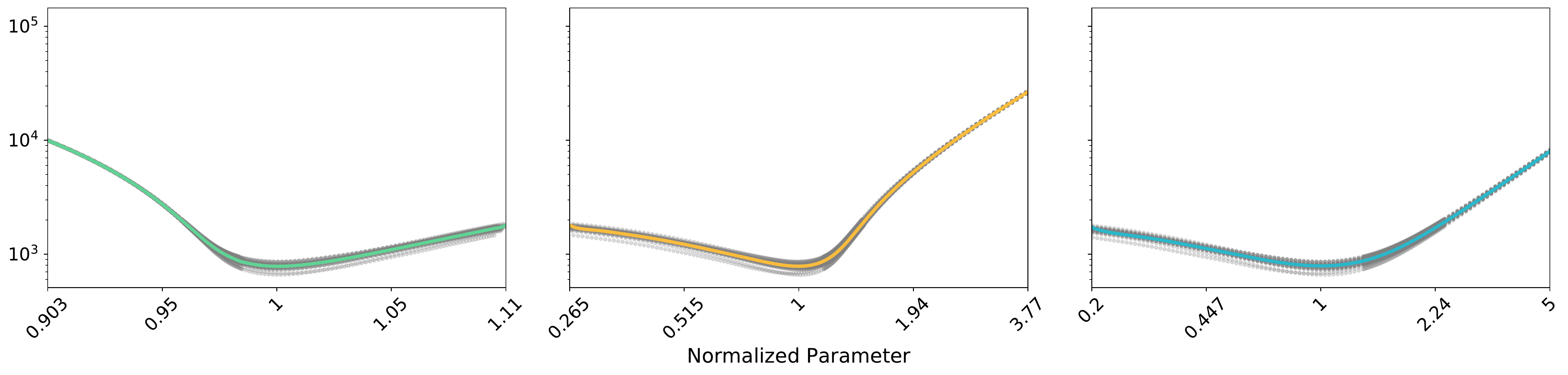}\hfill
\includegraphics[width=.45\textwidth]{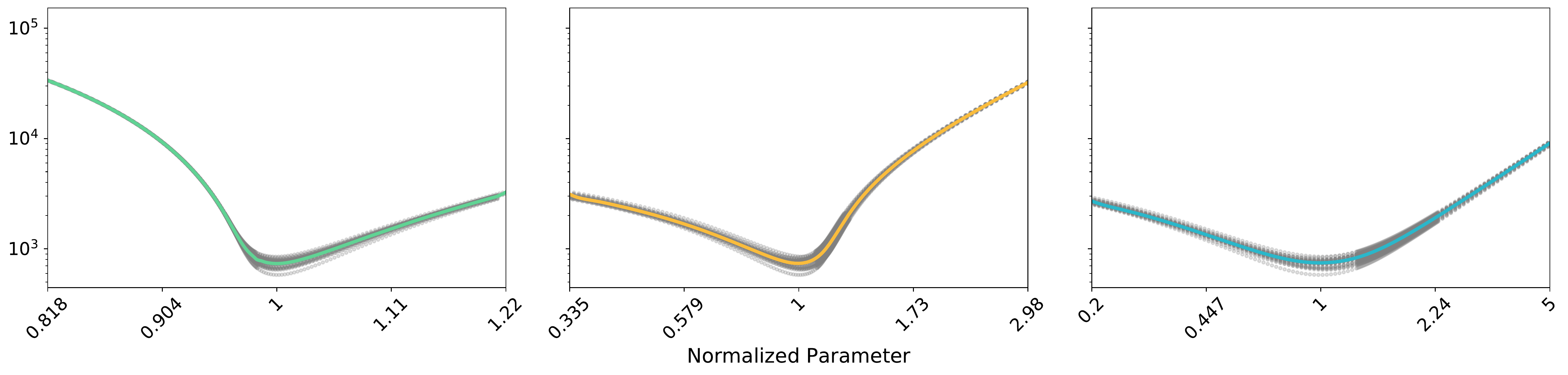}

\caption{Parameter instability numerics for intermediate parameter values:
  $(s, N, \eta, k, n) = (10^{2}, 10^{4}, 10^{-1}, 25, 201)$. \textbf{Left:}
  $m = 2500$; \textbf{Right:} $m = 4500$. \textbf{Top:} Average loss is plotted
  with respect to the normalized parameter for each program. \textbf{Bottom:}
  Visualizations of average loss approximation quality.}
  \label{fig:synthetic-example-2}
\end{figure*}

\subsection{Realistic Examples}
\label{sec:realistic-examples}

We next include two realistic examples in addition to the synethetic ones of the
previous sections. We show how CS programs may exhibit sensitivity as a function
of their governing parameter for a 1D and 2D wavelet problem. In each example,
we will include plots similar to those appearing above; however there will be
some key differences. As above, the average loss is computed from several
realizations of the loss, which depend in turn on realizations of the
noise. However, we will plot the loss corresponding to a single realization as a
function of the normalized parameter. Computing the normalized parameter is what
requires computing the average loss. As before, we approximate the average loss
and the normalized parameter using RBF interpolation, described in
\autoref{sec:rbf-approximation}. The figures of this section contain three main
pieces. We will plot psnr, as a function of the normalized parameter; loss,
equal to the nnse, as a function of the normalized parameter; and we will
include a grid of plots that allows for comparison of CS recovery by visualizing
the recovery in the signal domain. The latter grid of plots shall be referred to
as ``grid plots'' while the psnr and nnse plots shall be referred to as
``reference plots'', as they contain annotations that relate them to the grid
plots.

We now include a brief description of the so-called grid plots and associated
reference plots that appear in this section. Other than plotting loss, rather
than average loss, a key difference of the reference plots to the plots of
\autoref{sec:ls-numerics}--\ref{sec:bp-numerics} is that they have been
annotated with vertical black dashed lines, and coloured dots. Where the loss
for a program intersects the black dashed line, we show a representative
solution for that program where the normalized parameter for the problem is
given by the x intercept of the vertical line (approximately). Because the
programs were solved on a grid, the true value of the normalized program is
given by the coloured dot appearing nearest the black dashed line. The $x$ axis
for the reference plots is the normalized parameter (plotted on a log
scale). The $y$ axis for the reference plots is either the psnr (plotted on a
linear scale) or the nnse (plotted on a log scale). The representatives for each
chosen normalized parameter value and each program are plotted as a faceted grid
below the reference plot. The chosen normalized parameter value is given at the
top of each column, while the program used to recover the noisy ground truth
signal is described in the legend.

\subsubsection{1D wavelet compressed sensing}
\label{sec:1d-wavel-compr}

The signal $\xi_{0} \in \reals^{N}$, $N = 4096$, was constructed in the
Haar-wavelet domain. In particular, $x_{0} \in \reals^{N}$ has $10$ non-zero
coefficients, each equal to $N$. Let $\mathcal{W}_{1}$ denote the $1$D Haar
wavelet transform. Thus, $\xi_{0} = \mathcal{W}_{1}x_{0}$ where
$\|x_{0}\|_{0} = 10$. Next, for $A \in \reals^{m \times N}$ where $m = 1843$,
define $y = A x_{0} + \eta z$ where $z_{i} \iid \mathcal{N}(0, 1)$ and
$\eta = 50$. The signal's wavelet coefficients were recovered using {\ls}, {\qp}
and {\bp} for several realizations of the noise $z$ and over a grid of
normalized parameter values. For example $\hat x(\tau_{i}; x_{0}, A, z^{(j)})$
is the {\ls} recovery of the wavelet coefficients $x_{0}$ from $(y^{(j)}, A)$
with $\tau = \tau_{i}$, where $y^{(j)} := Ax_{0} + \eta z^{(j)}$. The recovered
signal is thus given by
$\hat \xi(\tau_{i}) := \mathcal{W}^{-1} \hat x(\tau_{i})$ and the loss given by
$\eta^{-2} \| \hat \xi(\tau_{i}) - \xi_{0}\|_{2}^{2}$. The loss modified
similarly for the other programs. Specifically, the loss is measured in the
signal domain and not the wavelet domain. The average loss was approximated from
$k = 50$ loss realizations using RBF interpolation, as described in
\autoref{sec:rbf-approximation}, on a grid of $n = 501$ points logarithmically
spaced and centered about $\rho = 1$.

Results of this simulation are depicted in \autoref{fig:lasso-realistic-1d} with
RBF interpolation parameter settings given in
\autoref{tab:lasso-realistic-1d}. The results shown in
\autoref{fig:lasso-realistic-1d} depict data from only a single noise
realization: the top-most graphic shows psnr as a function of the normalized
parameter; the middle graphic plots the loss as a function of the normalized
parameter; and the bottom group compares the ground-truth signal and recovered
signals in the signal domain. While psnr and loss are plotted instead of average
psnr and average loss, the normalized parameter was computed from the average
loss, as usual (\emph{cf.} \autoref{sec:rbf-approximation}). Correspondingly,
observe that the optimal parameter choice for each program may not appear at
$\rho = 1$, since the optimal normalized parameter for a particular loss
realization is not necessarily equal to the optimal normalized parameter for the
expected loss. In the bottom group of 15 plots, each row corresponds with a
particular program --- {\ls}, {\qp} and {\bp}, from top to bottom --- while each
column corresponds with a particular vaue of the normalized parameter --- 0.5,
0.75, 1, 1.3 and 2, from left to right. The recovered signal for that program
and normalized parameter value is shown as a coloured line, while the ground
truth signal is shown as a black line.

As $\eta = 50$, the problem lies outside of the small-noise regime. As such,
{\bp} is more sensitive to its parameter choice than {\qp}, and more sensitive
than {\ls} for $\rho > 1$ due to the relatively high sparsity of the
signal. Since suboptimality of {\bp} is observed for risk or average loss rather
than for individual loss realizations, we do not observe suboptimality of {\bp}
loss in these graphics. As expected, {\ls} is sensitive to its parameter choice
for $\rho < 1$, as the ground-truth solution lies outside the feasible set in
this setting. It appears that the loss is mildly more sensitive to
under-guessing $\tau$ in this regime, than it is to over-guessing $\sigma$. This
is readily observed from all of the plots in the figure, especially by comparing
those in the bottom group of 15.

For comparison with the middle plot of \autoref{fig:lasso-realistic-1d}, we
include a plot of the average loss for each program as a function of the
normalized parameter in \autoref{fig:lasso-realistic-1d-avg-loss} (left
plot). Beside it is a triptych visualizing the RBF approximation quality for the
average loss.

\begin{figure*}[h]
  \centering
\includegraphics[width=.8\textwidth]{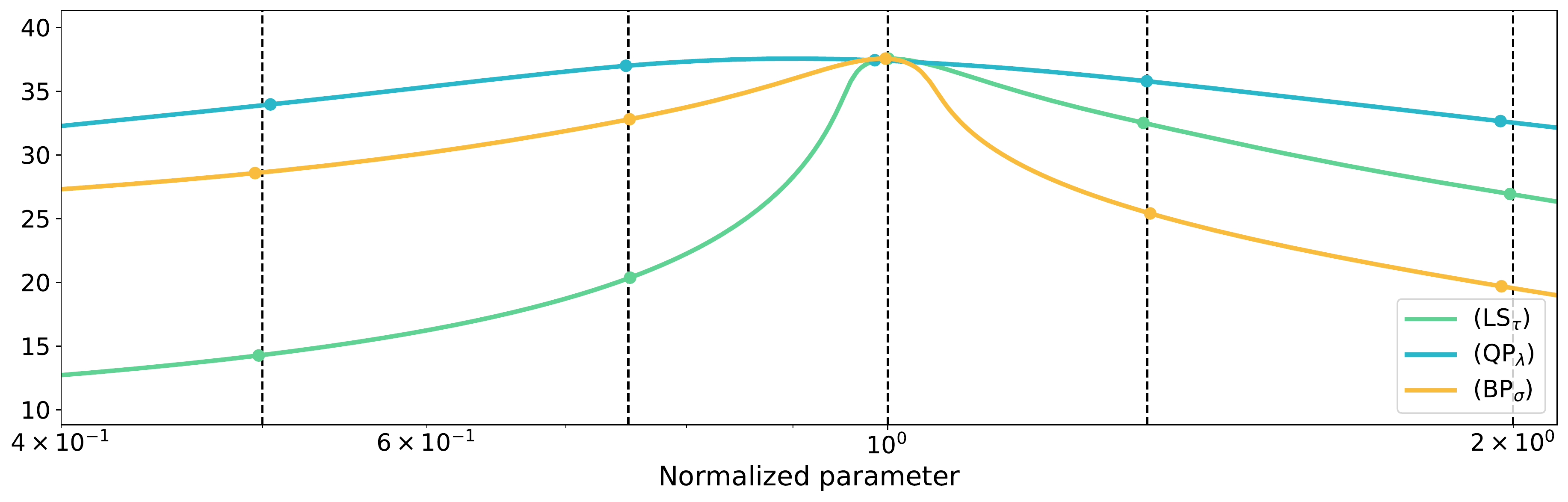}
\includegraphics[width=.8\textwidth]{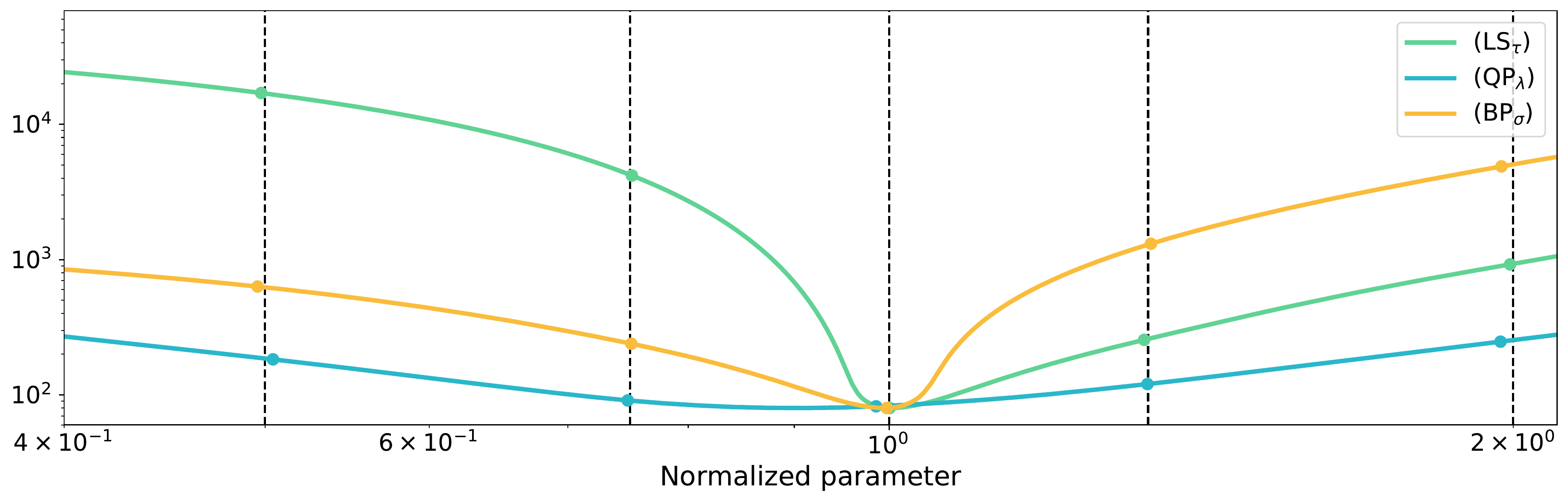}
\includegraphics[width=.8\textwidth]{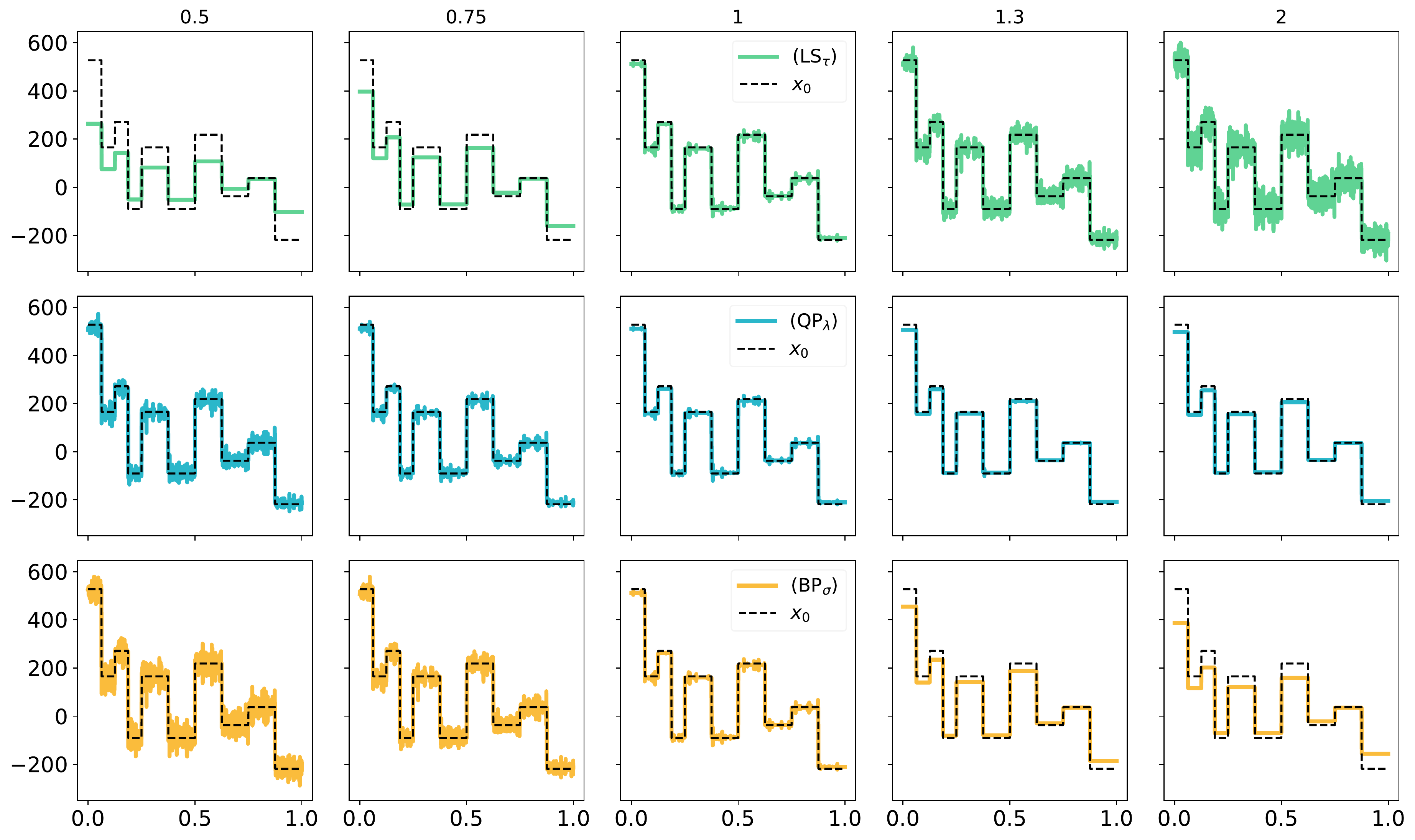}
\caption{Realistic Example in 1D for
  $(s, N, m, \eta, k, n) = (10, 4096, 1843, 50, 50, 501)$. Ground truth signal
  $x_{0}$ defined in the Haar wavelet domain with first $s$ coefficients equal
  to $N$. Noise added in the Haar wavelet domain; recovery error measured in the
  signal domain. \textbf{Top:} Average psnr as a function of the normalized
  parmaeter for each parameter. \textbf{Middle:} Average nnse as a function of
  the normalized parmaeter for each parameter. \textbf{Bottom:} The ground truth
  and recovered signal for a single realization of the noise, faceted by the
  approximate normalized parameter value (given in the title) and by program (as
  depicted in the legend).  }
  \label{fig:lasso-realistic-1d}
\end{figure*}

\begin{figure*}[h]
  \centering
  \includegraphics[width=.3\textwidth]{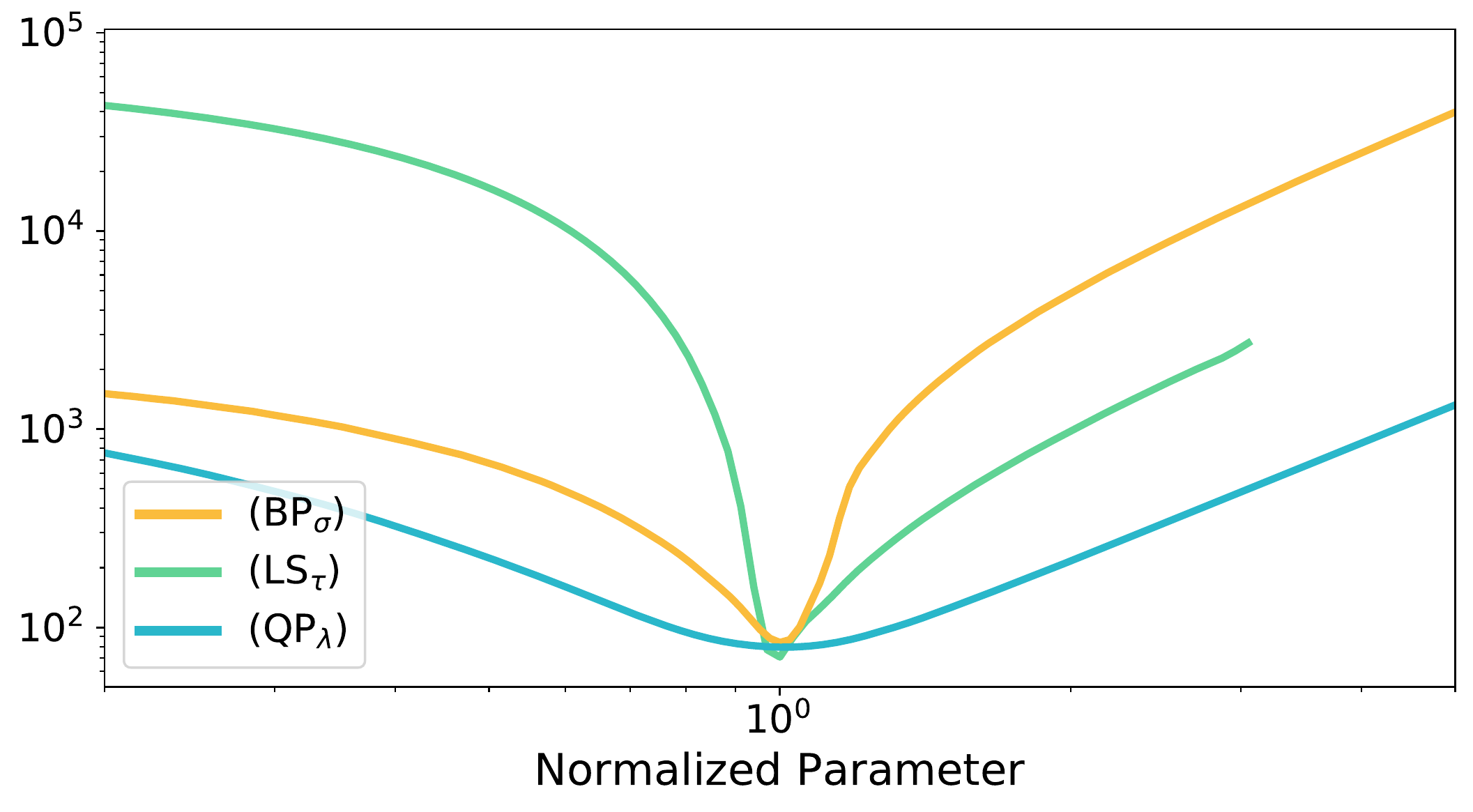}\hfill
  \includegraphics[width=.6\textwidth]{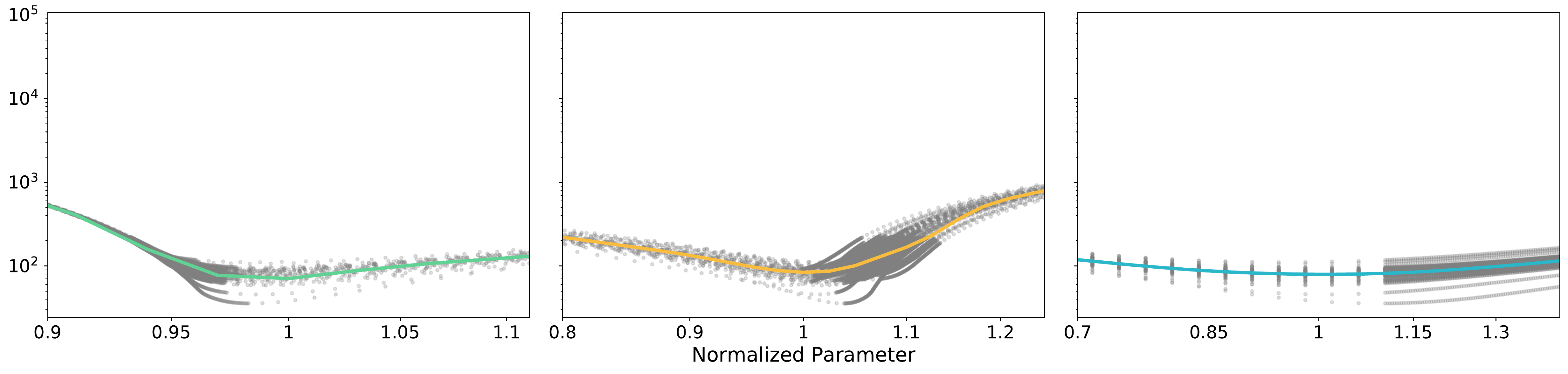}
  
  \caption{Parameter sensitivity numerics for $1$D wavelet CS example with
    parameter settings $(s, N, m, \eta, k n) = (10, 4096, 1843, 50, 50,
    501)$. \textbf{Left:} Average loss for each program plotted with respect to
    the normalized parameter. \textbf{Right:} A visualization of approximation
    quality for the average loss: {\ls}, {\bp} and {\qp}, from left to right.}
  \label{fig:lasso-realistic-1d-avg-loss}
\end{figure*}

\subsubsection{2D Wavelet Compressed Sensing}
\label{sec:2d-wavel-compr}

\begin{figure}[h]
  \centering
  \includegraphics[width=.3\textwidth]{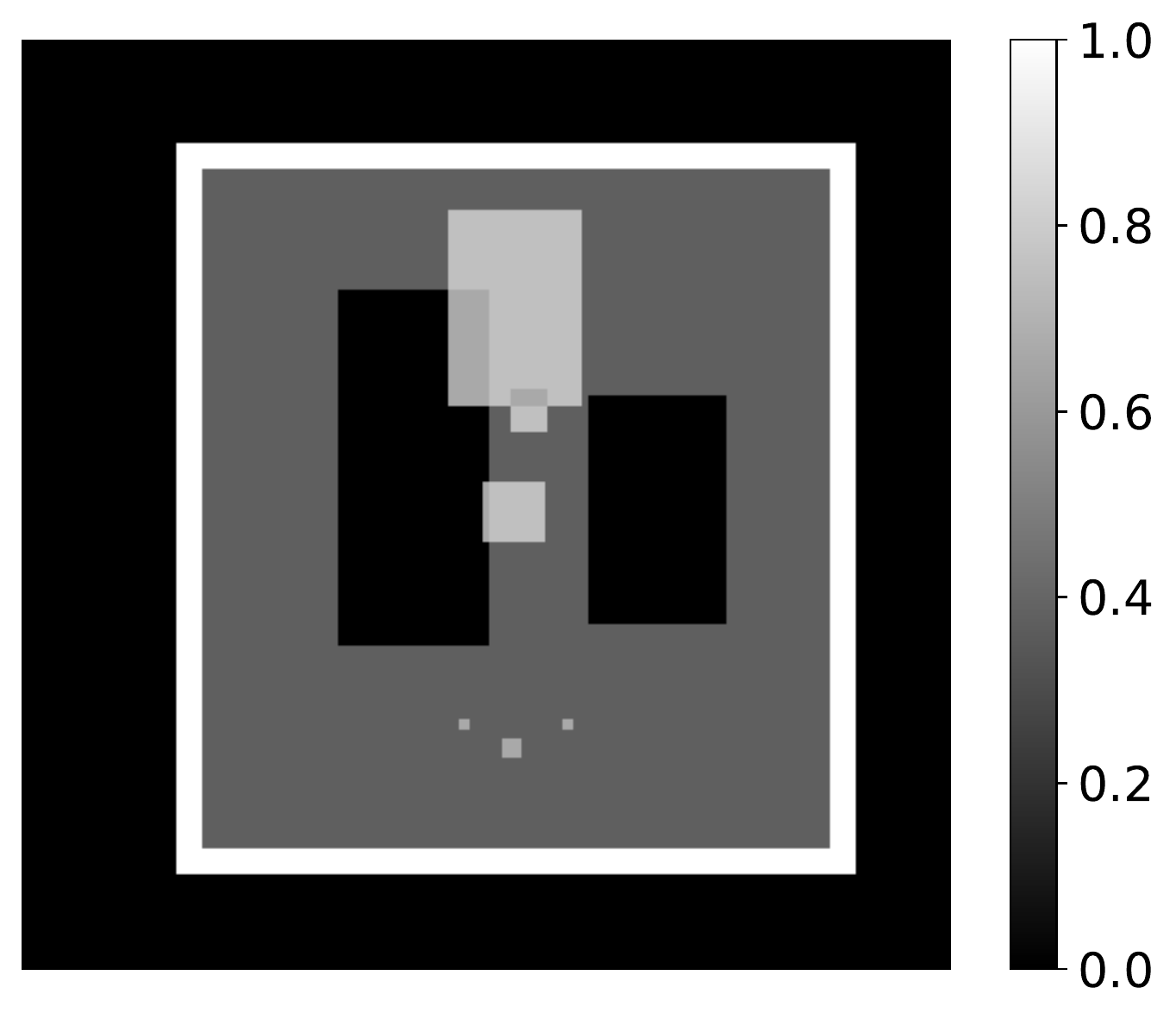}
  \caption{The square Shepp-Logan phantom (sslp).}
  \label{fig:sslp-base}
\end{figure}

In this section, we describe numerical simulations for a 2D wavelet compressed
sensing problem. The signal, $\xi_{0}$, is an $80\times 80$ image of the
so-called square Shepp-Logan phantom (sslp), visualized
in~\autoref{fig:sslp-base}. The sslp was first used
in~\cite{berk2020sl1mpc}. Let $\mathcal{W}$ denote the Haar wavelet transform
and define $x_{0} := \mathcal{W}\xi_{0} \in \reals^{6400}$ to be the vector of
Haar wavelet coefficients for the signal $\xi_{0}$. The linear measurements are
taken as
\begin{align*}
  y = A x_{0} + \eta z, \qquad %
  A_{ij} \iid \mathcal{N}\left(0, m^{-1}\right), z_{i} \iid \mathcal{N}(0, 1), \eta > 0. 
\end{align*}
The signal's wavelet coefficients were recovered using {\qp} to obtain
$x^{\sharp}(\lambda_{i})$, where $i \in [n]$ enumerates the grid of parameter
values. The recovered image is then given by
$\xi^{\sharp}(\lambda_{i}) := \mathcal{W}^{-1}(x^{\sharp}(\lambda_{i}))$. By
using the method described previously at the beginning of
\autoref{sec:numerical-results}, the corresponding solutions for {\ls} and {\bp}
were computed, obtaining $\hat x(\tau_{i})$ and $\tilde x(\sigma_{i})$,
respectively, in addition to the corresponding images
$\hat \xi(\tau_{i}) := \mathcal{W}^{-1}(\hat x (\tau_{i}))$ and
$\tilde \xi(\sigma_{i}) := \mathcal{W}^{-1}(\tilde x (\sigma_{i}))$,
$i \in [n]$. As in \autoref{sec:1d-wavel-compr}, the loss has been modified to
measure the nnse in the image domain. For example, the {\ls} loss is given as
$\eta^{-2}\|\hat\xi(\tau_{i}) - \xi_{0}\|_{2}^{2}$; similarly for the other two
programs.

Average loss as a function of the normalized parameter $\rho$ is shown in
\autoref{fig:2d-wavelet-nnse-1} for $\eta = 10^{-2}, 1/2$ with $m = 2888$ (\ie
$m / N \approx 0.45$). The average loss was approximated using RBF interpolation
from $k=50$ realizations along a logarithmically spaced grid of $501$ points
centered about $\rho=1$ using the method described in
\autoref{sec:rbf-approximation}. The parameter settings for the RBF
interpolation are provided in \autoref{tab:2d-wvlt-cs-params}. Plots showing the
approximation quality of the RBF interpolation quality are given in the bottom
row of \autoref{fig:2d-wavelet-nnse-1}. In these plots, individual realizations
of the nnse for the recovery are shown as grey points. The RBF interpolant is
given by the coloured line in each plot. The approximation quality is only
visualized for a narrow region about $\rho = 1$. Indeed, the approximation
quality of the RBF interpolant was observed, in every case, to be better away
from $\rho = 1$ than about $\rho = 1$: ensuring good interpolation of the loss
realizations about $\rho = 1$ was observed to be sufficient for ensuring good
interpolation of the average loss over the region of interest,
$\rho \in [10^{-1}, 10^{1}]$.

In the left column of the figure, where $\eta = 10^{-2}$, we observe that {\ls}
is relatively more sensitive to its parameter choice than either {\bp} or
{\qp}. In particular, for this problem, we observe that $\eta = 10^{-2}$ is
sufficient to lie within the low-noise regime. Due to how the solutions for
{\ls} were computed from those for {\qp}, the average loss curve for {\ls} is
not resolved over the full domain for the normalized parameter. This reinforces
how small changes in the normalized parameter value for {\ls} correspond to
relatively much larger changes in the normalized parameter value for {\qp}.

In the right column of the figure, where $\eta = 1/2$, we observe that {\bp} is
relatively more sensitive to its parameter choice than either {\ls} or {\qp}. We
expect this is due to the relatively high sparsity of the signal. Again, the
average loss curve for {\bp} is not resolved over the full plotted domain of the
normalized parameter. This underscores how changes in the governing parmeter for
{\qp} correspond with relatively smaller changes in the governing parameter for
{\bp}. In particular {\bp} is more sensitive to its governing parameter than
{\qp} in the present problem. This observation is supported by the theory of
\autoref{sec:bp-minimax-suboptimality}.

As in previous numerical simulations, the bottom row of
\autoref{fig:2d-wavelet-nnse-1} includes triptyches depicting the average loss
approximation quality for the RBF interpolation of the loss realizations
(\emph{cf.} \autoref{sec:rbf-approximation}).

\begin{figure*}[h]
  \centering
  \includegraphics[width=.45\textwidth]{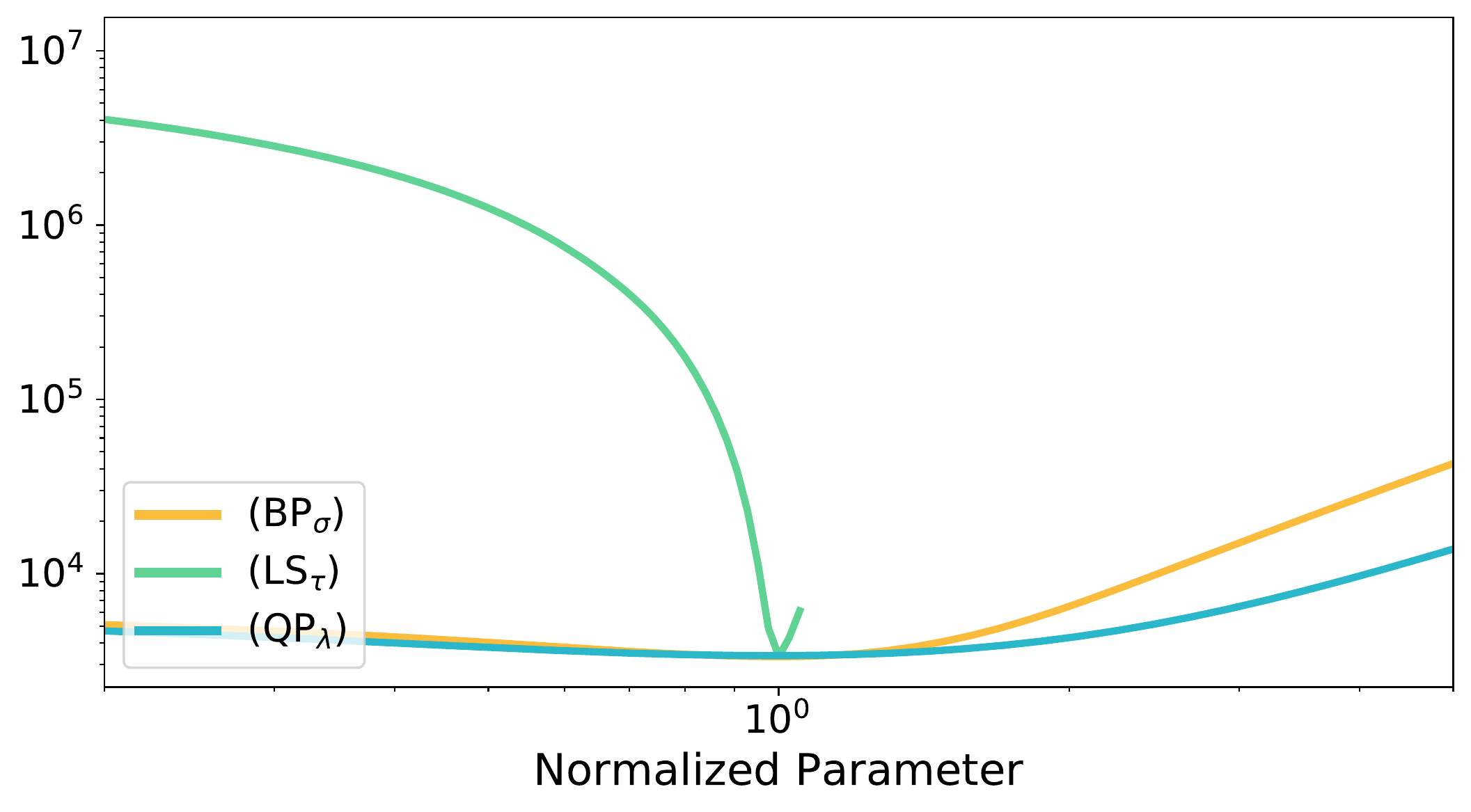}\hfill
  \includegraphics[width=.45\textwidth]{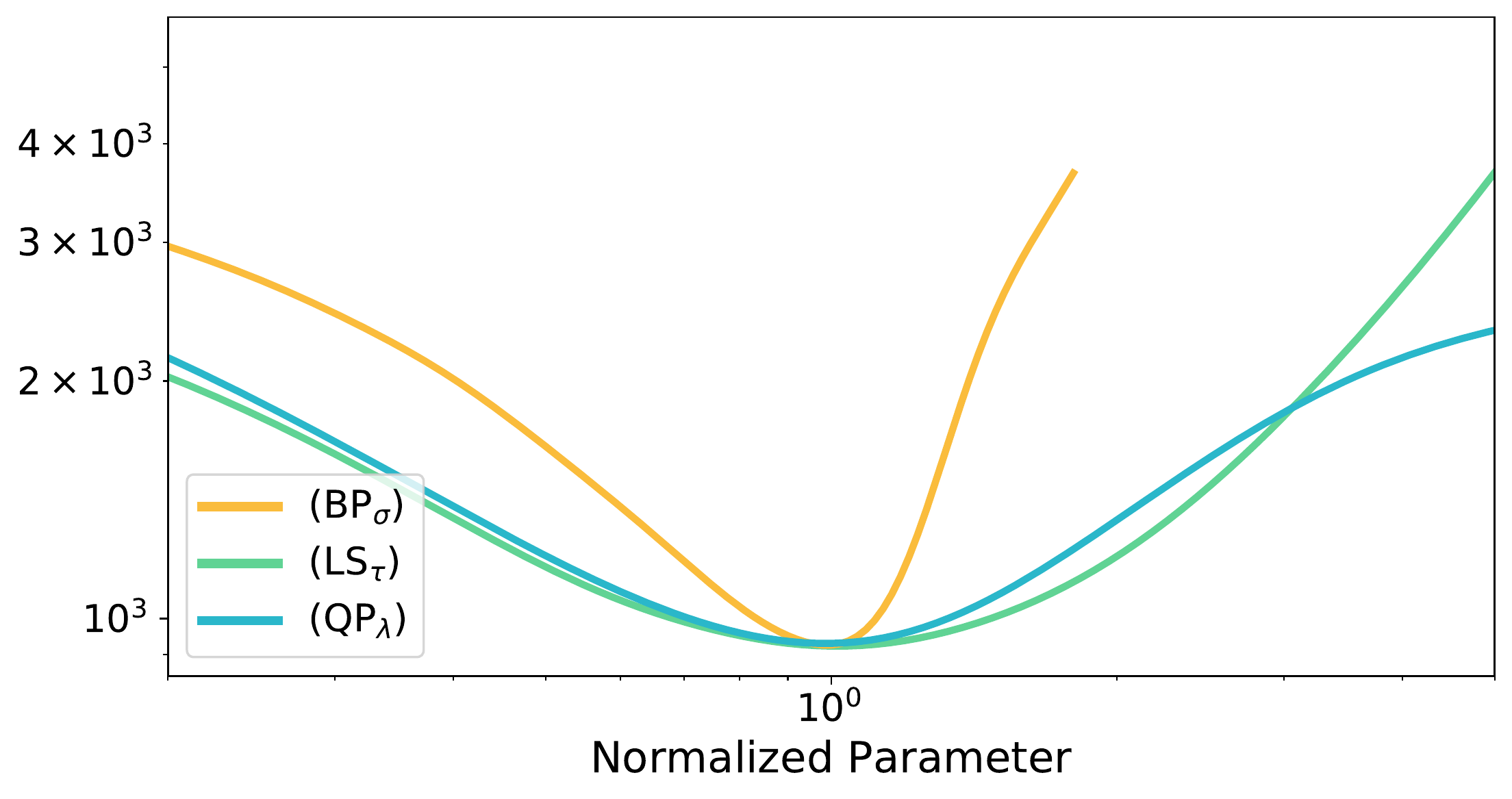}

  \vspace{-6pt}
  \textcolor{lightgrey}{\rule{.45\textwidth}{.75pt}}\hfill
  \textcolor{lightgrey}{\rule{.45\textwidth}{.75pt}}
  \vspace{6pt}

  \includegraphics[width=.45\textwidth]{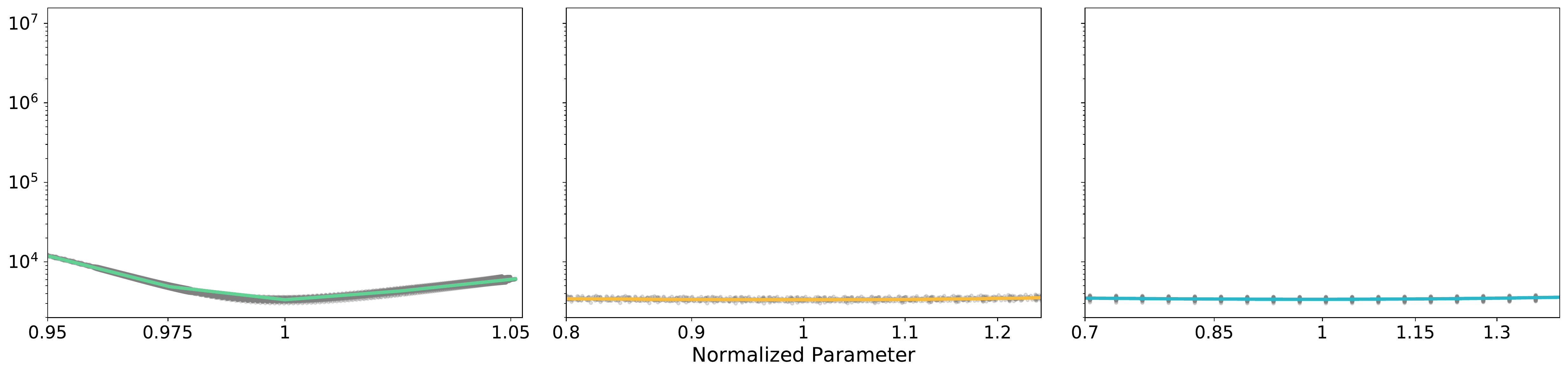}\hfill  \includegraphics[width=.45\textwidth]{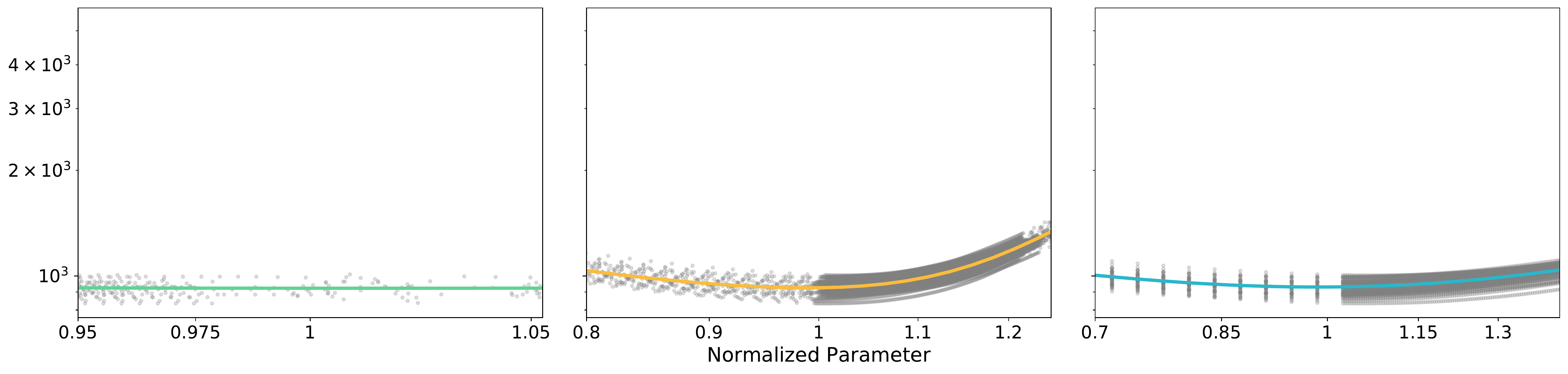}

  \caption{Average loss for a 2D wavelet compressed sensing problem, plotted as
    a function of $\rho$; $(s, N, m) = (416, 6418, 2888)$ with
    $(k, n) = (50, 501)$. \textbf{Left:} $\eta = 10^{-2}$. \textbf{Right:}
    $\eta = 1/2$. \textbf{Top:} The average loss (\ie nnse) for each program as
    a function of $\rho$. The average loss was approximated using RBF
    interpolation with parameters given
    in~\autoref{tab:2d-wvlt-cs-params}. \textbf{Bottom:} Plots to evaluate the
    quality of the RBF interpolation. In each polot, individual realizations of
    the loss are visible as grey points; the approximation to the average loss
    is visible as the coloured line through those points. }
  \label{fig:2d-wavelet-nnse-1}
\end{figure*}

In both \autoref{fig:realistic-lasso-sslp-1} and
\autoref{fig:realistic-lasso-sslp-2}, we use four main elements to depict the
results of a 2D wavelet CS problem, each for a single realization of the
noise. In each figure, the top row depicts the psnr curves for each program
(left) and loss curves for each program (right). The bottom row of the figure
contains two groupings of the 15 plots each. Each grid of 15 plots is faceted by
program ({\ls}, {\bp} and {\qp}, top-to-bottom) and normalized parameter value
($0.5$, $0.75$, $1$, $1.3$, $2$, left-to-right). Each (program, normalized
parameter) tuple on the left-hand side of the figure corresponds with its
partner on the right-hand side. Specifically, the left-hand grid of images
depicts the recovered image for a given (program, normalized parameter) tuple,
while the corresponding right-hand image depicts the pixel-wise nnse in the
signal domain. The details of these images are best examined on a computer.

In \autoref{fig:realistic-lasso-sslp-1}, the parameter settings are
$(s, N, m, \eta, k, n) = (416, 6418, 2888, 10^{-2}, 50, 501)$. In particular,
$\eta$ lies within the low-noise regime, as observed by the relative sensitivity
of {\ls} to its parameter choice. In~\autoref{fig:realistic-lasso-sslp-2}, the
parameter settings are
$(s, N, m, \eta, k, n) = (416, 6418, 2888, 1/2, 50, 501)$. In particular, the
noise scale is relatively larger. The relatively high sparsity of the signal
causes {\bp} to be relatively more sensitive to its parameter choice than either
{\ls} or {\qp}. These observations are supported by the theory
of~\autoref{sec:LS-instability} and~\autoref{sec:bp-minimax-suboptimality}.

\begin{figure*}[h]
  \centering
  \includegraphics[width=0.45\textwidth]{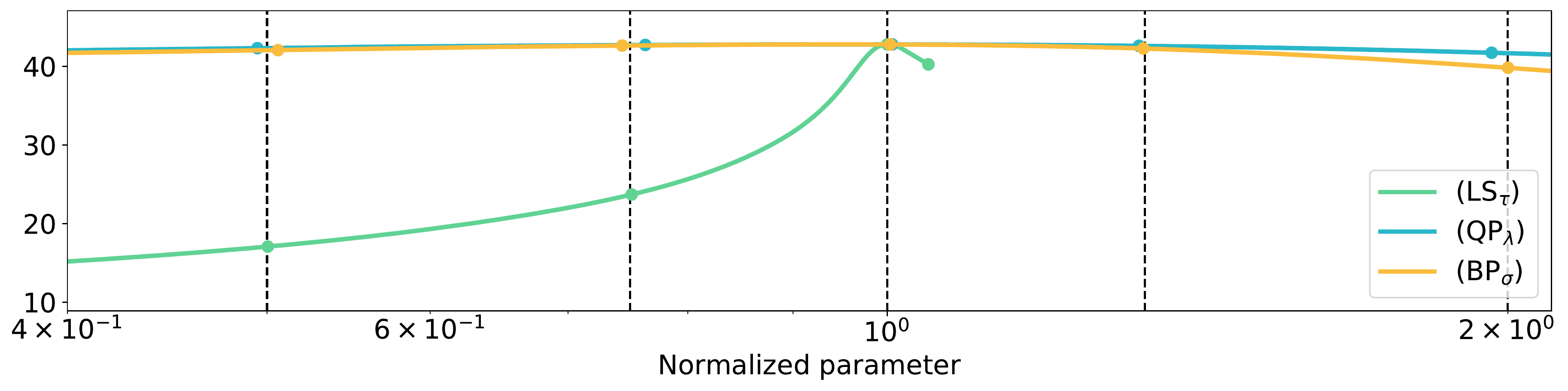}\hfill
  \includegraphics[width=0.45\textwidth]{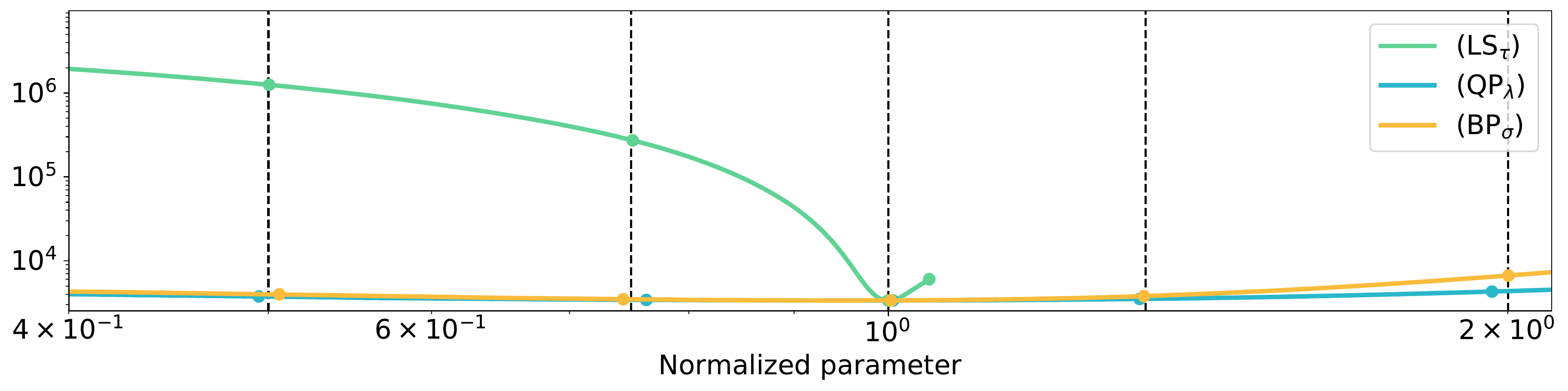}
  

  \includegraphics[width=0.45\textwidth]{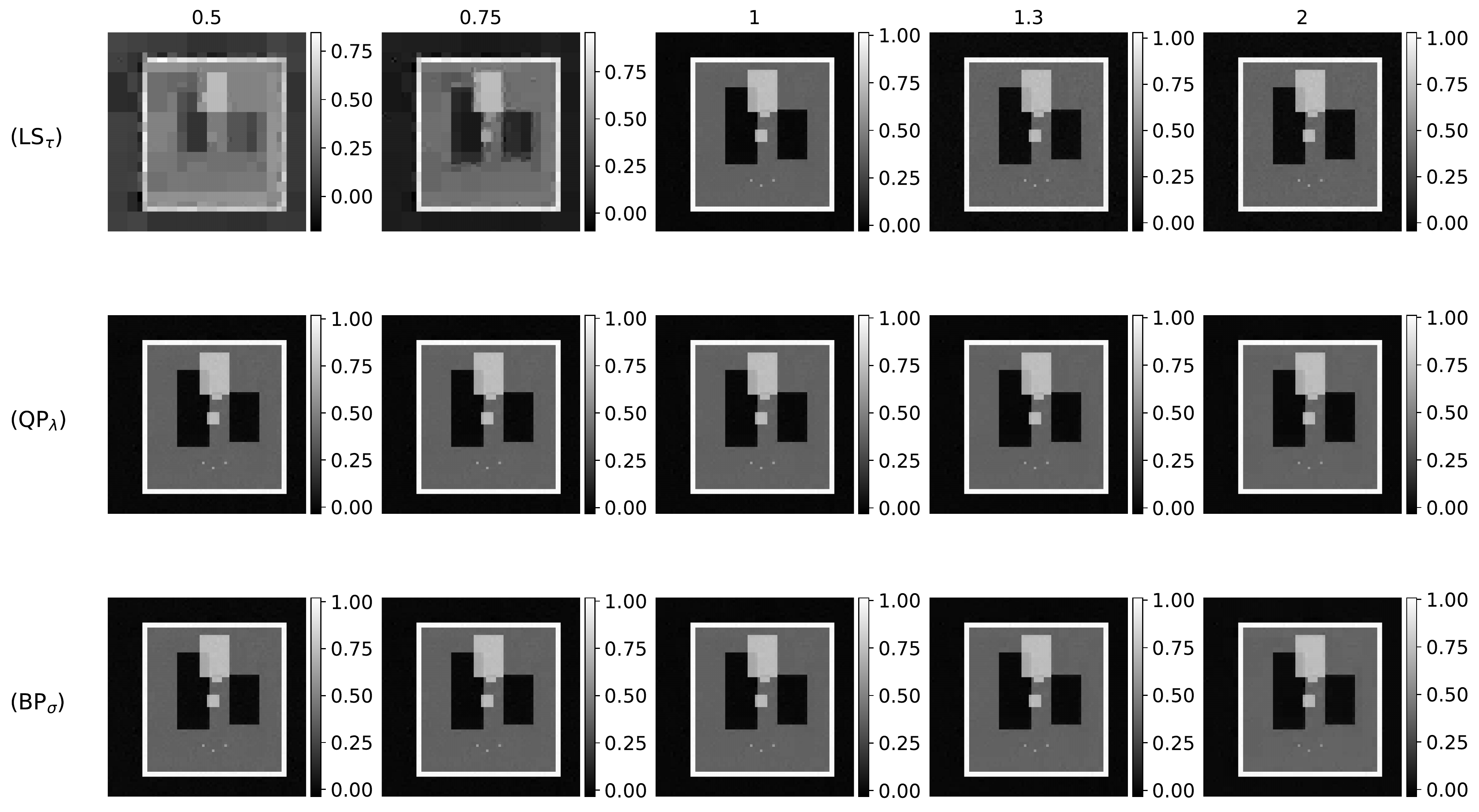}\hfill
  \includegraphics[width=0.45\textwidth]{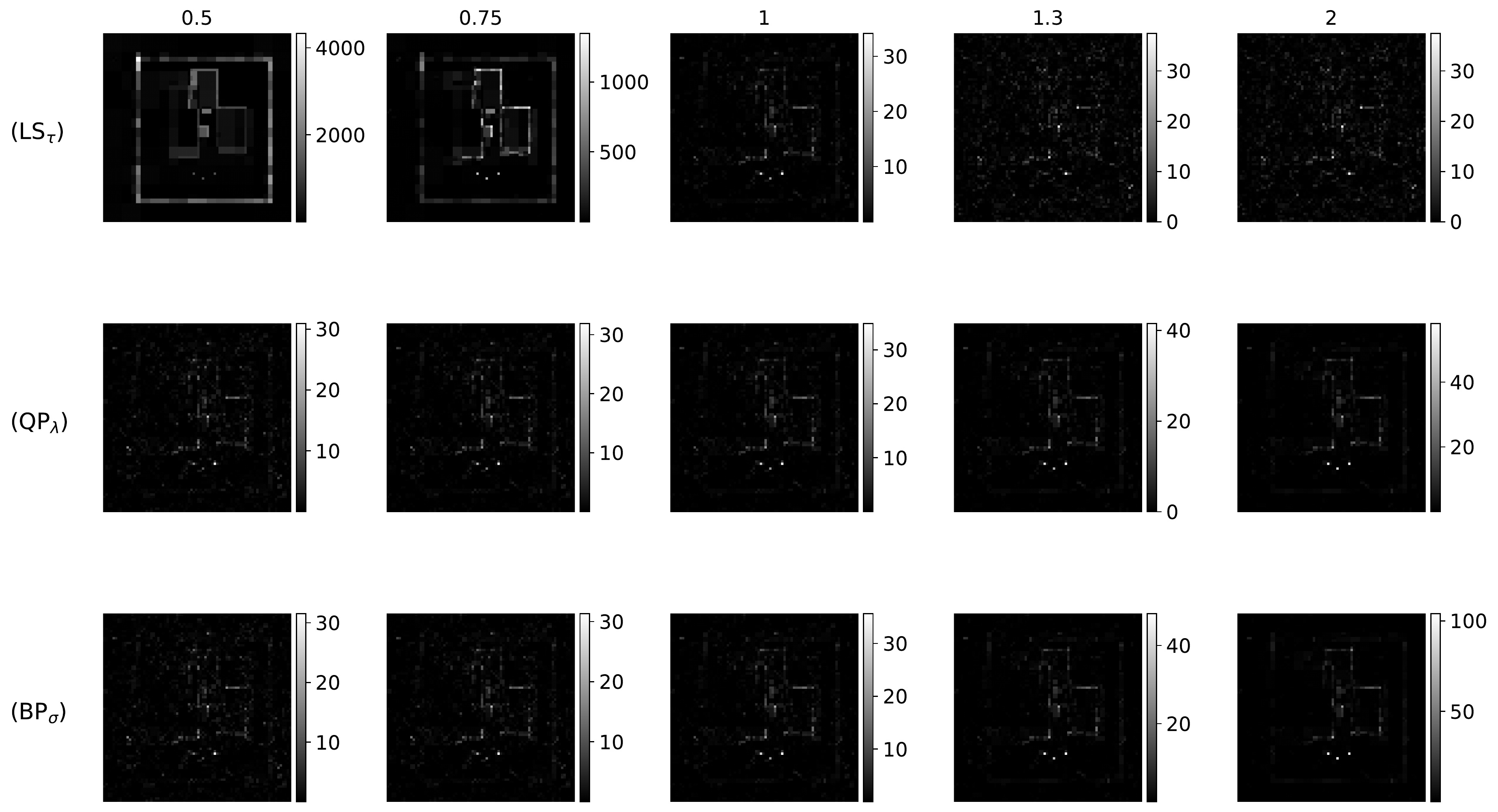}
  \caption{A 2D wavelet compressed sensing problem using the square Shepp-Logan
    phantom; $(s, N, m, \eta) = (416, 6418, 2888, 10^{-2})$ with
    $(k, n) = (50, 501)$. \textbf{Top row:} psnr (left) and nnse (right),
    plotted as a function of $\rho$. The plotted curves were generated from the
    single realization of the measurements that correspond to the grids depicted
    below them. \textbf{Bottom grids:} The left grid of 15 images shows the
    recovered image for each of five values of $\rho$:
    $\rho \in \{ \frac12, \frac34, 1, \frac43, 2\}$; and for each program:
    {\ls}, {\qp}, {\bp}. The right grid of 15 images shows the pixel-wise nnse
    of the recovered image for the same values of $\rho$, and for the three
    programs. Colour bars provide scale, and are best observed on a
    computer. The stated values of $\rho$ are approximate; the values of $\rho$
    for which the images are depicted are marked by points in the nnse and psnr
    plots of the same colour as the loss curve on top of which they're plotted.}
  \label{fig:realistic-lasso-sslp-1}
\end{figure*}

\begin{figure*}[h]
  \centering
  \includegraphics[width=0.45\textwidth]{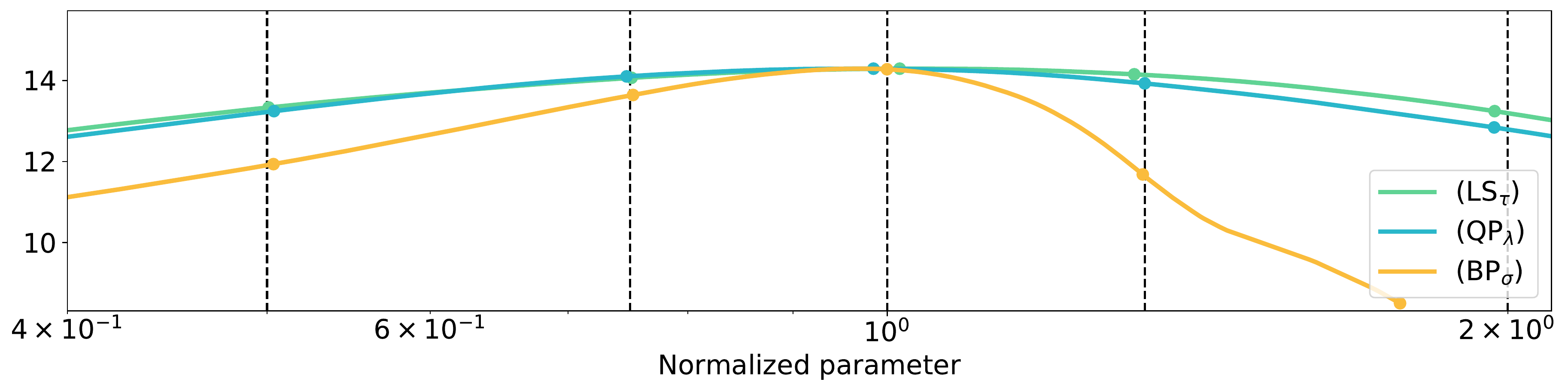}\hfill
  \includegraphics[width=0.45\textwidth]{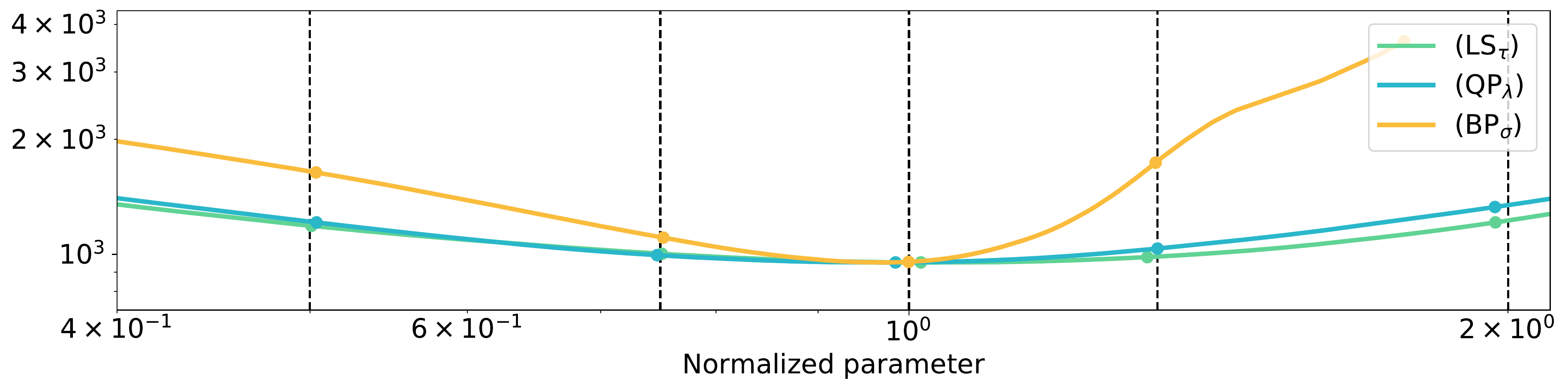}
  

  \includegraphics[width=0.45\textwidth]{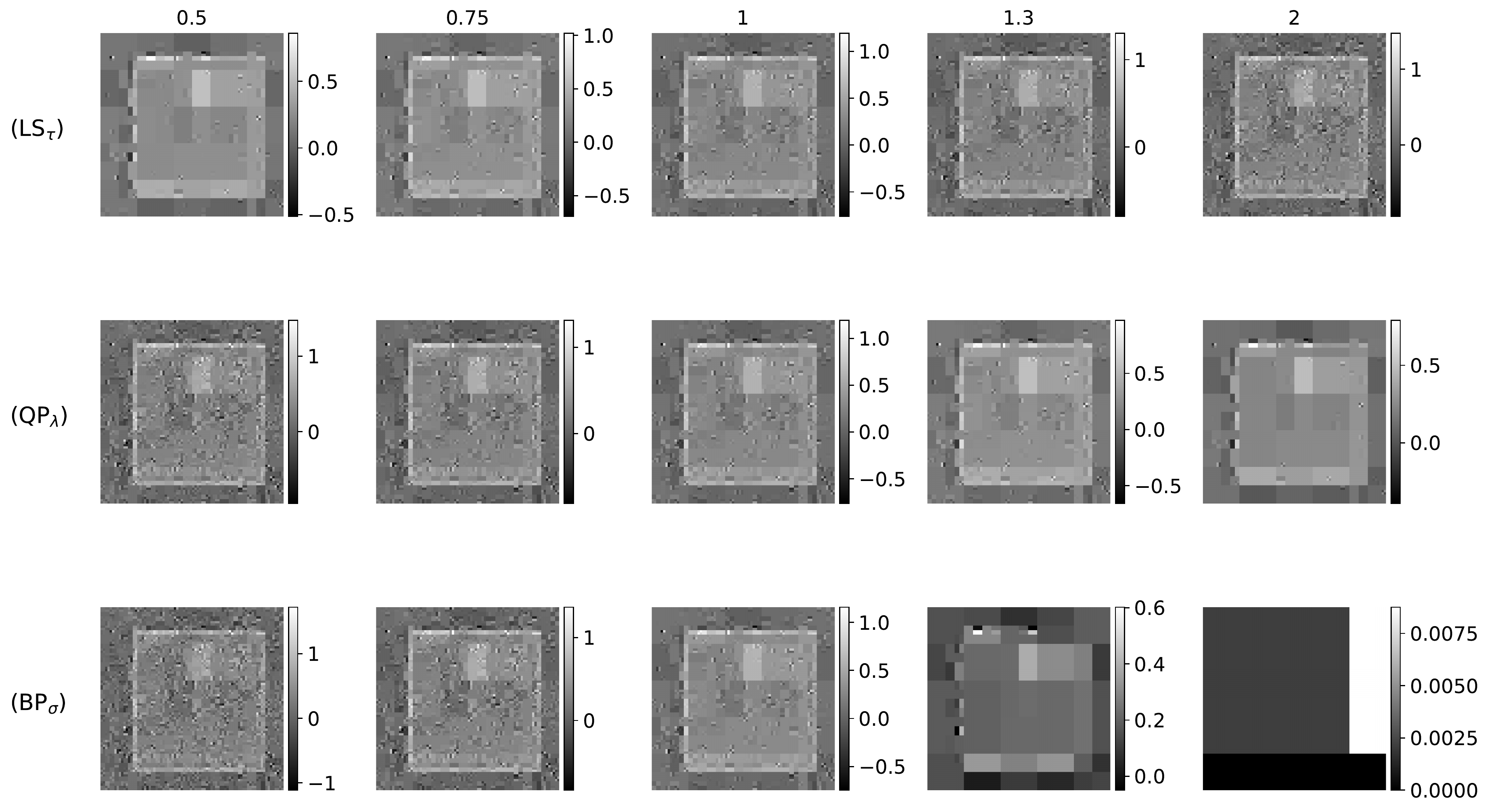}\hfill
  \includegraphics[width=0.45\textwidth]{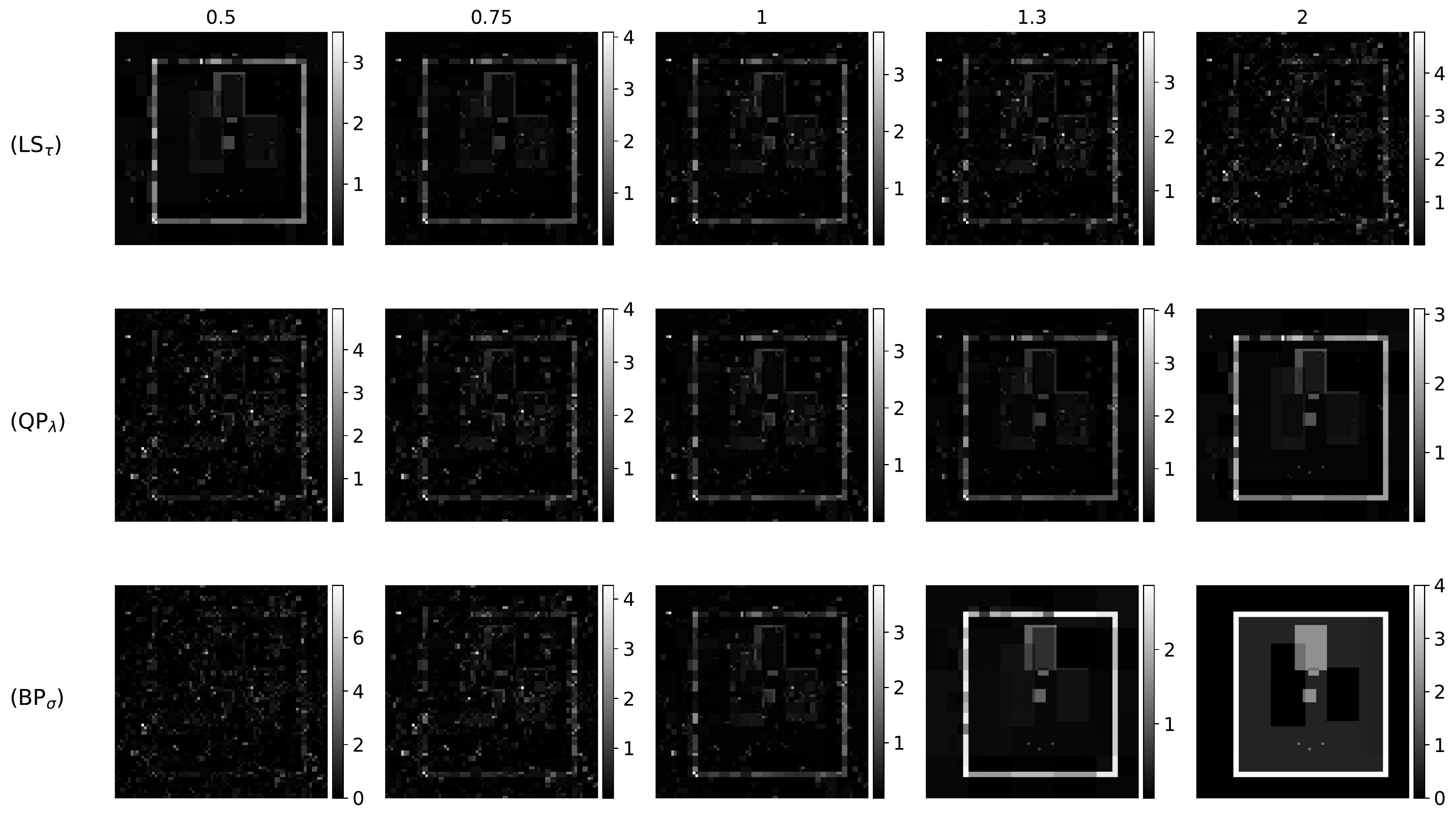}
  \caption{A 2D wavelet compressed sensing problem using the square Shepp-Logan
    phantom; $(s, N, m, \eta) = (416, 6418, 2888, 1/2)$ with
    $(k, n) = (50, 501)$. \textbf{Top row:} psnr (left) and nnse (right),
    plotted as a function of $\rho$ for each program. The plotted curves were
    generated from the single realization of the measurements that correspond to
    the grids depicted below them. \textbf{Bottom grids:} The left grid of 15
    images shows the recovered image for each of five values of $\rho$:
    $\rho \in \{ \frac12, \frac34, 1, \frac43, 2\}$; and for each program:
    {\ls}, {\qp}, {\bp}. The right grid of 15 images shows the pixel-wise nnse
    of the recovered image for the same values of the normalized parameter, and
    for the three programs. Colour bars provide scale, and are best observed on
    a computer. The stated values of $\rho$ are approximate; the values of
    $\rho$ for which the images are depicted are marked by points in the nnse
    and psnr plots of the same colour as the loss curve on top of which they're
    plotted.}
  \label{fig:realistic-lasso-sslp-2}
\end{figure*}

\section{Proofs}
\label{sec:proofs}

\subsection{Risk equivalences}
\label{sec:risk-equivalences}

\begin{proof}[Proof of {\autoref{prop:Rstar-minimax-optimal}}]
  The left-most inequality,
  \begin{align*}
    cs \log (N /s ) %
    \leq M^{*}(s, N), %
  \end{align*}
  is a consequence of~\cite[Theorem 1]{candes2013well}, and the second inequality,
  \begin{align*}
    M^{*}(s, N) %
    \leq  \inf_{\lambda > 0} \sup_{x_{0} \in \Sigma_{s}^{N}} %
    R^{\sharp}(\lambda; x_{0}, A, \eta) %
  \end{align*}
  is trivial. The third inequality,
  \begin{align*}
     \inf_{\lambda > 0} \sup_{x_{0} \in \Sigma_{s}^{N}} %
    R^{\sharp}(\lambda; x_{0}, A, \eta) %
    \leq C_{\delta} R^{*}(s, A)
  \end{align*}
  is a consequence of~\cite[Theorem 6]{bickel2009simultaneous} and
  \autoref{coro:opt-risk-near-equiv}. Indeed, $R^{*}(s, A)$ may be lower-bounded
  by the optimally tuned worst-case risk,
  $\sup_{x \in \Sigma_{s}^{N}} \hat R(\|x\|_{1}; x, A, \eta)$, which is again
  lower-bounded by $c s \log (N /s )$ due to~\cite[Theorem 1]{candes2013well}. In
  particular, selecting constants appropriately gives
  \begin{align*}
    &  \inf_{\lambda > 0} \sup_{x_{0} \in \Sigma_{s}^{N}} %
      R^{\sharp}(\lambda; x_{0}, A, \eta) %
      \\
    &\qquad \leq C_{\delta} s \log(N/s) %
    \leq C_{\delta} M^{*}(s, N) %
    \\
    &\qquad \leq C_{\delta} \sup_{x \in \Sigma_{s}^{N}} \hat R(\|x\|_{1}; x, A, \eta) %
      \leq C_{\delta} R^{*}(s, A).
  \end{align*}
  The final inequality, a variant of which may be found in~\cite{liaw2017simple}
  or~\cite{oymak2013squared}, easily follows
  from~\autoref{lem:ls-instability-ec}.

\end{proof}

\subsection{$\hat R$ is nearly monotone}
\label{sec:rhat-nearly-monotone}

We first quote a specialized version of a result introduced
in~\cite{liaw2017simple}, which gives a kind of local characterization of the
deviation inequality presented in~\autoref{thm:liaw-14}.

\begin{thm}[{\cite[Theorem~1.7]{liaw2017simple}}]
  \label{thm:liaw-17}
  Let $A$ be a normalized $K$-subgaussian matrix and $T\subseteq \reals^{N}$ a
  convex set. For any $t \geq 1$, it holds with probability at least
  $1 - \exp(-t^{2})$ that
  \begin{align*}
    \left| \|Ax\|_{2} - \sqrt m \|x \|_{2}\right| %
    \leq t \cdot C \tilde K \gamma(T \cap \|x\|_{2} B_{2}^{N}), %
    \qquad \text{for all } x \in T. 
  \end{align*}
\end{thm}

We now present the main result of this section.

\begin{prop}[$\hat R$ is nearly monotone]
  \label{prop:Rhat-nearly-monotone}

  Let $A$ be a normalized $K$-subgaussian matrix and
  $\mathcal{K} \subseteq \reals^{N}$ a non-empty closed convex set. Fix
  $\delta, \eta > 0$, $0 < \tau_{1} \leq \tau_{2} < \infty$ and
  $x_{0} \in \mathcal{K}$ with $\|x_{0}\|_{\mathcal{K}} = 1$. For any
  $t \geq 1$, if $m$ satisfies
  $m > Ct^{2}\tilde K^{2}\delta^{-2} \gamma^{2}(T_{\mathcal{K}}(x_{0}) \cap
  \sph^{N-1})$, then with probability at least $1 - \exp(-t^{2})$ on the
  realization of $A$,
  \begin{align*}
    \hat R(\tau_{1}; \tau_{1} x_{0}, A, \eta)
    \leq \frac{1 + \delta}{1 - \delta} \hat R(\tau_{2}; \tau_{2} x_{0}, A, \eta). 
  \end{align*}

\end{prop}

\begin{proof}[Proof of {\autoref{prop:Rhat-nearly-monotone}}]
  Given $0 < \tau_{1} \leq \tau_{2} < \infty$, let
  $\tau \in \{ \tau_{1}, \tau_{2}\}$, define $y(\tau) = A\tau x_{0} + \eta z$,
  and define
  \begin{align*}
    q(\tau) := A\hat w(\tau), %
    \quad \text{where} \quad %
    \hat w(\tau) &:= \hat x(\tau; A, y) - \tau x_{0}, %
    \\
    \tau &\in \{\tau_{1}, \tau_{2}\}. 
  \end{align*}
  Observe that $q(\tau)$ may be written as
  \begin{align*}
    q(\tau) &\in \argmin \{ \|q - \eta z\|_{2} : q \in \tau \mathcal{K}' \}, %
    \\
    \mathcal{K}' &:= \{ A(x - x_{0}) : x \in \mathcal{K} \}.
  \end{align*}
  The set $\mathcal{K}' \subseteq \reals^{m}$ is non-empty, closed and convex,
  with $0 \in \mathcal{K}'$. In particular, \autoref{lem:projection-lemma}
  implies
  \begin{align*}
    \|q(\tau_{1})\|_{2} \leq \|q(\tau_{2})\|_{2}. 
  \end{align*}
  By~\cite[Theorem 1.7]{liaw2017simple}, for any $t \geq 1$ it holds with
  probability at least $1 - \exp(-t^{2})$ on $A$ that for all
  $w \in T_{\mathcal{K}}(x_{0})$,
  \begin{align*}
    \sqrt m \cdot \left|\|Aw\|_{2} - \|w\|_{2} \right| %
    \leq Ct\tilde K \gamma(T_{\mathcal{K}}(x_{0}) \cap \sph^{N-1}).
  \end{align*}
  Accordingly, since $\hat w(\tau) \in T_{\mathcal{K}}(x_{0})$ for
  $\tau = \tau_{1}, \tau_{2}$, under the assumption on $m$ it holds with
  probability at least $1 - \exp(-t^{2})$
  \begin{align*}
    (1 - \delta) \|\hat w(\tau_{1})\|_{2}  %
    &\leq \|q(\tau_{1})\|_{2} \\
    &\leq \|q(\tau_{2})\|_{2} %
    \leq (1 + \delta) \|\hat w(\tau_{2})\|_{2}.  %
  \end{align*}
  In particular,
  $\|\hat w(\tau_{1})\|_{2} \leq \frac{1+\delta}{1-\delta} \|\hat
  w(\tau_{2})\|_{2}$. As $z$ was arbitrary, the result follows:
  \begin{align*}
    \hat R(\tau_{1}, \tau_{1} x_{0}, A, \eta) %
    \leq \frac{1+\delta}{1-\delta} \hat R(\tau_{2}; \tau_{2}x_{0}, A, \eta).
  \end{align*}

\end{proof}

\begin{coro}
  \label{coro:opt-risk-near-equiv}
  Under the assumptions of \autoref{prop:Rhat-nearly-monotone}, the optimally
  tuned worst-case risk for {\ls} is nearly equivalent to $R^{*}(s, A)$, in the
  sense that
  \begin{align*}
    R^{*}(s, A) %
    \leq \sup_{x \in \Sigma_{s}^{N}} \hat R(\|x\|_{1} ; x, A, \eta) 
    \leq C R^{*}(s, A).
  \end{align*}
\end{coro}

\begin{proof}[Proof of {\autoref{coro:opt-risk-near-equiv}}]
  The $\sup$ defining the optimally tuned worst-case risk may be decoupled as
  \begin{align}
    \label{eq:opt-risk-decouple}
    \sup_{x' \in \Sigma_{s}^{N}} \hat R(\|x'\|_{1}; x', A, \eta) %
    = \sup_{\tau > 0} \sup_{x \in \Sigma_{s}^{N} \cap \sph^{N-1}} \hat R (\tau; \tau x, A, \eta).
  \end{align}
  Applying a standard scaling property gives the relation:
  \begin{align*}
    \hat R(\tau; \tau x, A, \eta) %
    & = \left(\frac{\tau}{\eta}\right)^{2} %
    \E \|\hat x(1; y/\tau, A) - x \|_{2}^{2} %
    \\
    & = \hat R(1; x, A, \eta / \tau). %
  \end{align*}
  The lower bound follows trivially from these two observations. To prove the
  upper bound, we start by observing two facts. First,
  $\Sigma_{s}^{N} \cap \sph^{N-1}$ is compact, so there is $x^{*}(\tau)$
  achieving the supremum over the set $\Sigma_{s}^{N} \cap \sph^{N-1}$ in
  \eqref{eq:opt-risk-decouple}. Next, if the supremum over $\tau > 0$ is
  achieved for $\tau \to \infty$, there is nothing to show, since
  \begin{align*}
    \sup_{\tau > 0} &\sup_{x \in \Sigma_{s}^{N}\cap \sph^{N-1}} \hat R(\tau; \tau x, A, \eta) %
    \\
    &=  \lim_{\tau \to \infty} \sup_{x \in \Sigma_{s}^{N}\cap
      \sph^{N-1}}
      \hat R(\tau; \tau x, A, \eta)
    \\
    &=  \lim_{\tau \to \infty} \sup_{x \in \Sigma_{s}^{N}\cap \partial B_{1}^{N}}
      \hat R(1; x, A, \eta/ \tau)
    \\
    & =  \lim_{\eta \to 0} \sup_{x \in \Sigma_{s}^{N}\cap \partial B_{1}^{N}}
      \hat R(1; x, A, \eta).
  \end{align*}
  Otherwise, the supremum is achieved for some $0 \leq \tau^{*} < \infty$. Let
  $(\tau_{i})_{i \in \ints}$ be an arbitrary bi-infinite monotone sequence with
  $\tau_{i} \xrightarrow{i\to -\infty} \tau^{*}$ and
  $\tau_{i} \xrightarrow{i\to\infty} \infty$. For any $i \leq j$,
  \autoref{prop:Rhat-nearly-monotone} and properties of the supremum give
  \begin{align*}
    \sup_{x \in \Sigma_{s}^{N} \cap \sph^{N-1}} &\hat R(\tau_{i}; \tau_{i} x, A, \eta) %
    \\
    &= \hat R(\tau_{i}; \tau_{i} x^{*}(\tau_{i}), A, \eta) %
    \\
    & \leq C \hat R(\tau_{j}; \tau_{j} x^{*}(\tau_{i}), A, \eta) %
    \\
    &\leq C \hat R(\tau_{j}; \tau_{j} x^{*}(\tau_{j}), A, \eta) %
    \\
    & = C \sup_{x \in \Sigma_{s}^{N}\cap \sph^{N-1}}\hat R(\tau_{j}; \tau_{j} x, A, \eta)
  \end{align*}
  As the above chain of inequalities holds for any pair $i < 0$ and $j > 0$,
  taking $i \to -\infty$ and $j \to \infty$ gives,
  \begin{align*}
    \sup_{\tau > 0} \sup_{x \in \Sigma_{s}^{N} \cap \sph^{N-1}}\hat R(\tau; \tau x, A, \eta)
    &\leq C \sup_{x \in \Sigma_{s}^{N}\cap \sph^{N-1}}%
      \hat R(\tau_{j}; \tau_{j}x, A, \eta) %
    \\
    & \xrightarrow{j\to\infty} C  \liminf_{\tau \to \infty}
      \sup_{x \in \Sigma_{s}^{N} \cap \sph^{N-1}} \hat R(\tau; \tau x, A, \eta). %
  \end{align*}
  Finally, combining the above with an application of the standard scaling property yields
  \begin{align*}
    \sup_{\tau > 0}& \sup_{x \in \Sigma_{s}^{N} \cap \sph^{N-1}}\hat R(\tau; \tau x, A, \eta) %
    \\
    &\leq C  \liminf_{\tau \to \infty} \sup_{x \in \Sigma_{s}^{N} \cap \sph^{N-1}} %
      \hat R(\tau; \tau x, A, \eta) %
    \\
    &= C  \liminf_{\tau \to \infty} \sup_{x \in \Sigma_{s}^{N} \cap \partial B_{1}^{N}} %
      \hat R(1; x, A, \eta / \tau) %
    \\
    & = C  \liminf_{\eta \to 0} \sup_{x \in \Sigma_{s}^{N} \cap \partial B_{1}^{N}} %
      \hat R(1; x, A, \eta) %
  \end{align*}

\end{proof}

\subsubsection{Controlling a conditionally Gaussian process}
\label{sec:control-condl-gp}

Here we ready two technical results that are used to control the error of the
tuned approximation ($\tau = \tau^{*}$) uniformly with respect to the noise
scale $\eta > 0$. First, we specialize a result of~\cite{liaw2017simple}. Next,
with high probability we control in expectation the extreme values of a
conditionally Gaussian process. 

\begin{lem}[Corollary of {\autoref{thm:liaw-14}}]
  \label{lem:corollary-lmpv-thm1.4}
  Fix $\delta, \varepsilon, r > 0$ and let $A \in \reals^{m \times N}$ be a
  normalized $K$-subgaussian matrix. For a constant $C_{\varepsilon} > 0$, if
  \begin{align}
    \label{eq:m-bd-near-isometry}
    m %
    > C_{\varepsilon} \delta^{-2} \tilde K^{2} r^{2} s \log\left(
    \frac{2N}{s}\right),
  \end{align}
  it holds with probability at least $1 - \varepsilon$ on the realization of
  $A$ that
  \begin{align}
    \label{eq:near-isometry}
    \sup_{x \in \mathcal{L}_{s}(r)} \left|\|Ax \|_{2} - \|x\|_{2}\right| %
    < \delta.
  \end{align}
\end{lem}

\begin{proof}[Proof of {\autoref{lem:corollary-lmpv-thm1.4}}]
  If $s = 0$ the result holds trivially. For $s \geq 1$, this lemma is a
  straightforward consequence of \autoref{thm:liaw-14}. Set
  $u := \sqrt{\log(2\varepsilon^{-1})}$. Indeed, by that result, it holds with
  probability at least $1 - \varepsilon$ on the realization of $A$ that
  \begin{align*}
    \sup_{x \in \mathcal{L}_{s}(r)} &\left|\|Ax\|_{2} - \|x\|_{2}\right| %
    \\
    & \leq C_{1} m^{-1/2} \tilde K r \left[\w(\Sigma_{s}^{N} \cap \sph^{N-1})
    + u \right],
  \end{align*}
  where $C_{1}$ is an absolute constant. By
  \autoref{lem:mean-width-sparse-signals}, there is an absolute constant
  $C_{2} > 0$ so that
  \begin{align*}
    \w^{2}(\Sigma_{s}^{N} \cap \sph^{N-1}) %
    \leq C_{2}^{2} s \log \left(\frac{2N}{s}\right).
  \end{align*}
  In particular,~\eqref{eq:near-isometry} holds if
  \begin{align*}
    C_{1} \tilde K m^{-1/2}r \left[ %
    C_{2} \sqrt{s \log \left(\frac{2N}{s}\right)} + u\right] %
    < \delta.
  \end{align*}
  Observe that this condition is satisfied if~\eqref{eq:m-bd-near-isometry}
  holds:
  \begin{align*}
    m &> C_{\varepsilon}\delta^{-2} \tilde K^{2} r^{2} s
    \log\left(\frac{2N}{s}\right),
    \\
    C_{\varepsilon} %
    &:= 4C_{1}^{2} \cdot \max \left\{ %
    \log\left(2\varepsilon^{-1}\right), C_{2}^{2} \right\}.
  \end{align*}

\end{proof}

\begin{lem}[Conditionally Gaussian process]
  \label{lem:conditional-gp-expectation}
  Let $\mathcal{K} \subseteq \mathcal{K}_{s}^{N} \cap \sph^{N-1}$ and suppose
  $A \in \reals^{m \times N}$ is a normalized $K$-subgaussian matrix. Let
  $z \in \reals^{m}$ with $z_{i} \iid \mathcal{N}(0, 1)$ and define
  \begin{align*}
    f(A, z) := \sup_{x \in \mathcal{K}} \ip{A x, z}. %
  \end{align*}
  Let $\delta, \varepsilon > 0$ and $s \in \nats$ with $s \geq 1$. There is an
  absolute constant $C_{\varepsilon} > 0$, depending only on $\varepsilon$, so
  that if
  \begin{align*}
  m > C_{\varepsilon} \delta^{-2} \tilde K^{2} s \log(N/s),
  \end{align*}
  then with probability at least $1 - \varepsilon$ on the realization of $A$,
  \begin{align*}
    \E\left[ f(A,z) \mid A\right] %
    \leq C_{\delta} \sqrt{s \log (2N/s)}
  \end{align*}
  where $C_{\delta} > 0$ is an absolute constant depending only on
  $\delta$.
\end{lem}

\begin{proof}[Proof of {\autoref{lem:conditional-gp-expectation}}]

  By \autoref{lem:convexification},
  \begin{align*}
    \mathcal{K} %
    \subseteq \mathcal{K}_{s}^{N} \cap \sph^{N-1} %
    \subseteq \mathcal{L}_{s}.
  \end{align*}
  Therefore, $f(A,z) \leq \sup_{x \in \mathcal{L}_{s}} \ip{A x,
    z}$. Furthermore,
  \begin{align*}
    \mathcal{L}_{s} - \mathcal{L}_{s} \subseteq \mathcal{L}^{*}_{s}.
  \end{align*}
  By \autoref{lem:corollary-lmpv-thm1.4},
  \begin{align}
    \label{eq:conditional-gp-event}
    \max_{j \in [N]} \left|\|A^{j}\|_{2} - 1\right| %
    \leq \sup_{x \in \mathcal{L}^{*}_{s}} \left|\|Ax\|_{2} - \|x\|_{2}\right| %
    < \delta %
  \end{align}
  with probability at least $1 - \varepsilon$ if $m$ satisfies
  \begin{align}
    \label{eq:m-bd-conditional-gp}
    m %
    > 32 C_{\varepsilon} \delta^{-2} \tilde K^{2}  s \log(N/s).
  \end{align}
  Next, where $x \in \mathcal{L}_{s}$, define the random processes
  \begin{align*}
    X_{x} %
    &:= \ip{Ax, z},
      & z_{i} &\iid \mathcal{N}(0, 1);
      \\
      Y_{x} %
    &:= (1 + \delta) \ip{x, g}, 
      & g_{i} &\iid \mathcal{N}(0,1).
  \end{align*}
  Assume~\eqref{eq:m-bd-conditional-gp} holds and condition on the event
  $\mathcal{A}$ described by~\eqref{eq:conditional-gp-event}. Then
  $x - y \in \mathcal{L}_{s} - \mathcal{L}_{s} \subseteq \mathcal{L}^{*}_{s}$, so
  \begin{align*}
    \E (X_{y} - X_{x})^{2} %
    & = \|A(x - y)\|_{2}^{2}
    \\
    & \leq (1 + \delta)^{2} \|x - y\|_{2}^{2}
     = \E (Y_{y} - Y_{x})^{2}.
  \end{align*}
  Namely, conditioned on $\mathcal{A}$, the Sudakov-Fernique inequality
  (\autoref{thm:sudakov-fernique}) gives
  \begin{align*}
    \E\left[f(A,z) \mid A\right] %
    & \leq \E\left[\sup_{x \in \mathcal{L}_{s}} X_{x}\right] %
    \leq \E\left[\sup_{x \in \mathcal{L}_{s}} Y_{x}\right] %
    \\
    &= (1 + \delta) \w(\mathcal{L}_{s})
    \\
    & \leq 2 C (1+ \delta) \sqrt{s \log(2N/s)},
  \end{align*}
  where $C > 0$ is an absolute constant.

\end{proof}

\begin{rmk}[Subgaussianity of $f(A,z) \mid A$]
  \label{rmk:fAz-subgaussianity}
  Conditioned on $A$, Borell-TIS (\autoref{thm:borell-tis}) gives subgaussian
  concentration of $f(A,z)$ about $\E[f(A,z)\mid A]$. In particular,
  \begin{align*}
    \left\|f(A,z) - \E\left[f(A,z) \mid A\right] \right\|_{\Psi_{2}} %
    \lesssim \sigma_{\mathcal{K}}
  \end{align*}
  where, on the event $\mathcal{A}$ as defined in the proof of
  \autoref{lem:conditional-gp-expectation},
  \begin{align*}
    \sigma_{\mathcal{K}}^{2} %
    = \sup_{x \in \mathcal{K}} \E \left[|\ip{Ax, z}|^{2}\right] %
    = \sup_{x \in \mathcal{K}} \|Ax\|_{2}^{2} \leq (1 + \delta)^{2}.
  \end{align*}
  Note that subgaussianity of $f(A, z)$ about $\E[f(A, z)\mid A]$ can also be
  established using concentration of Lipschitz functions of Gaussians. Indeed,
  since $\mathcal{K} \subseteq \sph^{N-1}$, for each $A$ it holds that $f(A, z)$
  is Lipschitz in $z$. In fact, one can show that ``for most'' $A$, $f(A, z)$ is
  ``nearly'' $1$-Lipschitz.
\end{rmk}

\subsection{Proofs for constrained Lasso sensitivity}
\label{sec:proofs-constr-lasso}

\subsubsection{Suboptimal choice of $\tau$}
\label{sec:subopt-choice-tau}

The first result required to prove \autoref{thm:ls-instability} concerns the
case where {\ls} is controlled by a parameter that is too large. Under mild
regularity assumptions on the mapping $A$, we show that this underconstrained
problem cannot recover even the least-squares proximal denoising error rate in
the limiting low-noise regime. The second result of this section concerns the
situation where $\tau$ is too small, $\tau < \tau^{*}$. In this overconstrained
problem, the ground truth does not lie in the feasible set and one expects this
to be detrimental to recovery performance. We confirm this intuition
irrespective of the assumptions on the measurement matrix $A$.

\begin{lem}[Underconstrained {\ls}]
  \label{lem:ls-instability-uc}
  Let $A \in \reals^{m\times N}$ and assume that $\dim(\Null (A)) >
  0$. Given $x_{0} \in \reals^{N}, \eta > 0$ and $z \in \reals^{m}$ with
  $z_{i} \iid \mathcal{N}(0, 1)$, let $y := Ax_{0} + \eta z$. Suppose 
  $\tau > \|x_{0}\|_1$. Almost surely on the realization of $z$,
  \begin{align*}
    \lim_{\eta \to 0} \hat L (\tau; x_{0}, A, \eta z) = \infty.
  \end{align*}
\end{lem}

\begin{proof}[Proof of {\autoref{lem:ls-instability-uc}}] 
  Define $\rho := \tau - \tau^{*}$, where $\tau^{*} := \|x_{0}\|_1$. For
  simplicity, first assume $\Span (A) = \reals^{m}$. There exists $\zeta \in
  \reals^{N}$ such that $A\zeta = z$, and so $A(x_{0} + \eta \zeta) = Ax_{0} +
  \eta z = y$. Moreover, if $\eta < \rho \|\zeta\|_{1}^{-1}$ then $x_{0} + \eta
  \zeta \in \tau B_{1}^{N}$. In particular, $\xi := x_{0} + \eta \zeta$ solves
  {\ls}, because it is feasible and achieves the lowest possible objective value
  for {\ls}. Notice $\|\zeta\|_{1} < \infty$ almost surely, so for any
  realization of $z$, $\eta < \rho \|\zeta\|_{1}^{-1}$ holds for all $\eta$
  sufficiently small. Specifically, we have constructed $\xi$ solving {\ls}, and
  lying on the interior of $\tau B_{1}^{N}$. Consequently, almost surely there
  is $\nu \in \Null (A)$ so that $\xi + \nu \in \tau B_{1}^{N}$ and still $A
  (\xi + \nu) = y$. Scale $\nu$ if necessary so that $\|\xi + \nu\|_{1} \in
  [\frac{1}{2}(\tau + \tau_{*}), \tau]$. Then, almost surely on the realization
  of $z$,
  \begin{align*}
    \hat L (\tau; x_{0}, A, \eta z) %
    &\geq \eta^{-2} \|\xi + \nu - x_{0}\|_{2}^{2} %
    \\
    &\geq \frac{1}{N\eta^{2}} \|\xi + \nu - x_{0} \|_{1}^{2} %
    \\
    &\geq \frac{\rho^{2}}{4N \eta^{2}} %
    \xrightarrow{\eta \to 0} \infty.
  \end{align*}
  The case $\Span(A) \neq \reals^{m}$ is similar. This case is interesting only
  when $z \in \reals^{m} \setminus \Span(A) \neq \emptyset$, otherwise we argue
  as above. In this setting, define $P$ to be the projection onto the range of
  $A$ with $P^{\perp} = (I - P)$ being its orthogonal component. We may
  re-write the objective of {\ls} as
  \begin{align*}
    \|y - Ax\|_{2}^{2} %
    & = \|(P + P^{\perp})(y - Ax) \|_{2}^{2} %
    \\
    & = \|P(y - Ax) + P^{\perp} y\|_{2}^{2} %
    \\
    & =  \|\tilde y - Ax\|_{2}^{2} + \|P^{\perp} y\|_{2}^{2}
  \end{align*}
  where $\tilde y := Py$. Therefore, when $z \not \in \Span (A)$, solving {\ls}
  is equivalent to solving
  \begin{align}
    \argmin \left\{ \|Py - Ax\|_{2} : \|x\|_{1} \leq \tau\right\}.
    \tag{$\star$}
  \end{align}
  By construction, $\tilde y = Py \in \Rng(A)$, so we may apply the same
  argument as above to the program $(\star)$, implying
  \begin{align*}
    \lim_{\eta \to 0}\hat L (\tau; x_{0}, A, \eta z) = \infty.
  \end{align*}

\end{proof}

\begin{lem}[Overconstrained {\ls}]
  \label{lem:ls-instability-oc}
  Fix $\tau < \tau^{*}$. Almost surely on the realization $z$,
  \begin{align*}
    \lim_{\eta \rightarrow 0} \hat{L} \left(\tau ; x_{0}, A, \eta z\right)%
    = \infty.
  \end{align*}

\end{lem}

\begin{proof}[{Proof of \autoref{lem:ls-instability-oc}}]
  Let $\rho := \tau^{*} - \tau > 0$. For any solution $\xi$ to {\ls}, one has
  \begin{align*}
    \eta^{-2} \|\xi - x_{0} \|_{2}^{2} %
    \geq \frac{\rho^{2}}{N \eta^{2}}. %
  \end{align*}
  By definition, the desired result follows immediately:
  \begin{align*}
    \lim_{\eta \to 0}\hat L (\tau; x_{0}, A, \eta z) %
    \geq \eta^{-2} \|\xi - x_{0}\|_{2}^{2} %
    \geq \frac{\rho^{2}}{N \eta ^{2}} %
    \xrightarrow{\eta \to 0} \infty.
  \end{align*}

\end{proof}

\subsubsection{Uniform control over noise scales}
\label{sec:uniform-control-over-noise-scales}

In this section, we control {\ls} in the optimally tuned setting, uniform over
the noise scale $\eta$. Specifically, for any $x_{0} \in \Sigma_{s}^{N}$ we
control the expected error of recovery for {\ls} uniformly over the noise scale
$\eta > 0$. The results of \autoref{sec:control-condl-gp} are crucial for this
purpose. 

\begin{prop}[Uniform over noise scale]
  \label{prop:unif-noise-scale}
  Let $0 \leq s < N < \infty$ be integers and let $m \in \nats$. Let
  $A \in \reals^{m \times N}$ be a normalized $K$-subgaussian matrix, and fix
  $\delta, \varepsilon > 0$. Suppose that $y = Ax_{0} + \eta z$ for $\eta > 0$
  and $z \in \reals^{m}$ with $z_{i} \iid \mathcal{N}(0, 1)$. With probability
  at least $1 - \varepsilon$ on the realization of $A$, there exist constants
  $C_{\delta}, C_{\varepsilon} > 0$ so that if
  \begin{align*}
    m > C_{\varepsilon} \delta^{-2} \tilde K^{2} s \log \left(\frac{N}{2s}\right),
  \end{align*}
  then for all $\eta > 0$:
  \begin{align*}
      \E \left[ \|\hat x - x_{0} \|_{2}^{2} \mid A \right] %
    &\leq C_{\delta} \eta^{2} s \log \left(\frac{N}{2s}\right). %
  \end{align*}
  where $\hat x = \hat x(\tau^{*})$ solves {\ls} with
  $\tau = \tau^{*} := \|x_{0}\|_{1}$.
\end{prop}

\begin{proof}[Proof of {\autoref{prop:unif-noise-scale}}]
  If $s = 0$, the result holds trivially as, by construction,
  $\|\hat x - x_{0}\|_{2} = 0$ almost surely. Suppose $s \geq 1$. By definition
  of $\hat x$, where $h := \hat x - x_{0}$,
  \begin{align*}
    \|A \hat x - y\|_{2}^{2} \leq \|A x_{0} - y\|_{2}^{2} %
    \quad \implies \quad %
    \|Ah\|_{2}^{2} \leq 2 \eta \ip{A h, z}.
  \end{align*}

  \noindent\textbf{Step 1:} Lower bound $\|Ah\|_{2}$ with high
  probability. Note that $\|Aw\|_{2} = \|w\|_{2} \|A \hat w\|_{2}$ for
  $w \neq 0$ where $\hat w := w / \|w\|_{2}$. By
  \autoref{lem:corollary-lmpv-thm1.4}, there is an event $\mathcal{A}_{1}$ with
  $\mathbb{P}(\mathcal{A}_{1}) \geq 1 - \varepsilon /2$ on which
  \begin{align*}
    \sup_{x \in \mathcal{K}_{4s}^{N}} \left|\|Ax\|_{2} - \|x\|_{2}\right| %
    \leq \sup_{x \in \mathcal{L}_{4s}} \left|\|Ax\|_{2} - \|x\|_{2}\right| %
    < \delta_{1} %
  \end{align*}
  if $m$ satisfies
  \begin{align}
    \label{eq:unif-m-lb1}
    m %
    > 16 C_{\varepsilon}' \delta_{1}^{-2} \tilde K^{2}
    s\log\left(\frac{N}{2s}\right).
  \end{align}
  Specifically, $h \in \mathcal{J}_{4s}^{N}$ by
  \autoref{lem:descent-cone-condition}, meaning
  $\hat h \in \mathcal{J}_{4s}^{N} \cap \sph^{N-1} \subseteq
  \mathcal{K}_{4s}^{N}$ if $h \neq 0$. So, conditioning on $\mathcal{A}_{1}$
  and enforcing~\eqref{eq:unif-m-lb1}, one has
  \begin{align}
    \label{eq:unif-over-noise-scales-1}
    \|A h\|_{2}^{2} %
    & = \|h\|_{2}^{2} \|A\hat h\|_{2}^{2} %
    \geq \|h\|_{2}^{2} \left(\|\hat h\|_{2} - \delta_{1}\right)^{2}
    \geq (1 - \delta_{1})^{2}\|h\|_{2}^{2}.
  \end{align}
  The inequality $(1 - \delta_{1})^{2}\|h\|_{2}^{2} \leq \|Ah\|_{2}^{2}$
  holds also for $h = 0$.

  \noindent\textbf{Step 2a:} Upper bound $\ip{Ah, z}$. Again
  using that $h \in \mathcal{J}_{4s}^{N}$,
  \begin{align}
    \label{eq:unif-over-noise-scales-2}
    2\eta \ip{Ah, z} %
    \leq 2\eta\|h\|_{2} \sup_{\hat h \in \mathcal{K}_{4s}^{N} \cap \sph^{N-1}} %
    \ip{A\hat h, z}. %
  \end{align}
  \textbf{Step 2b:} Control the latter quantity in expectation. By
  \autoref{lem:conditional-gp-expectation}, there is $C_{\varepsilon}'' > 0$ so
  that for
  \begin{align}
    \label{eq:unif-m-lb2}
    m > 4 C_{\varepsilon}''\delta^{-2}\tilde K^{2} s \log \left(\frac{N}{4s}\right),
  \end{align}
  there is an event $\mathcal{A}_{2}$ holding
  with probability at least $1 - \varepsilon/2$, on which there is a constant
  $C_{\delta} > 0$ such that
  \begin{align*}
    \E \left[ \sup_{\hat h \in \mathcal{K}_{4s}^{N}} \ip{A\hat h, z} \mid A \right]%
    \leq 2C_{\delta} \sqrt{s \log \left(\frac{N}{2s}\right)}.
  \end{align*}
  \textbf{Step 3:} Now combine steps 1 and 2a. Assume $m$ simultaneously
  satisfies~\eqref{eq:unif-m-lb1} and~\eqref{eq:unif-m-lb2}, and condition on
  $\mathcal{A}_{1} \cap \mathcal{A}_{2}$, which holds with probability at least
  $1- \varepsilon$.  Combining~\eqref{eq:unif-over-noise-scales-1}
  and~\eqref{eq:unif-over-noise-scales-2}, and letting
  $\delta_{1} := 1 - 2^{-1/2}$ gives
  \begin{align*}
    \|h\|_{2} %
    \leq 4\eta \sup_{\hat h \in \mathcal{K}_{4s}^{N}} \ip{A \hat h, z}.
  \end{align*}
  Take expectation of both sides and bound the quantity by applying step
  2b. This yields,
  \begin{align*}
    \E \left[ \|h\|_{2}\mid A\right] %
    \leq 8 C_{\delta} \eta \sqrt{s \log \left(\frac{N}{2s}\right)}.
  \end{align*}
  Note that by setting $C_{\varepsilon} %
  := \max \{ \frac{32 C_{\varepsilon}'}{3 - 2 \sqrt 2} , 4
  C_{\varepsilon}''\}$, it suffices to require
  \begin{align*}
    m %
    > C_{\varepsilon} \delta^{-2} \tilde K^{2} s \log \left( \frac{N}{2s}\right).
  \end{align*}

  Alternatively, one may also apply a standard fact for subgaussian random
  variables. Recall as in~\eqref{eq:borell-norm},
  $\vertiii{X_{w}}_{K_{4s}^{N}} := \sup_{w \in K_{4s}^{N}} X_{w}$. Then
  $\left\| \vertiii{X_{w}}_{K_{4s}^{N}} - \E \vertiii{X_{w}}_{K_{4s}^{N}}
  \right\|_{\Psi_{2}} \leq \sigma_{K_{4s}^{N}}^{2}$ by \autoref{thm:borell-tis},
  and so there is an absolute constant $C > 0$ such that on $\mathcal{A}_{1}$,
  \begin{align*}
    &\left\|\vertiii{X_{w}}_{K_{4s}^{N}} - \E_{z} \vertiii{X_{w}}_{K_{4s}^{N}} %
    \right\|_{L^{2}}^{2} %
    \\
    &= \E_{z} \vertiii{X_{w}}_{K_{4s}^{N}}^{2} - \left(
    \E_{z} \vertiii{X_{w}}_{K_{4s}^{N}}\right)^{2} %
    \\
    & \leq C \sigma_{K_{4s}^{N}}^{2} \leq C (1 + \delta_{1})^{2}.
  \end{align*}
  where $X_{w} := \ip{Aw, z}$ conditioned on $A$. In particular, choosing
  instead $\delta_{1} := 1 - 2^{-1/4}$,
  \begin{align*}
    \E \left[ \|h\|_{2}^{2}\mid A\right] %
    \leq 8 \eta^{2} \left(4 C_{\delta}^{2} s \log \frac{N}{2s} %
    + C\sqrt 2 \right).
  \end{align*}
  Rearranging, and observing that the right hand term in parentheses is small
  relative to the left hand term, we may obtain a new absolute constant
  $C_{\delta} > 0$ depending only on $\delta$ such that
  \begin{align*}
    \eta^{-2} \E \left[ \|\hat x - x_{0}\|_{2}^{2} \mid A \right] %
    \leq C_{\delta} s \log \frac{N}{2s}.
  \end{align*}

\end{proof}

\begin{rmk}
  In the proof above, no attempt was made to optimize constants. In fact,
  several simplifications were made for clarity of presentation, which in turn
  resulted in larger than necessary constants.
\end{rmk}

\begin{rmk}[Uniform control over noise scale and signal class]
  Observe that the result above is uniform over noise scale $\eta > 0$ and
  signal $x_{0} \in \Sigma_{s}^{N}$. In particular, we could have written
  (conditioning on $A$),
  \begin{align*}
     \sup_{\eta > 0} \sup_{x_{0} \in \Sigma_{s}^{N}}
    \hat R (\tau^{*}; x_{0}, A, \eta) %
    \leq C_{\delta} s \log \frac{N}{2s}.
  \end{align*}

\end{rmk}

\subsubsection{Optimal choice of $\tau$ and phase transition}
\label{sec:optimal-choice-tau}

Here, we synthesize the technical results
of~\autoref{sec:uniform-control-over-noise-scales} to show that, with high
probability on the realization of $A$, {\ls} achieves order-optimal risk in the
limiting low-noise regime when $m$ is sufficiently large and $\tau =
\tau^{*}$. 

\begin{lem}[Tuned {\ls}]
  \label{lem:ls-instability-ec}
  Fix $\delta, \varepsilon > 0$ and let $A \in \reals^{m\times N}$ be a
  normalized $K$-subgaussian matrix. For $s\in \nats$ fixed with
  $0 \leq s \leq m$, suppose $x_{0} \in \Sigma_{s}^{N}$ and $\eta > 0$. If
  $m$ satisfies
  \begin{align*}
    m > C_{\varepsilon}' \delta^{-2} \tilde K^{2} s \log \frac{N}{2s},
  \end{align*}
  then, with probability at least $1 - \varepsilon$ on the realization $A$,
  there exist constants $0 < c_{\delta} < C_{\delta} < \infty$ such that
  \begin{align*}
    c_{\delta} \cdot s \log\left( \frac{N}{s}\right) %
    &\leq  \lim_{\eta \to 0}
    \sup_{x_{0} \in \Sigma_{s}^{N}}
    \hat{R} (\tau^{*}; x_{0}, N, \eta) %
    \\
    & \leq C_{\delta} \cdot s\log\left( \frac{N}{2s} \right).
  \end{align*}

\end{lem}

\begin{proof}[{Proof of \autoref{lem:ls-instability-ec}}]

  For simplicity of the proof, we assume $\tilde K^{2} = 1$.

  \noindent\textbf{Upper bound:} Given $\delta, \varepsilon_{1} > 0$, assume
  \begin{align*}
    m > C_{\varepsilon_{1}} \delta^{-2} s \log \frac{N}{2s}.
  \end{align*}
  With probability at least $1 - \varepsilon_{1}$ on the realization of $A$, by
  \autoref{prop:unif-noise-scale}, for any $x_{0} \in \Sigma_{s}^{N}$ and
  $\eta > 0$,
  \begin{align*}
    \hat R(\tau^{*}; x_{0}, A, \eta) %
    \leq C_{\delta} \cdot s \log \frac{N}{2s}.
  \end{align*}
  In particular,
  \begin{align*}
     \lim_{\eta \to 0} \sup_{x_{0} \in \Sigma_{s}^{N}}
    \hat R(\tau^{*}; x_{0}, A, \eta) %
    \leq C_{\delta} \cdot s \log \frac{N}{2s}.
  \end{align*}

  \noindent \textbf{Lower bound:} From \autoref{coro:opt-risk-near-equiv}
  and~\cite[Theorem 1]{candes2013well}, 
  \begin{align*}
    \sup_{x_{0} \in \Sigma_{s}^{N}}\hat R(\tau^{*}; x_{0}, A, \eta) %
    & \geq  \inf_{x_{*}} \sup_{x_{0} \in \Sigma_{s}^{N}} \eta^{-2}
    \E \|x_{*} - x_{0} \|_{2}^{2}
    \\
    & \geq \frac{C_{1} N}{\|A\|_{F}^{2}}s \log\left(\frac{N}{s}\right).
  \end{align*}
  In particular,
  \begin{align*}
     \lim_{\eta \to 0} \sup_{x_{0} \in \Sigma_{s}^{N}}
    \hat R(\tau^{*}; x_{0}, A, \eta) %
    \geq \frac{C_{1} N}{\|A\|_{F}^{2}}s \log\left(\frac{N}{s}\right).
  \end{align*}
  Now, $\E \|A\|_{F}^{2} = N$, and $\|A\|_{F}^{2}$ admits subexponential
  concentration around its expectation by Bernstein's inequality~\cite[Corollary
  2.8.3]{vershynin2018high}. Therefore, with probability at least
  $1 - \varepsilon_{2}$ on the realization of $A$, there is a constant
  $c_{\delta} > 0$ depending only on $C_{1}$ and $\delta$ such that
  \begin{align*}
     \lim_{\eta \to 0} \sup_{x_{0} \in \Sigma_{s}^{N}}
    \hat R(\tau^{*}; x_{0}, A, \eta) %
    \geq c_{\delta} \cdot s \log\left(\frac{N}{s}\right),
  \end{align*}
  under the condition that
  \begin{align*}
    m \geq C \delta^{-2}N^{-1}  \log \frac{2}{\varepsilon_{2}}.
  \end{align*}

  \noindent\textbf{Combine:} Finally, set $\varepsilon_{1} = \varepsilon_{2}
  = \varepsilon/2$. Under the assumptions on $m$, with probability at least
  $1 - \varepsilon$ on the realization of $A$ it holds that
  \begin{align*}
    c_{\delta} \cdot s \log\left(\frac{N}{s}\right) %
    & \leq  \lim_{\eta \to 0} \sup_{x_{0} \in \Sigma_{s}^{N}}
    \hat R(\tau^{*} ; x_{0}, A, \eta) %
    \\
    & \leq C_{\delta} \cdot s \log \frac{N}{2s}.
  \end{align*}

\end{proof}

We conclude this section with the proof of \autoref{thm:ls-instability} which
combines \autoref{lem:ls-instability-ec} and the results of
\autoref{sec:subopt-choice-tau}. Namely, even when $m$ is sufficiently large,
{\ls} admits order-optimal risk in the limiting low-noise regime only when the
governing parameter is chosen optimally.

\begin{proof}[Proof of {\autoref{thm:ls-instability}}]
  This result follows immediately from the lemmas of this section. Indeed, a
  direct application of \autoref{lem:ls-instability-ec} gives
  \begin{align*}
    c_{\delta} \cdot s \log\left(\frac{N}{s}\right) %
    & \leq  \lim_{\eta \to 0} \sup_{x \in \Sigma_{s}^{N}}
    \hat R(\tau^{*} ; x, A, \eta) %
    \\
    & \leq C_{\delta} \cdot s \log \frac{N}{2s}.
  \end{align*}
  Otherwise,
  $\tau \neq \tau^{*}$. First, if $\tau < \tau^{*}$, then
  \autoref{lem:ls-instability-oc} immediately implies
  \begin{align*}
    \lim_{\eta \to 0} \hat L (\tau; x_{0}, A, \eta z) = \infty.
  \end{align*}
  Otherwise, assume $\tau > \tau^{*}$. In order to apply
  \autoref{lem:ls-instability-uc}, $A$ must satisfy $\dim(\Null(A)) > 0$, which
  holds trivially, as $m < N$. In particular, \autoref{lem:ls-instability-uc}
  implies almost surely on $(A, z)$,
  \begin{align*}
    \lim_{\eta \to 0} \hat L (\tau; x_{0}, A, \eta z) = \infty.
  \end{align*}
\end{proof}

\begin{rmk}
  The proof for \autoref{thm:ls-instability} proceeds whether $z$ be
  deterministic (say with fixed norm $\|z\|_{2} = \sqrt m$) or have entries
  $z_{i} \iid \mathcal{N}(0,1)$. We have presented it this way so that the
  assumption is consistent with the implicit assumption on the noise for the
  result concerning $\hat R(\tau^{*}; x_{0}, A, \eta)$.
\end{rmk}

\subsection{Proofs for basis pursuit suboptimality}
\label{sec:proofs-basis-pursuit}

\subsubsection{Suboptimal regime for underconstrained basis pursuit}
\label{sec:subopt-regime-underc}

This section contains the proof for \autoref{lem:uc-bp-subgaus} in
\autoref{sec:underc-param-inst}.

\begin{proof}[Proof of {\autoref{lem:uc-bp-subgaus}}]

  It suffices to prove this result for the best choice of $\sigma$ and any
  $x \in \Sigma_{s}^{N}$. In particular, choose $x_{0} \in \Sigma_{s}^{N}$
  having at least one non-zero entry, and for which the non-zero entries have
  magnitude satisfying $|x_{0,j}| \geq C\eta \sqrt{m}$,
  $j \in \supp(x_{0}) \subseteq [N]$. For this choice of $x_{0}$, let
  $y = Ax_{0} + \eta z$ and define the event
  $\mathcal{F} := \{ \|y\|_{2} \leq \sigma \}$.


  For any $\sigma \geq \eta \sqrt m$, re-choose $x_{0} \in \Sigma_{s}^{N}$ if
  necessary so that moreover $\mathbb{P}(\mathcal{F}^{C}) \geq
  0.99$. Restricting to $\mathcal{F}^{C}$, the solution to {\bp} satisfies, by
  the KKT conditions~\cite{bertsekas2003convex},
  \begin{align*}
    \eta^{2} m %
    \leq \sigma^{2} %
    = \|Ah \|_{2}^{2} - 2 \eta \ip{Ah, z} + \eta^{2} \|z\|_{2}^{2}.
  \end{align*}
  By \autoref{lem:corollary-lmpv-thm1.4}, it holds with probability at least
  $1 - \varepsilon$ on the realization of $A$ that
  \begin{align*}
    (1 + \delta)^{2} \|h\|_{2}^{2} %
    \geq \|Ah \|_{2}^{2} %
    \geq \eta^{2}(m - \|z\|_{2}^{2}) + 2\eta \ip{Ah, z}
  \end{align*}
  Define the event
  $\mathcal{Z}_{\leq} := \{ \|z\|_{2}^{2} \leq m - 2 \sqrt m\}$ and observe
  that further restricting to $\mathcal{F}^{C} \cap \mathcal{Z}_{\leq}$ thereby
  gives
  \begin{align*}
    (1 + \delta)^{2} \|h\|_{2}^{2} %
    & \geq 2 \eta^{2}\sqrt m - 2 \eta \|h\|_{2} f(A,z)
    \\
    & \geq 2 \eta^{2}\sqrt m - \frac{1}{2} \|h\|_{2}^{2}
    - 2\eta^{2} f^{2}(A,z),
  \end{align*}
  where $f(A, z)$ is defined as in \autoref{lem:conditional-gp-expectation} with
  $\mathcal{K} = \mathcal{K}_{2s}^{N}\cap \sph^{N-1}$. Indeed, where
  $\hat h = h / \|h\|_{2}$, one has $\ip{A\hat h, z} \leq f(A,z)$ since
  $\hat h \in \mathcal{K}_{2s}^{N}\cap \sph^{N-1}$ with high probability on the
  realization of $A$. This yields the following bound on the risk:
  \begin{align}
    \label{eq:uc-bp-subgaus-1}
    \tilde R(\sigma; &\, x_{0}, A, \eta)
    \nonumber
    \\
    &\geq \eta^{-2} \E_{z} \left[\|h\|_{2}^{2} \cdot
      \1 \left( \mathcal{F}^{C} \cap \mathcal{Z}_{\leq}\right)\right] %
      \nonumber 
    \\
    & \geq C_{\delta} \E_{z}\left[\big(\sqrt m - f^{2}(A,z)\big)
      \cdot \1 \left( \mathcal{F}^{C} \cap \mathcal{Z}_{\leq}\right) \right]
      \nonumber
    \\
    & = C_{\delta} \sqrt m
      \mathbb{P}\left(\mathcal{F}^{C} \cap \mathcal{Z}_{\leq}\right)
      - C_{\delta} \E_{z} \left[f^{2}(A,z)\cdot
      \1\left(\mathcal{F}^{C} \cap \mathcal{Z}_{\leq}\right)\right]
      \nonumber
    \\
    &\geq C_{\delta} \sqrt m
      \mathbb{P}\left(\mathcal{F}^{C} \cap \mathcal{Z}_{\leq}\right)
      - C_{\delta} \E_{z} f^{2}(A,z)
  \end{align}
  Finally, we bound $\E_{z} f^{2}(A,z) = \E [ f^{2}(A,z) \mid A]$. With high
  probability on the realization of $A$:
  \begin{align*}
    \E [ f^{2}(A,z) \mid A] \leq C \E[f(A, z) \mid A]^{2} \leq C_{\delta}s \log (N/s). 
  \end{align*}
  Above, we have first used~\cite[Exercise~7.6.1]{vershynin2018high} followed by
  an application of~\autoref{lem:conditional-gp-expectation}. Another way to see
  this would be through the successive application
  of~\autoref{rmk:fAz-subgaussianity}
  and~\autoref{lem:conditional-gp-expectation}, noting that
  $f(A, z) - \E[f(A, z)\mid A]$ is a centered subgaussian random variable.


  Consequently, using that
  $\mathbb{P}(\mathcal{F}^{C} \cap \mathcal{Z}_{\leq}) \geq
  C$,~\eqref{eq:uc-bp-subgaus-1} becomes
  \begin{align}
    \label{eq:uc-bp-subgaus-2}
    \tilde R(\sigma; x_{0}, A, \eta) %
    \geq C_{\delta} \left(\sqrt m - s\log(N/s)\right).
  \end{align}
  The result follows trivially from the definition of $\sup$ and by the initial
  assumption on $m$.

\end{proof}

\subsubsection{Suboptimal regime for overconstrained basis pursuit}
\label{sec:overc-param-inst}

In this section, we show that $\tilde R(\sigma; x_{0}, A, \eta)$ is suboptimal
for $\sigma \leq \eta \sqrt m$. To the chagrin of the beleaguered reader, the
proofs in this section require several technical lemmas, some assumptions and
notation. As much as possible, we attempt to relegate to the appendix those
details that, we believe, do not aid the reader's intuition and are particularly
technical.

The flow of this section will proceed as follows. After establishing required
preliminary details, we state and prove results concerning the ability of {\bp}
to recover the $0$ vector from noisy random measurements. The results exhibit a
regime in which $\tilde R(\sigma; x_{0}, A, \eta)$ may be lower-bounded in the
case where $\sigma = \eta \sqrt m$. Then, we proceed by showing that {\bp}
performs no better if $\sigma$ is allowed to be smaller. In particular, we
obtain lower bounds on $\tilde R(\sigma; x_{0}, A, \eta)$ for
$\sigma \leq \eta \sqrt m$. Motivation for this latter result is readily
observed by a re-phrasing of the projection lemma in
\autoref{prop:oc-solution-ordering}.

\paragraph{Preliminaries.}

For $z \in \reals^{m}$ and $\sigma > 0$ define
$F(z; \sigma) := \{ q \in \reals^{m} : \|q - z\|_{2}^{2} \leq \sigma^{2}\}$ and
denote $F := F(z; \sqrt m)$. For a matrix $A \in \reals^{m \times N}$, denote
$B_{1, A} := \{ Ax \in \reals^{m} : x \in B_{1}^{N}\}$, and define the gauge of
$B_{1, A}$ by
\begin{align}
  \|q \|_{1, A} %
  &:= \inf \{ \|x\|_{1} : Ax = q, x \in \reals^{N} \} %
    \nonumber
  \\
  \label{eq:defn-gauge-B1A}
  & = \inf \{ \lambda > 0 : q \in \lambda B_{1, A}\}.
\end{align}
Recall a gauge is nonnegative, positively homogeneous, convex and vanishes at
the origin. Moreover, note that $B_{1, A}$ is a random set, and so
$\|\cdot\|_{1, A}$ is random. Now, for a matrix $A \in \reals^{m\times N}$,
$z \in \reals^{m}$ and $\sigma > 0$, define the program
\begin{align}
  \label{eq:bq}
  \tilde q(\sigma; A, z) %
  := \argmin \big\{ \|q\|_{1, A} : q \in F(z; \sigma) \big\},
  \tag{$\mathrm{BQ}_{\sigma}$}
\end{align}
where $\|\cdot\|_{1, A}$ is defined as in~\eqref{eq:defn-gauge-B1A}. Where
clear, we omit notating the dependence of $\tilde q(\sigma; A, z)$ on $A$ and
$z$, writing simply $\tilde q(\sigma)$.

With the above notation, we define an admissible ensemble. The elements of an
admissible ensemble will be used to state the main lemma,
\autoref{lem:geometric-lemma}. The technical arguments characterizing their
properties appear in \autoref{sec:techn-lemm-supp}.

\begin{defn}[Admissibile ensemble]
  Let $0 \leq s < N$ be integers, and let $m : \nats \to \nats$ be an
  integer-valued function mapping $N \mapsto m(N)$ such that
  $\lim_{N\to \infty} m(N) / N = \gamma \in (0, 1)$. For any
  $0 < \theta < \min\{1-\gamma, \gamma\}$, define $N_{\theta} \geq 1$ to be the
  least integer such that for all $N \geq N_{\theta}$,
  \begin{align*}
    \left|\frac{m(N)}{N} - \gamma\right| < \theta.
  \end{align*}
  Where $N \geq 2$, let $A(N)$ be a family of normalized $K$-subgaussian
  matrices $A = A(N) \in \reals^{m(N) \times N}$. Define
  $N_{*} := \max \{ N_{\theta}, N_{\text{RIP}}\}$ where $N_{\mathrm{RIP}} \geq 1$ is the
  least positive integer such that for all $N \geq N_{\mathrm{RIP}}$,
  \begin{align*}
    m(N) \geq C_{\varepsilon} \delta^{-2} \tilde K^{2} s \log \frac{2N}{s}.
  \end{align*}
  where $\delta, \varepsilon > 0$ are fixed in advance.

  Let $z = z(N) \in \reals^{m(N)}$ with $z_{i} \iid \mathcal{N}(0, 1)$. Define
  $F = F(z; \sqrt{m(N)}) = \{ q \in \reals^{m(N)} : \|q - z\|_{2}^{2} \leq
  m(N)\}$ and omit writing explicitly its dependence on $N$, unless
  necessary. Define $\alpha_{1} = \alpha_{1}(N) := a_{1} m(N)^{1/4}$ for some
  dimension-independent constant $a_{1} > 0$;
  $\lambda = \lambda(N) := L \sqrt{\frac{m(N)}{\log N}}$ for some
  dimension-independent constant $L > 1$; and
  \begin{align*}
    K_{1} &= K_{1}(N) := \lambda(N) B_{1, A} \cap \alpha_{1}(N) B_{2}^{m(N)},
    \\
      K_{2} &= K_{2}(N) := \lambda(N) B_{1, A} \cap \alpha_{2}(N) B_{2}^{m(N)},
  \end{align*}
  where $0 < \alpha_{2} = \alpha_{2}(N) \leq \alpha_{1}$ will be quantified in
  \autoref{prop:gw-ub-1}. Lastly, define the following random processes. For
  $g \in \reals^{m}$ with $g_{i} \iid \mathcal{N}(0, 1)$, let
  \begin{align*}
    X_{1} &:= \sup_{x \in K_{1}}|\ip{x, g}|,
    &
    X_{2} &:= \sup_{x \in K_{2}} |\ip{x, g}|.
  \end{align*}
  Thus we define an $(s, m(N), N, \delta, \varepsilon, \theta)$-admissible
  ensemble as the collection $(A(N), z(N), K_{1}(N), K_{2}(N), X_{1}, X_{2})$
  satisfying the conditions just described, defined for all $N \geq N_{*}$. This
  collection will generally be abbreviated to
  $(A, z, K_{1}, K_{2}, X_{1}, X_{2})$ where clear.
\end{defn}

Where possible, we simplify notation by omitting explicit dependence on
arguments. For example, if $N$ is fixed, then we may refer to $m(N)$ simply as
$m$. Note, however, that for any $N \geq 2$, $K_{1}$ and $K_{2}$ always depend
on $\alpha_{1} = \alpha_{1}(N)$ and $\alpha_{2} = \alpha_{2}(N)$,
respectively. Further observe that $K_{1}$ and $K_{2}$ are random, as they
depend on the matrix $A$. Observe that $N_{\theta}$ depends on $\theta$,
$\gamma$ and $m(\cdot)$, and omit writing explicitly its dependence on the
latter two; we assume $m(\cdot)$ and $\gamma$ are fixed in advance. Requiring
$N \geq N_{\text{RIP}}$ is the key condition on $m(N)$ so that
\autoref{lem:corollary-lmpv-thm1.4} holds. Clearly, $N_{\mathrm{RIP}}$ depends
on the parameters $\delta, \varepsilon, K, s$ and the function $m(\cdot)$; for
simplicity of presentation we omit writing explicitly its dependence on these
parameters. Finally, note that the parameters on which $N_{*}$ depends are
exactly those for which $N_{\theta}$ and $N_{\text{RIP}}$ depend.

\begin{prop}
  \label{prop:oc-solution-ordering}
  Let $z \in \reals^{m}$ and $A \in \reals^{m \times N}$ be a normalized
  $K$-subgaussian matrix with $1 \leq m < N$. If
  $0 < \sigma_{1} < \sigma_{2} < \infty$ and $\tilde q(\sigma)$ solves {\bq}
  then almost surely on $(A, z)$,
  \begin{align*}
    \|\tilde q (\sigma_{1}) \|_{2} \geq \|\tilde q (\sigma_{2})\|_{2}.
  \end{align*}
\end{prop}
\begin{proof}[Proof of {\autoref{prop:oc-solution-ordering}}]
  The result follows by~\autoref{cor:projection-lemma}, because
  $\|\cdot\|_{1, A}$ is a gauge.
\end{proof}

\paragraph{The geometric lemma.}

Next, we state a lemma with a geometric flavour, \autoref{lem:geometric-lemma},
which is the main workhorse for proving suboptimality of $\tilde R$ in the
overconstrained setting. It is a generalization of \cite[Lemma
6.2]{berk2020sl1mpc}.

\begin{lem}[Geometric Lemma]
  \label{lem:geometric-lemma}

  Fix $\delta, \varepsilon_{1}, \varepsilon_{2} > 0$ and
  $\theta \in (0, \gamma)$. Given an
  $(s, m, N, \delta, \varepsilon_{1}, \theta)$-admissible ensemble, there is a
  choice of $a_{1} > 0$ defining $\alpha_{1}(N)$; $L > 1$ defining $\lambda(N)$;
  an integer $N_{0} \geq N_{*}$; and absolute constants $p, k > 0$, so that the
  following occurs. For all $N \geq N_{0}$, with probability at least
  $1 - \varepsilon_{1}$ on the realization of $A$, there is an event
  $\mathcal{E} := \mathcal{E}(\varepsilon_{1}, \varepsilon_{2})$ for $z$ on
  which
  \begin{align*}
    &1.~K_{1} \cap F \neq \emptyset, %
      &
      &2.~K_{2} \cap F = \emptyset,  %
    \\
      &3.~\alpha_{2} > C N^{p}, %
    &
      &4.~\mathbb{P}(\mathcal{E}) > k.
  \end{align*}
  Above, $k$ depends on $N_{0}$ and $\varepsilon_{2}$ only; $p$ on
  $\delta, \gamma$ and $\theta$ only.
\end{lem}

\begin{proof}[Proof of \autoref{lem:geometric-lemma}]
  For constants $0 < C_{2} < C_{1} < \infty$, define the events
  \begin{align*}
    \mathcal{Z}_{<} &:= \{ \|z\|_{2}^{2} \leq m + C_{1}\sqrt m \}
    \\
      \mathcal{Z}_{>} &:= \{ \|z\|_{2}^{2} \geq m + C_{2} \sqrt m \}.
  \end{align*}
  By \autoref{prop:gw-lb-3} and \ref{prop:gw-ub-3}, there is an integer
  $N_{0} \geq N_{*}$ (select the larger of the two bestowed by each result),
  and respective events, $\mathcal{E}_{1}, \mathcal{E}_{2}$, so that with
  probability at least $1 - \varepsilon_{1}$ on the realization of $A$,
  \begin{align*}
    \mathbb{P}(\mathcal{E}_{1}) %
    &\geq \mathbb{P}(\mathcal{Z}_{<}) - \varepsilon_{2} %
    & %
      \mathbb{P}(\mathcal{E}_{2}) %
    &\geq \mathbb{P}(\mathcal{Z}_{>}) - \varepsilon_{2} .
  \end{align*}
  In particular, for $\mathcal{E} := \mathcal{E}_{1} \cap \mathcal{E}_{2}$,
  choose a largest such absolute constant
  $k := k(N_{0}, C_{1}, C_{2}, \varepsilon_{2}) > 0$ so that
  \begin{align*}
    \mathbb{P}(\mathcal{E}) %
    = \mathbb{P}(\mathcal{E}_{1} \cap \mathcal{E}_{2}) %
    \geq \mathbb{P}(\mathcal{Z}_{<}\cap \mathcal{Z}_{>})
    - 2 \varepsilon_{2} %
    \geq k.
  \end{align*}
  As per \autoref{prop:gw-lb-3} and \autoref{prop:gw-ub-3}, conditioning on
  $\mathcal{E}$ and letting $N \geq N_{0}$ gives $K_{1} \cap F \neq \emptyset$
  and $K_{2} \cap F = \emptyset$ with probability at least $1 - \varepsilon_{1}$
  on the realization of $A$, as desired. In this regime, that there exists
  $p > 0$ satisfying $\alpha_{2} = \alpha_{2}(N) \geq C N^{p}$ is a consequence
  of \autoref{prop:gw-ub-2}. One need simply select the largest $p$ satisfying
  for all $N \geq N_{0}$:
  \begin{align*}
    C N^{p}\sqrt{\log N} \leq C_{\delta, \gamma, L, \theta} N^{d/2}.
  \end{align*}
  Thus, for all $N \geq N_{0}$, with probability at least $1 - \varepsilon$ on
  the realization of $A$ there exists an event $\mathcal{E}$ for $z$ on which
  all four of the desired criteria hold.

\end{proof}

\paragraph{Implications for overconstrained basis pursuit.} Finally, we state
the main results of this section. The first result,
\autoref{lem:bp-oc-tuned-zero}, uses the geometric lemma to show that there
exists a regime in which $\tilde R$ is suboptimal in the setting where
$x_{0} = 0$ and $\sigma = \eta \sqrt m$. From there, we show in
\autoref{lem:bp-oc-zero} that $\tilde R$ is no better if $\sigma$ is any
larger. This is enough to state a maximin suboptimality result for {\bp}, with
$\sigma$ restricted to $(0, \eta \sqrt m]$, in
\autoref{thm:bp-oc-maximin}. Notably, this result is stronger than the analogous
minimax statement, which necessarily follows from the maximin result.

\begin{lem}[Lower bound $\tilde R(\eta \sqrt m; 0, A, \eta)$]
  \label{lem:bp-oc-tuned-zero}

  Fix $\delta, \varepsilon, \eta > 0$ and suppose $m : \nats \to \nats$
  satisfies $m(N)/N \to \gamma \in (0, 1)$. There is $N_{0} \in \nats$ and an
  absolute constant $p > 0$ so that for all $N \geq N_{0}$, if
  $A \in \reals^{m(N) \times N}$ is a normalized $K$-subgaussian matrix, then
  with probability at least $1 - \varepsilon$ on the realization of $A$,
  \begin{align*}
    \tilde R (\eta \sqrt m; 0, A, \eta) \geq C_{\delta, \gamma, K} N^{p}.
  \end{align*}
\end{lem}

\begin{proof}[Proof of \autoref{lem:bp-oc-tuned-zero}]
  By a simple scaling argument, it suffices to assume $\eta = 1$. Consider an
  $(s, m, N, \delta, \varepsilon, \theta)$-admissible ensemble. By
  \autoref{lem:geometric-lemma}, there is a choice of $a_{1} > 0$ for
  $\alpha_{1}(N)$ and $L > 1$ for $\lambda (N)$, an integer $N_{0} \geq N_{*}$
  and absolute constants $k, p > 0$ so that with probability at least
  $1 - \varepsilon/2$ on the realization of $A$, there is an event $\mathcal{E}$
  for $z$ on which $K_{1} \cap F \neq \emptyset$, $K_{2} \cap F = \emptyset$,
  and for which $\mathbb{P}(\mathcal{E}) \geq k_{3}$. Where $\tilde q$ solves
  {\bq}, observe that $\tilde q = A\tilde x(\sqrt m)$ and moreover, by
  construction, $\tilde q \in (K_{1} \setminus K_{2}) \cap F$. In particular,
  \begin{align*}
    \|\tilde q\|_{1, A} \leq \lambda, \qquad %
    \alpha_{2} \leq \|\tilde q\|_{2} \leq \alpha_{1}.
  \end{align*}
  By \autoref{coro:largest-subgaus-singval} and our initial assumptions,
  \begin{align*}
    \|A\| %
    &\leq 1 + C \tilde K\left(1 + \sqrt{\frac{N}{m}}\right)  %
    \\
    &\leq 1 + C\tilde K\left(1 + (\gamma - \theta)^{-1/2}\right) %
    = C_{\gamma, K, \theta},
  \end{align*}
  with probability at least $1 - C\exp(-m)$. Note, by re-choosing $N_{0}$ if
  necessary,
  \begin{align*}
    1 - C\exp(-m) %
    &\geq 1 - C\exp(-N(\gamma - \theta)) %
    \\
    & \geq 1 - C_{\gamma, \theta}\exp(-N_{0}) %
    \geq 1 - \varepsilon/2.
  \end{align*}
  In particular, for $N \geq N_{0}$, with probability at least $1 - \varepsilon$
  on the realization $A$, it holds with probability at least $k$ on $z$ that
  \begin{align*}
    \alpha_{2} \leq \|\tilde q\|_{2} \leq \|A\| \|\tilde x(\sqrt m)\|_{2} %
    \leq C_{\gamma, K, \theta}\|\tilde x (\sqrt m)\|_{2}.
  \end{align*}
  On the same event, by item 3 of \autoref{lem:geometric-lemma}, there is an
  absolute constant $p > 0$ so that
  $\alpha_{2} \geq C_{\delta, \gamma, L, \theta} N^{p}$, whence
  \begin{align*}
    \|\tilde x(\sqrt m)\|_{2} \geq C_{\delta, \gamma, K, L, \theta} N^{p}.
  \end{align*}
  Finally, this immediately implies that for $N \geq N_{0}$, with probability
  at least $1 - \varepsilon_{1}$ on the realization of $A$,
  \begin{align*}
    \tilde R(\sqrt m; 0, N, 1) %
    &\geq \E \left[\|\tilde x(\sqrt m)\|_{2}^{2}
    \mid \mathcal{E}\right] \mathbb{P}(\mathcal{E})
    \\
    &\geq C_{\delta, \gamma, K, L, \theta} k N^{p}.
  \end{align*}

\end{proof}

\begin{lem}[Lower bound $\tilde R(\sigma; 0, A, \eta)$, $\sigma < \eta \sqrt m$]
  \label{lem:bp-oc-zero}

  Fix $\delta, \varepsilon, \eta > 0$ and suppose $m : \nats \to \nats$
  satisfies $m(N)/N \to \gamma \in (0, 1)$. There is $N_{0} \in \nats$ and
  absolute constant $p > 0$ so that for all $N \geq N_{0}$, if
  $A \in \reals^{m(N) \times N}$ is a normalized $K$-subgaussian matrix, it
  holds with probability at least $1 - \varepsilon$ on the realization of $A$
  that for any $0 < \sigma \leq \eta \sqrt m$,
  \begin{align*}
    \tilde R(\sigma; 0, A, \eta) \geq C_{\delta, \gamma, K} N^{p}.
  \end{align*}

\end{lem}

\begin{proof}[Proof of {\autoref{lem:bp-oc-zero}}]
  The proof of this result is nearly identical to that of
  \autoref{lem:bp-oc-tuned-zero}. The crucial difference is its use of
  \autoref{prop:oc-solution-ordering}, using which one argues
  \begin{align*}
    \alpha_{2} %
    \leq \|\tilde q(\sqrt m) \|_{2} %
    \leq \|\tilde q(\sigma) \|_{2} %
    \leq C_{\gamma, K, \theta} \|\tilde x(\sigma)\|_{2}
  \end{align*}
  to show, in the appropriate regime, that
  $\|\tilde x(\sigma) \|_{2} \geq C_{\delta, \gamma, K, L, \theta} N^{p}$.
\end{proof}

\begin{thm}[Overconstrained maximin]
  \label{thm:bp-oc-maximin}
  Fix $\delta, \varepsilon, \eta > 0$ and suppose $m : \nats \to \nats$
  satisfies $m(N) / N \to \gamma \in (0, 1)$.  For any $s \geq 0$, there is an
  integer $N_{0} \in \nats$ and an absolute constant $p > 0$ so that for all
  $N \geq N_{0}$, if $A \in \reals^{m(N)\times N}$ is a normalized
  $K$-subgaussian matrix, then it holds with probability at least
  $1 - \varepsilon$ on the realization of $A$ that
  \begin{align*}
     \sup_{x \in \Sigma_{s}^{N}} \inf_{\sigma \leq \eta \sqrt m}
    \tilde R(\sigma; x, A, \eta) %
    \geq C_{\delta, \gamma, K, \theta} N^{p}.
  \end{align*}
\end{thm}

\begin{proof}[Proof of {\autoref{thm:bp-oc-maximin}}]
  By a scaling argument, it suffices to consider the case $\eta =
  1$. Establishing an admissible ensemble and using \autoref{lem:bp-oc-zero},
  there is $N_{0} \geq N_{*}$ such that for any $N \geq N_{0}$, with probability
  at least $1 - \varepsilon$ on $A$,
  \begin{align*}
     \sup_{x \in \Sigma_{s}^{N}} \inf_{\sigma \leq \sqrt m}
    \tilde R(\sigma; x, A, 1) %
    &\geq \inf_{\sigma \leq \sqrt m}
    \tilde R(\sigma; 0, A, 1) %
    \\
    & \geq C_{\delta, \gamma, K, \theta} N^{p}.
  \end{align*}
\end{proof}

\subsubsection{Suboptimal regime for basis pursuit}
\label{sec:subopt-regime-basis}

This section contains the proof for \autoref{thm:bp-minimax-suboptimality}, the
main result of \autoref{sec:analysis-bp} establishing a regime in which {\bp} is
minimax suboptimal. In essence, it combines \autoref{lem:uc-bp-subgaus} and
\autoref{lem:bp-oc-zero}.

\begin{proof}[Proof of {\autoref{thm:bp-minimax-suboptimality}}]

  By a scaling argument, it suffices to consider the case $\eta = 1$. Re-write
  the minimax expression \eqref{eq:bp-minimax-suboptimality} as
  \begin{align*}
     \inf_{\sigma > 0} \sup_{x \in \Sigma_{s}^{N}} %
    \tilde R(\sigma; x, A, 1) %
    &= \min \left\{ %
    \inf_{\sigma \leq \sqrt m} S(\sigma), %
      \inf_{\sigma > \sqrt m} S(\sigma) \right\},
    \\
    S(\sigma) &:= \sup_{x \in \Sigma_{s}^{N}} \tilde R(\sigma; x, A, 1).
  \end{align*}
  For any $N \geq N_{*}$, by \autoref{lem:uc-bp-subgaus}, it holds with
  probability at least $1 - \varepsilon/2$ on the realization of $A$ that
  \begin{align*}
    \inf_{\sigma > \sqrt m} S(\sigma) \geq C_{\delta, \gamma, \theta} \sqrt N.
  \end{align*}
  Next observe that the trivial lower bound
  $S(\sigma) \geq \tilde R(\sigma; 0, A, 1)$ holds for any $\sigma >0$, because
  $0 \in \Sigma_{s}^{N}$. In particular, \autoref{lem:bp-oc-zero} yields an
  $N_{0} \geq N_{*}$ and absolute constant $p > 0$ such that, with probability
  at least $1 - \varepsilon/2$ on the realization of $A$,
  \begin{align*}
    \inf_{\sigma \leq \sqrt m} S(\sigma) %
    \geq \inf_{\sigma \leq \sqrt m} \tilde R(\sigma; 0, A, 1) %
    \geq C_{\delta,\gamma,K,\theta} N^{p}.
  \end{align*}
  Consequently, there is an absolute constant $p > 0$ so that, for all
  $N \geq N_{0}$, it holds with probability at least $1 - \varepsilon$ on the
  realization of $A$ that
  \begin{align*}
    \inf_{\sigma > 0} \sup_{x \in \Sigma_{s}^{N}}
    \tilde R (\sigma; x, A, 1) %
    &\geq \min \left\{ C_{\delta, \gamma, \theta} \sqrt N,
    C_{\delta, \gamma, K, \theta} N^{p} \right\} %
    \\
    & \geq C_{\delta, \gamma, K, \theta} N^{p}.
  \end{align*}

\end{proof}

\section{Conclusion}
\label{sec:conclusion}

This work examined the relative sensitivity of three \textsc{Lasso} programs to
their governing parameters: {\ls}, {\bp} and {\qp}. We proved asymptotic
cusp-like behaviour of $\hat R(\tau; x_{0}, A, \eta)$ in the limiting low-noise
regime in \autoref{sec:LS-instability}. Numerical simulations in
\autoref{sec:numerical-results} support these observations for even modest
dimensional parameters and noise scales.

In \autoref{sec:qp}, we recall a result establishing right-sided stability of
{\qp} for a class of matrices that satisfy a version of RIP.\@ The result does
not address sensitivity of {\qp} to its governing parameter when the governing
parameter is less than its optimal value. In \autoref{sec:numerical-results}, we
demonstrate numerically that there are regimes in which {\qp} is sensitive to
its governing parameter $\lambda$ when $\lambda < \lambda^{*}$. This sensitivity
is readily observed in the rightmost plot of \autoref{fig:qp-instability}. This
observation establishes a numerical connection to the numerics and theory
of~\cite{berk2019pdparmsens, berk2020sl1mpc}, in which the authors analyze the
proximal denoising setting. Moreover, we observe that {\qp} is more sensitive to
its choice of parameter when the aspect ratio is larger. We believe this is due
to there being a smaller null-space, which has the effect of shrinking the space
of possible solutions. This behaviour is visible in both plots of
\autoref{fig:qp-instability}: the error curves for larger $\delta$ are steeper
for $\lambda < \lambda^{*}$ than those for smaller values of $\delta$.

In \autoref{sec:analysis-bp}, we proved asymptotic suboptimality of
$\tilde R(\sigma; x_{0}, A, \eta)$ in a certain dimensional regime that falls
outside the typical CS regime where $m \approx Cs\log(N/s)$. In particular, for
$m \approx \delta N$, $\delta \in (0, 1)$, we show that {\bp} risk is
asymptotically suboptimal for ``very sparse'' signals. We demonstrate that this
theory is relevant to the CS regime in \autoref{sec:numerical-results}, in which
we show that the loss and average loss for {\bp} are sensitive to the value of
the governing parameter if the ground-truth signal is very sparse. Furthermore,
\autoref{fig:bp-numerics-1} and \autoref{fig:bp-numerics-2} depict suboptimality
of the {\bp} risk for modest choices of dimesional parameters.

Future works include extending the main results to the generalized
\textsc{Lasso} setting, using more general atomic norms. A rigorous examination
of low-rank matrix recovery could be interesting. Finally, it would be useful to
understand when a convex program is expected to exhibit sensitivity to its
governing parameter, and to determine systematically the regime in which that
instability arises.

\section{Acknowledgments}
\label{sec:acknowledgments}

We would like to thank Xiaowei Li for a careful reading of portions of the
manuscript.

\appendix
\section{Appendix}
\label{sec:appendix}

\subsection{Appendix I}
\label{sec:appendix-i}

\subsubsection{Proofs for refinements on bounds for gw}
\label{sec:proofs-refin-bounds}

\begin{proof}[Proof of {\autoref{coro:bellec-random-hulls}}]
  Assuming
  \begin{align*}
    m > C_{\varepsilon} \delta^{-2} \tilde K^{2} s \log\frac{2N}{s},
  \end{align*}
  \autoref{lem:corollary-lmpv-thm1.4} gives
  \begin{align*}
    \sup_{x \in \mathcal{L}_{s}(1)} \left|\|Ax\|_{2} - \|x\|_{2}\right| < \delta,
  \end{align*}
  with probability at least $1 - \varepsilon$ on the realization of $A$. In
  particular, $(1- \delta)^{-1}\|A^{j}\|_{2} \geq 1$ and
  $(1 + \delta)^{-1} \|A^{j}\|_{2} \leq 1$ for all $j \in [N]$. Define the sets
  \begin{align*}
    T_{+} &:= \cvx \{ \pm (1+\delta)^{-1}A^{j} : j \in [N]\}, %
    \\
      T_{-} &:= \cvx \{ \pm (1-\delta)^{-1}A^{j} : j \in [N]\}.
  \end{align*}
  We will apply \autoref{prop:bellec1} and \autoref{prop:bellec2} to $T_{+}$
  and $T_{-}$, respectively, then use that $T_{+} = (1 + \delta)^{-1} T$ and
  $T_{-} = (1 - \delta)^{-1} T$. Indeed, observe that
  \begin{align*}
    (1 + \delta)^{-1} & \w (T \cap (1 + \delta)\gamma B_{2}^{m}) %
    \\
    & = \w ((1+ \delta)^{-1} T \cap \gamma B_{2}^{m}) %
    = \w (T_{+} \cap \gamma B_{2}^{m}) %
    \\
    &\leq \min \big\{ %
      4 \sqrt{ \max \big\{1, \log(8e N \gamma^{2}) \big\} }, %
      \gamma \sqrt{ \min\{m, 2N\}} \big\}.
  \end{align*}
  Rearranging, with $\alpha = (1 + \delta)\gamma$ gives
  \begin{align*}
    \w (T \cap \alpha B_{2}^{m}) %
    & \leq  \min \big\{ %
    4 (1 + \delta) \sqrt{ \max \big\{1,
    \log(8e N (1 + \delta)^{-2}\alpha^{2}) \big\} }, %
    \alpha \sqrt{ \min\{m, 2N\}} \big\}.
  \end{align*}
  Similarly, one may derive the lower bound for $\alpha \in (0, (1-\delta))$,
  using that $\kappa = 1 - \delta$,
  \begin{align*}
    \w (T \cap \alpha B_{2}^{m}) %
    \geq (\sqrt 2 / 4) (1 - \delta)^{2} \sqrt{
    \log \frac{N\alpha^{2}}{5(1-\delta)^{2}}}. %
  \end{align*}
\end{proof}

\subsubsection{Proofs for projection lemma}
\label{sec:proofs-proj-lemma}

\begin{proof}[Proof of {\autoref{cor:projection-lemma}}]
  Define $\beta := \|q_{\alpha}\|_{\mathcal{K}}$. Then
  $q_{\alpha} \in \beta \mathcal{K}$ and so
  $\|y - \mathrm{P}_{\beta \mathcal{K}}(y)\|_{2} \leq \alpha$ by definition of
  $\mathrm{P}_{\beta \mathcal{K}}(\cdot)$. Again by definition of
  $\mathrm{P}_{\beta \mathcal{K}}(\cdot)$, it holds that
  $\|\mathrm{P}_{\beta \mathcal{K}}(y)\|_{\mathcal{K}} \leq \beta$. In particular,
  $\mathrm{P}_{\beta \mathcal{K}}(y)$ is feasible and
  \begin{align*}
    \|\mathrm{P}_{\beta \mathcal{K}}(y)\|_{\mathcal{K}} \leq \|q_{\alpha}\|_{\mathcal{K}},
  \end{align*}
  whence optimality of $q_{\alpha}$ implies
  $q_{\alpha} = \mathrm{P}_{\beta \mathcal{K}}(y)$. Thus, by an elementary sequence of
  steps, the proof follows from the projection lemma.
\end{proof}

\subsubsection{Parameter instability of nuclear norm recovery}
\label{sec:param-inst-nucl}

This section is a supplement to \autoref{sec:subopt-choice-tau}, in support of
the comment made at the end of \autoref{sec:LS-instability}. We include a result
that ports the two lemmas of \autoref{sec:subopt-choice-tau} from the setting of
constrained \textsc{Lasso} to that of constrained nuclear norm recovery. Before
we state the lemma, define
$B_{*} := \{ X\in \reals^{d\times d} : \|X\|_{*} = \sum_{i = 1}^{d}
\sigma_{i}(X) \leq r \}$, where $\sigma_{i}(X)$ is the $i$th largest singular
value of the matrix $X$. Recall that $\|\cdot\|_{*}$ is dual to the operator
norm $\|\cdot\|$ and so they admit the following inequality for any
$X, Y \in \reals^{d\times d}$
\begin{align*}
  \ip{X, Y} \leq \|X\| \|Y\|_{*}. 
\end{align*}
Finally, since the space of matrices is finite, all matrix norms are
equivalent. In particular, if $X$ is a rank $r$ matrix (\ie $\sigma_{i}(X) = 0$
for $i = r+1, \ldots, d$), then
\begin{align*}
  \|X\|_{F} \leq \|X\|_{*} \leq \sqrt r \|X\|_{F}.
\end{align*}

\begin{lem}[Nuclear Norm Recovery]
  \label{lem:nuc-norm-recovery}
  Let $\mathcal{A} : \reals^{d \times d} \to \reals^{m}$ be an operator mapping
  $\mathcal{A} X = (\ip{A_{i}, X})_{i=1}^{m}$ for
  $A_{i} \in \reals^{d\times d}$. Given $X_{0} \in \reals^{d\times d}$,
  $\eta > 0$ and $z \in \reals^{m}$ with $z_{i} \iid \mathcal{N}(0, 1)$, let
  $y = \mathcal{A}X_{0} + \eta z$. Suppose that either $\tau > \|X_{0}\|_{*}$
  and $\dim \ker \mathcal{A} > 0$ or $\tau < \|X_{0}\|_{*}$. Almost surely on
  the realization of $z$,
  \begin{align*}
    \lim_{\eta \to 0} \hat L ( \tau; X_{0}, \mathcal{A}, \eta z) %
    &= \lim_{\eta \to 0} \eta^{-2} \|\hat X(\tau; y, \mathcal{A}, B_{*}) - X_{0} \|_{F}^{2}
    \\
    &= \infty.
  \end{align*}

\end{lem}

\begin{proof}[Proof of {\autoref{lem:nuc-norm-recovery}}]
  First, we examine the setting where $\tau > \|X_{0}\|_{*}$.  In like manner as
  the proof for \autoref{lem:ls-instability-uc}, define
  $\tau^{*} := \|X_{0}\|_{*}$ and set $\rho := \tau - \tau^{*}$. Assume
  $\Span (\mathcal{A}) = \reals^{m}$. There exists
  $\zeta \in \reals^{d\times d}$ such that $\mathcal{A} \zeta = z$, and so
  $\mathcal{A}(X_{0} + \eta \zeta) = y$. If $\eta$ is sufficiently small then
  $X' := X_{0} + \eta \zeta \in \tau B_{*}$ where
  $B_{*} = \{ X \in \reals^{d\times d} : \|X\|_{*} \leq 1\}$ is the nuclear norm
  ball. In particular, $X'$ is a solution for {\gls} where
  $\mathcal{K} = B_{*}$. Applying a similar argument as before, noting that
  \begin{align*}
    \|X' - X_{0}\|_{F} \geq \frac{\|X'\|_{*} - \|X_{0}\|_{*}}{d} \geq \rho > 0,
  \end{align*}
  will complete the proof in the case where $\mathcal{A}$ spans
  $\reals^{m}$. The case where $\Span\mathcal{A} \subsetneq \reals^{m}$ is
  similar.

  Next, we examine the setting where $\tau < \|X_{0}\|_{*}$. For any solution
  $\Xi$ to {\gls}, one has by norm equivalence and triangle inequality:
  \begin{align*}
    \hat L(\tau; X_{0}, \mathcal{A}, B_{*}) %
    \geq \eta^{-2} \|X_{0} - \Xi\|_{F}^{2} %
    &\geq \frac{(\|X_{0}\|_{*} - \|\Xi\|_{*})^{2}}{\eta^{2} d^{2}} %
    \\
    &\xrightarrow{\eta \to 0} \infty.
  \end{align*}
\end{proof}

Note that risk bounds are well-known for constrained nuclear norm recovery in
the case where $\tau = \|X_{0}\|_{*}$ and
$\mathcal{A} : X \mapsto (\ip{A_{i}, X})_{i=1}^{m}$ with
$A_{i} \in \reals^{d\times d}$ independent and having independent subgaussian
entries~\cite{candes2010matrix, candes2011tight}. In particular, combining such
a result with the lemmas above gives an analogue to \autoref{thm:ls-instability}
in the constrained nuclear norm setting.

\subsubsection{Technical lemmas for overconstrained basis pursuit}
\label{sec:techn-lemm-supp}

\begin{prop}[Lower bound $\w (K_{1})$]
  \label{prop:gw-lb-1}
  Fix $C_{1}, \delta, \varepsilon > 0$ and
  $0 < \theta < \min\{1-\gamma, \gamma\}$. Given an admissible ensemble, there
  exists a choice of absolute constants $a_{1} > 0$ and $L > 1$, as well as an
  integer $N_{0}^{\eqref{prop:gw-lb-1}} \geq N_{*}$ so that, for each
  $N \geq N_{0}^{\eqref{prop:gw-lb-1}}$, it holds with probability at least
  $1 - \varepsilon$ on the realization of $A$ that
  \begin{align*}
    \w (K_{1}) \geq \left(\frac{a_{1}^{2} + C_{1}}{2}\right) \sqrt m.
  \end{align*}
\end{prop}

\begin{proof}[Proof of \autoref{prop:gw-lb-1}]
  Since $\w (K_{1}) = \E \sup_{q \in K_{1}} \ip{q, z}$ is the gw of
  $K_{1}$ we may invoke \autoref{coro:bellec-random-hulls} to obtain a
  sufficient chain of inequalities:
  \begin{align}
    \label{eq:gw-lb-1-1}
    \w (K_{1}) %
    & \overset{\eqref{coro:bellec-random-hulls}}{\geq}
    \frac{\sqrt{2}}{4} (1 - \delta)^{2} \lambda
    \sqrt{\log \left(\frac{N \alpha_{1}^{2}}{5(1-\delta)^{2}\lambda^{2}}\right)}
    \\
    & \overset{(*)}{\geq} \left(\frac{a_{1}^{2} + C_{1}}{2}\right)\sqrt m.
  \end{align}
  The first inequality holds with probability at least $1 - \varepsilon$ on the
  realization of $A$. Therefore, showing that $(*)$ holds implies the desired
  result. Rewriting $(*)$ gives the equivalent condition
  \begin{align*}
    C_{\delta, L} \sqrt{\frac{m}{\log N}} \sqrt{
    \log\left(\frac{N \log N}{C_{\delta, L, a_{1}} \sqrt m}\right)}
    \geq C_{a_{1}, C_{1}} \sqrt m
  \end{align*}
  The latter term of the left-hand side may be simplified using that
  $m \leq N(\gamma + \theta)$, since $N \geq N_{\theta}$:
  \begin{align*}
    \log\left(\frac{N \log N}{C_{\delta, L, a_{1}} \sqrt m}\right) %
    \geq \log \left( C_{\delta, L, a_{1}, \gamma, \theta} \sqrt N \log N\right).
  \end{align*}
  Thus, $(*)$ is satisfied if
  \begin{align*}
    \log \left( C_{\delta, L, a_{1}, \gamma, \theta} \sqrt N \log N\right)
    \geq C_{\delta, L, a_{1}, C_{1}} \log N,
  \end{align*}
  which holds when
  \begin{align*}
    \left(\frac{1}{2} - C_{\delta, L, a_{1}, C_{1}}\right) \log \left( N\right)
    \geq - \log \left( C_{\delta, L, a_{1}, \gamma, \theta} \log N\right).
  \end{align*}
  This is eventually true so long as one chooses $(a_{1}, L)$ abiding
  \begin{align*}
    C_{\delta, L, a_{1}, C_{1}} %
    = 2 \left(\frac{a_{1}^{2} + C_{1}}{(1-\delta)^{2} L}\right)^{2} %
    < \frac{1}{2}.
  \end{align*}
\end{proof}

\begin{prop}[Lower bound $X_{1}$]
  \label{prop:gw-lb-2}

  Fix $\delta, \varepsilon_{1}, \varepsilon_{2} > 0$ and
  $\theta \in (0, \gamma)$. Given an admissible ensemble, there exists a choice
  of absolute constants $a_{1} > 0$ and $L > 1$, as well as an integer
  $N_{0}^{\eqref{prop:gw-lb-2}} \geq N_{*}$ so that, for each
  $N \geq N_{0}^{\eqref{prop:gw-lb-2}}$, it holds with probability at least
  $1 - \varepsilon_{1}$ on the realization of $A$ that with probability at
  least $1 - \varepsilon_{2}$ on the realization of $z$, for any $c \in (0, 1)$
  there exists $q \in K_{1}$ satisfying $\ip{q, z} \geq c \w (K_{1})$.
\end{prop}

\begin{proof}[Proof of \autoref{prop:gw-lb-2}]
  Observe that $K_{1} \subseteq \reals^{m}$ is a topological space and define
  the centered Gaussian process $T_{x} := \ip{x, g}$ for
  $g_{i} \iid \mathcal{N}(0, 1)$. Observe that
  $X_{1} = \sup_{x \in K_{1}} |T_{x}|$ is almost surely finite. So, for any
  $u > 0$,
  \begin{align*}
    \mathbb{P} \left( X_{1} < \w (K_{1}) - u\right) %
    \leq \exp\left(- \frac{u^{2}}{2 \sigma_{K_{1}}^{2}}\right)
  \end{align*}
  by \autoref{thm:borell-tis}, where
  \begin{align*}
    \sigma_{K_{1}}^{2} %
    & = \sup_{x \in K_{1}} \E T_{x}^{2} %
    = \sup_{x \in K_{1}} \sum_{i = 1}^{N} x_{i}^{2} \E|g_{i}|^{2} %
    \\
    & = \sup_{x\in K_{1}} \|x\|_{2}^{2} %
    = \alpha_{1}^{2} %
    = a_{1}^{2} \sqrt m.
  \end{align*}
  Now, combine \autoref{thm:borell-tis} and \autoref{coro:bellec-random-hulls}.
  For $N \geq N_{*}$, it holds with probability at least $1 - \varepsilon_{1}$ on
  the realization of $A$ that for any $c \in (0, 1)$,
  \begin{align*}
    \mathbb{P}& \left(X_{1} < c \w (K_{1})\right) %
    \leq \exp\left(- \frac{(1-c)^{2} \w^{2}(K_{1})}{2 \sigma_{K_{1}}^{2}}\right) %
    \\
    & \leq \exp\left(
      - C_{a_{1}, c, \delta, L} \sqrt{m} \cdot \frac{
      \log\left(C_{a_{1}, \delta, \gamma, L, \theta} \sqrt N \log N\right)}{
      \log N}\right). %
  \end{align*}
  Choose $N_{1} \geq N_{*}$ so that the following chain of inequalities is satisfied:
  \begin{align*}
    \frac{
    \log\left(C_{a_{1}, \delta, \gamma, L, \theta} \sqrt N \log N\right)
    }{\log N} %
    & = \frac{1}{2} + \frac{
    \log\left(C_{a_{1}, \delta, \gamma, L, \theta} \log N\right)
    }{\log N} %
    \\
    & \geq \frac{1}{4}.
  \end{align*}
  Further, select $N^{\eqref{prop:gw-lb-2}} \geq N_{1}$ such that all
  $N \geq N^{\eqref{prop:gw-lb-2}}$ satisfy
  \begin{align*}
    N \geq \frac{\log^{2}\varepsilon_{2}^{-1}}{C_{a_{1}, c, \delta, \gamma, L, \theta}}.
  \end{align*}
  Then, for all $N \geq N^{\eqref{prop:gw-lb-2}}$, it holds with probability
  at least $1 - \varepsilon_{1}$ on the realization of $A$ that for any
  $c \in (0, 1)$,
  \begin{align*}
    \mathbb{P}\left(X_{1} < c \w (K_{1})\right) %
    \leq \exp\left(
    - C_{a_{1}, c, \delta, L} \sqrt{m} \right) %
    < \varepsilon_{2}.
  \end{align*}
  Thus, under the specified conditions, $X_{1} \geq c \w (K_{1})$ with
  probability at least $1 - \varepsilon_{1}$ on the realization of $A$ and
  probability at least $1 - \varepsilon_{2}$ on the realization of $z$. In
  particular, since $K_{1}$ is closed, it holds with probability at least
  $(1 - \varepsilon_{1})(1 - \varepsilon_{2})$ that there exists $q \in K_{1}$
  realizing the supremum, thereby admitting existence of a $q$ as claimed.
\end{proof}

\begin{prop}[Control $K_{1}\cap F$]
  \label{prop:gw-lb-3}

  Fix $C_{1}, \delta, \varepsilon_{1}, \varepsilon_{2} > 0$ and
  $\theta \in (0, \gamma)$. Given an admissible ensemble, there is an integer
  $N_{0}^{\eqref{prop:gw-lb-3}} \geq N_{*}$ and an absolute constant
  $k_{1} = k_{1}(N_{0}^{\eqref{prop:gw-lb-3}}, C_{1}, \varepsilon_{2}) > 0$ so
  that for all $N \geq N_{0}^{\eqref{prop:gw-lb-3}}$, with probability at
  least $1 - \varepsilon_{1}$ on the realization of $A$, there is an event
  $\mathcal{E}$ for $z$ satisfying
  \begin{align*}
    K_{1} \cap F \neq \emptyset \quad\text{on}\quad \mathcal{E}, %
    \qquad \text{and} \qquad %
    \mathbb{P}(\mathcal{E}) \geq k_{1}.
  \end{align*}

\end{prop}

\begin{proof}[Proof of {\autoref{prop:gw-lb-3}}]

  Fix $c_{1} \in (0, 1)$. By \autoref{prop:gw-lb-2}, there is a choice of
  $a_{1} > 0$ and $L > 1$, and an integer $N_{0} \geq N_{*}$ such that, with
  probability at least $1 - \varepsilon_{1}/2$ on the realization of $A$, there
  is an event $\mathcal{E}_{1}$ for $z$, with
  $\mathbb{P}(\mathcal{E}_{1}) \geq 1 - \varepsilon_{2}$, on which
  \begin{align*}
    \sup_{q \in K_{1}} \ip{q, z} \geq c_{1} \w (K_{1}).
  \end{align*}
  Further, there exists $q \in K_{1}$ realizing that supremum because $K_{1}$
  is closed. Selecting this $q$, we have ${\ip{q, z} \geq c_{1}
    \w (K_{1})}$. Next, define
  $C_{1}' := c_{1}^{-1}(a_{1}^{2} + C_{1}) - a_{1}^{2}$. By
  \autoref{prop:gw-lb-1}, increasing $L$ if necessary, there is an integer
  $N_{1} \geq N_{0}$ so that with probability at least $1 - \varepsilon_{1}/2$
  on the realization of $A$,
  \begin{align*}
    \w (K_{1}) %
    \geq \left(\frac{a_{1}^{2} + C_{1}'}{2}\right) \sqrt m.
  \end{align*}
  In particular, with probability at least $1 - \varepsilon_{1}$ on the
  realization of $A$ and probability at least $1 - \varepsilon_{2}$ on the
  realization of $z$, one has simultaneously:
  \begin{align*}
    \ip{q, z} %
    \geq c_{1} \w (K_{1}) %
    \geq \left(\frac{a_{1}^{2} + C_{1}}{2}\right) \sqrt m.
  \end{align*}
  Define the event
  $\mathcal{Z}_{<} := \{ \|z\|_{2}^{2} \leq m + C_{1} \sqrt m\}$. Because
  $q \in K_{1}$, $\|q\|_{2} \leq a_{1} m^{-1/4}$, whence conditioning on
  $\mathcal{E}_{1} \cap \mathcal{Z}_{<}$ gives $q \in F$. Indeed,
  \begin{align*}
    \|q - z\|_{2}^{2} %
    & = \|q\|_{2}^{2} - 2\ip{q, z} + \|z\|_{2}^{2} %
    \\
    & \leq a_{1}^{2} \sqrt m - (a_{1}^{2} + C_{1}) \sqrt m + m + C_{1} \sqrt m %
    \\
    & = m.
  \end{align*}
  Choose $N_{0}^{\eqref{prop:gw-lb-3}} := N_{1}$. Then, for each
  $N \geq N_{0}^{\eqref{prop:gw-lb-3}}$, with probability at least
  $1- \varepsilon_{1}$ on the realization of $A$, there is an event
  $\mathcal{E} := \mathcal{E}_{1} \cap \mathcal{Z}_{<}$ on which
  $q \in K_{1} \cap F$. Next, define
  \begin{align*}
    k_{1} := k_{1}(N_{0}^{\eqref{prop:gw-lb-3}}, C_{1}, \varepsilon_{2}) %
    := \left[\inf_{N \geq N_{0}^{\eqref{prop:gw-lb-3}}}
    \mathbb{P}(\mathcal{Z}_{<})\right] - \varepsilon_{2}
  \end{align*}
  and observe that $k_{1} > 0$, because $\mathbb{P}(\mathcal{Z}_{<})$ is
  bounded below by a dimension independent constant for $N \geq 2$. Finally,
  because $\mathcal{E}_{1}$ holds with probability at least
  $1 - \varepsilon_{2}$ on the realization of $z$ one has
  \begin{align*}
    \mathbb{P}(\mathcal{E}) %
    \geq \mathbb{P}(\mathcal{Z}_{<}) - \varepsilon_{2} %
    \geq k_{1}.
  \end{align*}
\end{proof}

\begin{prop}[Upper bound $\w (K_{2})$]
  \label{prop:gw-ub-1}

  Fix $C_{2}, \delta, \varepsilon > 0$ and $\theta \in (0, \gamma)$. Given an
  admissible ensemble, there exists a choice of absolute constants $a_{1} > 0$
  and $L > 1$, an integer $N_{0}^{\eqref{prop:gw-ub-1}} \geq N_{*}$, and a
  maximal choice of $\alpha_{2} = \alpha_{2}(N)$ so that, for each
  $N \geq N_{0}^{\eqref{prop:gw-ub-1}}$, it holds with probability at least
  $1 - \varepsilon$ on the realization of $A$ that
  \begin{align*}
    \w (K_{2}) \leq \frac{C_{2}}{2} \sqrt m.
  \end{align*}

\end{prop}

\begin{proof}[Proof of \autoref{prop:gw-ub-1}]
  First, invoke \autoref{coro:bellec-random-hulls} in obtaining a sufficient
  chain of inequalities on $\w (K_{2})$:
  \begin{align*}
    \w (K_{2}) %
    \overset{\eqref{coro:bellec-random-hulls}}{\leq} %
    4 (1 + \delta) \lambda \sqrt{\log \left(
    \frac{8eN\alpha_{2}^{2}}{(1 + \delta)^{2}\lambda^{2}}\right)} %
    \overset{(**)}{\leq} %
    \frac{C_{2}}{2} \sqrt m.
  \end{align*}
  The first inequality holds with probability at least $1 - \varepsilon$ on the
  realization of $A$. Therefore, showing $(**)$ implies the desired
  result. Rewriting $(**)$ gives the equivalent condition
  \begin{align*}
    \log \left(C_{\delta, L} \alpha_{2}^{2} \frac{N\log N}{m}\right) %
    \leq C_{C_{2}, \delta, L} \log N.
  \end{align*}
  The left-hand side may be simplified using that $m \geq N(\gamma - \theta)$,
  since $N \geq N_{\theta}$, yielding a new sufficient condition:
  \begin{align}
    \label{eq:gw-ub-1-1}
    \log \left( C_{\delta,\gamma,L,\theta} \alpha_{2}^{2} \log N\right) %
    \leq C_{C_{2}, \delta, L} \log N.
  \end{align}
  Thus, \eqref{eq:gw-ub-1-1} is valid for any $\alpha_{2}$ satisfying
  $\alpha_{2} \leq \alpha_{2}(N)$, where
  \begin{align*}
    \alpha_{2}^{2}(N) &:= C_{\delta,\gamma,L,\theta} \frac{N^{d}}{\log N}, %
                        \quad \text{where} %
    \\ 
    d &:= \left(\frac{C_{2}}{8 (1 + \delta) L}\right)^{2}, %
    \\
    C_{\delta, \gamma, L, \theta} &:= \frac{(1+\delta)^{2} L^{2} (\gamma - \theta)}{8e}.
  \end{align*}
  Finally, set $N_{0}^{\eqref{prop:gw-ub-1}} := N_{*}$ and observe that for
  any $N \geq N_{0}^{\eqref{prop:gw-ub-1}}$, with
  $\alpha_{2} := \alpha_{2}(N)$, it holds with probability at least
  $1 - \varepsilon$ on the realization of $A$ that
  $\w (K_{2}) \leq \frac{C_{2}}{2} \sqrt m$, as desired.
\end{proof}

\begin{rmk}
  It will be convenient to reselect $N_{0}^{\eqref{prop:gw-ub-1}} \geq N_{*}$
  in \autoref{prop:gw-ub-1} so that $\alpha_{2}(N)$ is increasing for all
  $N \geq N_{0}^{\eqref{prop:gw-ub-1}}$. A quick calculation verifies that
  $N_{0}^{\eqref{prop:gw-ub-1}} \geq \exp(d^{-1})$ suffices.
\end{rmk}

\begin{prop}[Upper bound $X_{2}$]
  \label{prop:gw-ub-2}

  Fix $\delta, \varepsilon_{1}, \varepsilon_{2} > 0$ and
  $0 < \theta < \min \{1 - \gamma, \gamma\}$. Given an admissible ensemble,
  there exists a choice of absolute constants $a_{1} > 0$ and $L > 1$, as well
  as an integer $N_{0}^{\eqref{prop:gw-ub-2}} \geq N_{*}$ so that, for each
  $N \geq N_{0}^{\eqref{prop:gw-ub-2}}$, it holds with probability at least
  $1 - \varepsilon_{1}$ on the realization of $A$ and with probability at least
  $1 - \varepsilon_{2}$ on the realization of $z$ that for any $C > 1$,
  \begin{align*}
    \sup_{q \in K_{2}} \ip{q, z} \leq C \w (K_{2}).
  \end{align*}

\end{prop}
\begin{proof}[Proof of {\autoref{prop:gw-ub-2}}]

  Define the centered Gaussian process $T_{x} := \ip{x, g}$ for
  $x \in K_{2} \subseteq \reals^{m}$, a topological space, where
  $g_{i} \iid \mathcal{N}(0, 1)$. Observe that
  $X_{2} = \sup_{x \in K_{2}} |T_{x}| < \infty$ almost surely. So, for any
  $u > 0$,
  \begin{align*}
    \mathbb{P}\left(
    X_{2} > \w (K_{2}) + u\right) %
    \leq \exp\left(- \frac{u^{2}}{2 \sigma_{K_{2}}^{2}}\right)
  \end{align*}
  by \autoref{thm:borell-tis}, where
  \begin{align*}
    \sigma_{K_{2}}^{2} %
    & = \sup_{x \in K_{2}}\E_{g}|\ip{x,g}|^{2} %
    = \sup_{x \in K_{2}}\sum_{i=1}^{m} x_{i} \E_{i}|g_{i}|^{2} %
    \\
    & = \sup_{x\in K_{2}} \|x\|_{2}^{2} %
    = \alpha_{2}^{2} \leq \alpha_{1}^{2} %
    = a_{1}^{2}\sqrt m.
  \end{align*}
  Now, invoke \autoref{coro:bellec-random-hulls}. For $N \geq N_{*}$, it holds
  with probability at least $1 - \varepsilon_{1}$ on the realization of $A$
  that for any $C > 1$,
  \begin{align*}
    \mathbb{P}&\left(X_{2} > C \w (K_{2})\right) %
    \leq \exp\left(- \frac{(C-1)^{2} \w^{2}(K_{2})}{2 \sigma_{K_{2}}^{2}}\right) %
    \\
    & \leq \exp\left(- \frac{(C-1)^{2} C_{\delta, L} \sqrt m
      \log(C_{\delta, \gamma, L, \theta} \alpha_{2}^{2} \log N) }{
      2a_{1}^{2} \log N}\right) %
    \\
    & = \exp\left(- C_{a_{1}, C, \delta, L}\frac{ \sqrt m
      \log(C_{\delta, \gamma, L, \theta} N^{d}) }{\log N}\right) %
    \\
    & = \exp\left(- C_{a_{1}, C, \delta, L}\sqrt m \left(
      d + \frac{C_{\delta, \gamma, L, \theta}}{\log N}\right)\right) %
    \\
    & \leq \exp\left(- C_{a_{1}, C, \delta, L} d \sqrt m\right) %
      < \varepsilon_{2}.
  \end{align*}
  The latter line follows by taking $N_{1} \geq N_{*}$ sufficiently large so
  that for all $N \geq N_{1}$,
  \begin{align*}
    \frac{C_{\delta, \gamma, L, \theta}}{\log N} > - \frac{d}{2}, %
    \qquad \text{and}\qquad %
    N %
    \geq \frac{\log^{2}\varepsilon_{2}^{-1}}{
    C_{a_{1}, C, \delta, L} d^{2} (\gamma - \theta)}.
  \end{align*}
  Finally, set $N_{0}^{\eqref{prop:gw-ub-2}} := N_{1}$. Then, for all
  $N \geq N_{0}^{\eqref{prop:gw-ub-2}}$, it holds with probability at least
  $1 - \varepsilon_{1}$ on the realization of $A$ that, for any $C > 1$,
  $X_{2} \leq C \w (K_{2})$ with probability at least $1 - \varepsilon_{2}$ on
  the realization of $z$, as desired.
\end{proof}

\begin{prop}[Control $\w (K_{2})\cap F$]
  \label{prop:gw-ub-3}

  Fix $C_{2}, \delta, \epsilon_{1}, \varepsilon_{2} > 0$ and
  $\theta \in (0, \gamma)$. Given an admissible ensemble, there is an integer
  $N_{0}^{\eqref{prop:gw-ub-3}} \geq N_{*}$ and an absolute constant
  $k_{2} = k_{2}(N_{0}^{\eqref{prop:gw-ub-3}}, C_{2}, \varepsilon_{2}) > 0$ so
  that for all $N \geq N_{0}^{\eqref{prop:gw-ub-3}}$, with probability at
  least $1 - \varepsilon_{1}$ on the realization of $A$, there is an event
  $\mathcal{E}$ for $z$ satisfying
  \begin{align*}
    K_{2} \cap F = \emptyset \quad \text{on} \quad \mathcal{E} %
    \qquad \text{and}\qquad %
    \mathbb{P}(\mathcal{E}) \geq k_{2}.
  \end{align*}
\end{prop}

\begin{proof}[Proof of {\autoref{prop:gw-ub-3}}]

  Fix $c_{2} > 1$. By \autoref{prop:gw-ub-2}, there is a choice of $a_{1} > 0$
  and $L > 1$, and an integer $N_{0} \geq N_{*}$ such that, with probability at
  least $1 - \varepsilon_{1} / 2$ on the realization of $A$, there is an event
  $\mathcal{E}_{2}$ for $z$, with
  $\mathbb{P}(\mathcal{E}_{2}) \geq 1 - \varepsilon_{2}$, on which
  \begin{align*}
    \sup_{q \in K_{2}} \ip{q, z} \leq c_{2} \w (K_{2}).
  \end{align*}
  Now, select $C_{2}' > 0$ so that $0 < c_{2} C_{2}' < C_{2}$. By
  \autoref{prop:gw-ub-1}, increasing $L$ if necessary, there is an integer
  $N_{1} \geq N_{0}$ so that with probability at least
  $1 - \varepsilon_{1} / 2$ on the realization of $A$,
  \begin{align*}
    \w (K_{2}) \leq \frac{C_{2}'}{2} \sqrt m.
  \end{align*}
  In particular, with probability at least $1 - \varepsilon_{1}$ on the
  realization of $A$ and probability at least $1 - \varepsilon_{2}$ on the
  realization of $z$, one has simultaneously:
  \begin{align*}
    \sup_{q \in K_{2}} \ip{q, z} %
    \leq c_{2} \w (K_{2}) %
    < c_{2} \frac{C_{2}'}{2} \sqrt m %
    <  \frac{C_{2}}{2} \sqrt m.
  \end{align*}
  Define the event
  $\mathcal{Z}_{>} := \{ \|z\|_{2}^{2} \geq m + C_{2} \sqrt m\}$. With
  probability at least $1 - \varepsilon_{1}$ on the realization of $A$,
  conditioning on $\mathcal{E}_{2} \cap \mathcal{Z}_{>}$ gives, for any
  $q \in K_{2}$,
  \begin{align*}
    \|q - z\|_{2}^{2} %
    & = \|q \|_{2}^{2} - 2 \ip{q, z} + \|z\|_{2}^{2} %
    \\
    & > \|q\|_{2}^{2} - C_{2} \sqrt m + m + C_{2} \sqrt m %
    \\
    & \geq m.
  \end{align*}
  In particular, $\|q - z\|_{2}^{2} > m$ for all $q \in K_{2}$, and so
  $K_{2} \cap F = \emptyset$. Choose $N_{0}^{\eqref{prop:gw-ub-3}} :=
  N_{1}$. Then, for each $N \geq N_{0}^{\eqref{prop:gw-ub-3}}$, with
  probability at least $1 - \varepsilon_{1}$ on the realization of $A$, there
  is an event $\mathcal{E} := \mathcal{E}_{2} \cap \mathcal{Z}_{>}$ on which
  $K_{2} \cap F = \emptyset$. Next, define
  \begin{align*}
    k_{2} %
    := k_{2}(N_{0}^{\eqref{prop:gw-ub-3}}, C_{2}, \varepsilon_{2}) %
    := \left[\inf_{N\geq N_{0}^{\eqref{prop:gw-ub-3}}}
    \mathbb{P}(\mathcal{Z}_{>})\right] - \varepsilon_{2}
  \end{align*}
  and observe that $k_{2} > 0$ because $\mathbb{P}(\mathcal{Z}_{>})$ is bounded
  below by a dimension independent constant for $N \geq 2$. Finally, because
  $\mathcal{E}_{2}$ holds with probability at least $1 - \varepsilon_{2}$ on
  the realization of $z$, one has
  \begin{align*}
    \mathbb{P}(\mathcal{E}) %
    \geq \mathcal{P}(\mathcal{Z}_{>}) - \varepsilon_{2} %
    \geq k_{2}.
  \end{align*}
\end{proof}

\begin{rmk}[Dimension-independent bound for $k_{1}, k_{2}$]
  In \autoref{prop:gw-lb-3} and \autoref{prop:gw-ub-3} above, it is not
  necessary to have $k_{1}$ depend on $N_{0}^{\eqref{prop:gw-lb-3}}$ or
  $k_{2}$ on $N_{0}^{\eqref{prop:gw-ub-3}}$. For example, one could bound
  $\mathbb{P}(\mathcal{Z}_{<})$ for all $N \geq 2$. However, the dependence of
  $k_{1}$ on $N_{0}^{\eqref{prop:gw-lb-3}}$ is simply to note that the lower
  bound on $\mathbb{P}(\mathcal{Z}_{<})$ is improved by considering
  $N \geq N_{0}^{\eqref{prop:gw-lb-3}}$ as opposed to merely $N \geq
  2$. Analogously so for $k_{2}$.
\end{rmk}

\subsection{Appendix II}
\label{sec:appendix-ii}

\subsubsection{RBF approximation of the average loss}
\label{sec:rbf-approximation}

For a particular program $\mathfrak P$, $A \in \reals^{m \times N}$,
$z_{i} \iid \mathcal{N}(0, 1)$, $\eta > 0$ and $x_{0} \in \Sigma_{s}^{N}$,
define $y^{(j)} = Ax_{0} + \eta \hat z^{(j)}, j \in [k]$ where $\hat z^{(j)}$ is
an iid copy of $z$. For the program $\mathfrak P$ denote the solution to the
program by $x^{*}(\upsilon; y^{(j)}, A)$, where $\upsilon > 0$ is the governing
parameter. Here, $\upsilon \in \Upsilon^{(j)}$, where $\Upsilon^{(j)}$ is a
logarithmically spaced grid of $n$ points whose centre is approximately equal to
the optimal parameter choice $\upsilon^{*}(x_{0}, A, \eta z^{(j)})$.  Next,
define the concatenated vectors
\begin{align*}
  \mathbf{z} := (z^{(j)} : j \in [k]) \in \reals^{kn},\quad
  \mathbf{y} := (y^{(j)} : j \in [k]) \in \reals^{kn},
\end{align*}
and define the collection
\begin{align*}
  \Gamma %
  & := \Gamma(x_{0}, A, \eta\mathbf{z}) %
  \\
  & := \{ (\upsilon_{ij}, \mathscr{L}(\upsilon_{ij}; x_{0}, A, \eta z^{(j)})) : %
  \\
  & \quad\quad \upsilon_{i,j} \in \Upsilon^{(j)}, i \in [n], j \in [k] \}.
\end{align*}
From here, we describe how to approximate the average loss and the normalized
parameter. Specifically, we construct the RBF approximator $L^{\dagger}$
satisfying $L^{\dagger}(\rho; \Gamma) \approx \bar L(\rho; x_{0}, A, \eta,
k)$. Define the multiquadric RBF kernel by
\begin{align*}
  \kappa(\upsilon, \upsilon') %
  := \sqrt{1 + \left(\frac{|\upsilon-\upsilon'|}{\varepsilon_{\text{rbf}}}\right)^{2}}, %
  \qquad \upsilon, \upsilon' > 0,
\end{align*}
and define the matrix $X \in \reals^{kn \times kn}$ by
\begin{align*}
  X_{ij} = \kappa(\upsilon_{i}, \upsilon_{j}), \quad %
  \upsilon_{i}, \upsilon_{j} \in \Gamma.
\end{align*}
For $\mu_{\text{rbf}} \geq 0$, the coefficients of the RBF approximator are
given by $\tilde w \in \reals^{kn}$ where $\tilde w$ solves
\begin{align*}
  \mathbf{y} = (X - \mu_{\text{rbf}}I_{kn})\tilde w
\end{align*}
with $I_{kn} \in \reals^{kn \times kn}$ being the identity matrix. To evaluate
the approximant at a set of points $\xi \in \reals^{n_{\text{rbf}}}$, one
simply computes
\begin{align*}
  \tilde y &= \mathscr{L}^{\dagger}(\xi; \Gamma(\mathbf{y}, A)) := \tilde X \tilde w
  \quad \text{where} \quad%
  \\
  \tilde X_{ij} &:= \kappa(\xi_{i}, \upsilon_{j}), i \in [n_{\text{rbf}}], j \in [kn]. 
\end{align*}
The optimal parameter choice for the approximator, $\upsilon^{\dagger} > 0$, is
given as
\begin{align*}
  \upsilon^{\dagger} \in \argmin_{\upsilon > 0} \mathscr{L}^{\dagger}(\upsilon; \Gamma),
\end{align*}
Finally, the normalized parameter $\rho$ is approximated as
$\rho \approx \upsilon / \upsilon^{\dagger}$ and the average loss thus
approximated by
\begin{align*}
  L^{\dagger} (\rho; \Gamma(x_{0}, A, \eta \mathbf{z})) %
  := \mathscr{L} (\rho \upsilon^{\dagger}; \Gamma(x_{0}, A, \eta \mathbf{z})). 
\end{align*}

\subsubsection{Interpolation parameter settings}
\label{sec:interp-param-sett}

The RBF interpolation parameter settings for each of the approximations of the
average loss pertaining to {\ls} parameter sensitivity numerics appearing
in~\autoref{sec:ls-numerics} can be found in~\autoref{tab:ls-numerics-1}. For
those pertaining to {\qp}, appearing in~\autoref{fig:qp-instability-2}
and~\ref{fig:synthetic-example-0}--\ref{fig:synthetic-example-2},
see~\autoref{tab:qp-numerics-2}. For those pertaining to {\bp}, appearing
in~\autoref{sec:bp-numerics}, see~\autoref{tab:rbf-parameter-settings}. The RBF
interpolation parameter settings for each of the average loss approximations
in~\autoref{sec:1d-wavel-compr} are given
in~\autoref{tab:lasso-realistic-1d}. For those pertaining
to~\autoref{fig:2d-wavelet-nnse-1}--\ref{fig:realistic-lasso-sslp-2}
of~\autoref{sec:2d-wavel-compr}, see~\autoref{tab:2d-wvlt-cs-params}.
  
\begin{table}[h]
  \centering
  \begin{tabular}{crr}
    \toprule
    program & $\epsilon_{\text{rbf}}$ & $\mu_{\text{rbf}}$ \\
    \midrule
    {\ls} & $10^{-5}$        & $0.1$ \\
    {\qp} & $3\cdot 10^{-2}$ & $0.5$ \\
    {\bp} & $3\cdot 10^{-2}$ & $0.5$ \\
    \bottomrule
  \end{tabular}
  \caption[interpolation settings]{Average loss interpolation parameter settings
    for {\ls} parameter instability numerics in \autoref{sec:ls-numerics}.
    $(s, N, m, \eta) = (1, 10^{5}, 2500, 2\cdot 10^{-3})$;
    $(n_{\text{rbf}}, \text{function}) = (301, \text{multiquadric})$.}
  \label{tab:ls-numerics-1}
\end{table}

\begin{table}[h]
  \centering
  \begin{tabular}{rrrrlrr}
\toprule
      $s$ & $N$    &  $m$   &  $\eta$ & program & $\epsilon_{\text{rbf}}$ & $\mu_{\text{rbf}}$ \\
  \midrule
  $1$ & $10000$ & $2500$ & $10^{-5}$ & {\ls} & $0.005$ & $1$ \\
  $1$ & $10000$ & $2500$ & $10^{-5}$ & {\qp} & $0.05$  & $1$ \\
  $1$ & $10000$ & $2500$ & $10^{-5}$ & {\bp} & $0.04$  & $0.9$ \\
  $1$ & $10000$ & $4500$ & $10^{-5}$ & {\ls} & $0.005$ & $1$ \\
  $1$ & $10000$ & $4500$ & $10^{-5}$ & {\qp} & $0.05$  & $1$ \\
  $1$ & $10000$ & $4500$ & $10^{-5}$ & {\bp} & $0.04$  & $0.9$ \\
\midrule
  $750$ & $10000$ & $4500$ & $0.1$ & {\ls} & $0.005$ & $1$ \\
  $750$ & $10000$ & $4500$ & $0.1$ & {\qp} & $0.05$  & $1$ \\
  $750$ & $10000$ & $4500$ & $0.1$ & {\bp} & $0.05$  & $0.5$ \\
\midrule
  $100$ & $10000$ & $2500$ & $100$ & {\ls} & $0.005$ & $1$ \\
  $100$ & $10000$ & $2500$ & $100$ & {\qp} & $0.05$  & $1$ \\
  $100$ & $10000$ & $2500$ & $100$ & {\bp} & $0.05$  & $0.5$ \\
  $100$ & $10000$ & $4500$ & $100$ & {\ls} & $0.005$ & $1$ \\
  $100$ & $10000$ & $4500$ & $100$ & {\qp} & $0.05$  & $1$ \\
  $100$ & $10000$ & $4500$ & $100$ & {\bp} & $0.05$  & $0.5$ \\

\bottomrule
\end{tabular}

\caption[interpolation settings]{Average loss interpolation parameter settings
  for {\qp} parameter instability numerics in \autoref{fig:qp-instability-2} and
  \autoref{fig:synthetic-example-0},~\ref{fig:synthetic-example-1}
  and~\ref{fig:synthetic-example-2}. $n_{\text{rbf}} = 501$, function $=$ multiquadric.}
\label{tab:qp-numerics-2}
\end{table}

  

\begin{table}[h]
  \centering
  \begin{tabular}{rrrrlrrr}
    \toprule
    $N$ & $m$ & $\eta$ & $\delta$ & program & $\epsilon_{\text{rbf}}$ & $\mu_{\text{rbf}}$ & $n_{\text{rbf}}$\\
    \midrule
    $4000$ &  $400$  & $1$   & $0.1$   & {\bp} & $.05$  & $1$ & $301$ \\
    $4000$ &  $400$  & $1$   & $0.1$   & {\ls} & $.001$ & $1$ & $301$ \\
    $4000$ &  $400$  & $1$   & $0.1$   & {\qp} & $.05$  & $1$ & $301$ \\
    $4000$ &  $1000$ & $1$   & $0.25$  & {\bp} & $.05$  & $1$ & $301$ \\
    $4000$ &  $1000$ & $1$   & $0.25$  & {\ls} & $.001$ & $1$ & $301$ \\
    $4000$ &  $1000$ & $1$   & $0.25$  & {\qp} & $.05$  & $1$ & $301$ \\
    $4000$ &  $1800$ & $1$   & $0.45$  & {\bp} & $.05$  & $1$ & $301$ \\
    $4000$ &  $1800$ & $1$   & $0.45$  & {\ls} & $.001$ & $1$ & $301$ \\
    $4000$ &  $1800$ & $1$   & $0.45$  & {\qp} & $.05$  & $1$ & $301$ \\
    $7000$ &  $700$  & $1$   & $0.1$   & {\bp} & $.05$  & $1$ & $301$ \\
    $7000$ &  $700$  & $1$   & $0.1$   & {\ls} & $.001$ & $1$ & $301$ \\
    $7000$ &  $700$  & $1$   & $0.1$   & {\qp} & $.05$  & $1$ & $301$ \\
    $7000$ &  $1750$ & $1$   & $0.25$  & {\bp} & $.05$  & $1$ & $301$ \\
    $7000$ &  $1750$ & $1$   & $0.25$  & {\ls} & $.001$ & $1$ & $301$ \\
    $7000$ &  $1750$ & $1$   & $0.25$  & {\qp} & $.05$  & $1$ & $301$ \\
    $7000$ &  $700$ & $100$ &  $0.1$   & {\ls} & $0.05$ & $1$ & $501$ \\
    $7000$ &  $700$ & $100$ &  $0.1$   & {\qp} & $0.05$ & $1$ & $501$ \\
    $7000$ &  $700$ & $100$ &  $0.1$   & {\bp} & $0.05$ & $1$ & $501$ \\
    $7000$ &  $1750$ & $100$ & $0.25$  & {\ls} & $0.05$ & $1$ & $501$ \\
    $7000$ &  $1750$ & $100$ & $0.25$  & {\qp} & $0.05$ & $1$ & $501$ \\
    $7000$ &  $1750$ & $100$ & $0.25$  & {\bp} & $0.05$ & $1$ & $501$ \\
    \bottomrule
  \end{tabular}
  \caption{The RBF interpolation parameter settings for each of the
    approximations of the average loss in \autoref{sec:bp-numerics}, including
    for \autoref{fig:bp-numerics-1}. $s=1$, function $=$ multiquadric.}
  \label{tab:rbf-parameter-settings}
\end{table}

\begin{table}[h]
  \centering
  \begin{tabular}{rlrr}
    \toprule
    program   & $\varepsilon_{\mathrm{rbf}}$ & $\mu_{\mathrm{rbf}}$ & $n_{\text{rbf}}$ \\
    \midrule
    {\ls} & $0.01$ & $1$ & $501$ \\
    {\qp} & $0.01$ & $1$ & $501$ \\
    {\bp} & $0.01$ & $1$ & $501$ \\
    \bottomrule
  \end{tabular}
  \caption{The RBF interpolation parameter settings for each of the average loss
    approximations in~\autoref{sec:1d-wavel-compr}. $(s, N, m, \eta) = (10, 4096, 1843, 50)$; function = multiquadric.}
\label{tab:lasso-realistic-1d}
\end{table}

\begin{table}[h]
  \centering
  \begin{tabular}{rrlrrr}
    \toprule
    $\eta$ & program & $\varepsilon_{\text{rbf}}$ & $\mu_{\mathrm{rbf}}$ & $n_{\text{rbf}}$ \\
    \midrule
    $0.01$ & {\ls} &     $0.001$  & $1$ & $501$ \\
    $0.01$ & {\qp} &     $0.05$   & $1$ & $501$ \\
    $0.01$ & {\bp} &     $0.05$   & $1$ & $501$ \\
    $0.5$  & {\ls} &     $0.05$   & $1$ & $501$ \\
    $0.5$  & {\qp} &     $0.05$   & $1$ & $501$ \\
    $0.5$  & {\bp} &     $0.05$   & $1$ & $501$ \\
    \bottomrule
  \end{tabular}
  \caption{The RBF interpolation parameter settings for each of the RBF
    interpolations
    in~\autoref{fig:2d-wavelet-nnse-1}--\ref{fig:realistic-lasso-sslp-2}
    of~\autoref{sec:2d-wavel-compr}. $(s, N, m) = (416, 6418, 2888)$; function
    $=$ multiquadric.}
  \label{tab:2d-wvlt-cs-params}
\end{table}

\bibliographystyle{IEEEtran}
\bibliography{lasso-parameter-instability}

\end{document}